\documentclass{LMCS}

\def\dOi{10(2:6)2014}
\lmcsheading%
{\dOi}
{1--52}
{}
{}
{Nov.~\phantom01, 2010}
{Jun.~\phantom03, 2014}
{}

\ACMCCS{[{\bf Theory of computation}]: Formal languages and automata
  theory---Formalisms---Rewrite systems}

\usepackage{hyperref}

\usepackage[english]{babel}

\usepackage{textcomp}           
\usepackage[T1]{fontenc}
\usepackage[utf8]{inputenc}

\usepackage{amssymb}
\usepackage{stmaryrd}
\usepackage{color}
\usepackage{tikz}
\usepackage{ifthen}
\usepackage{pifont}

\mathchardef\mhyphen="2D

\renewcommand\theta\vartheta
\renewcommand\epsilon\varepsilon
\renewcommand\phi\varphi

\newcommand\seq[1]{\langle #1 \rangle}
\newcommand\emptyseq{\seq{}}

\newcommand\oh\widehat
\newcommand\ul\underline
\newcommand\ol\overline
\newcommand\ot\widetilde

\newcommand\wsuc\oh

\newcommand\lab{\textsc{l}}

\newcommand\glbbot{\sqcap}
\newcommand\Lubbot{\mathop{\bigsqcup}}
\newcommand\Glbbot{\mathop{\bigsqcap}}
\newcommand\lebot{\le_\bot}
\newcommand\gebot{\ge_\bot}
\newcommand\lbot{<_\bot}

\newcommand\glb{\sqcap}
\newcommand\Lub{\bigsqcup}
\newcommand\Glb{\bigsqcap}

\newcommand\nat{{\mathbb N}}

\newcommand\real{{\mathbb R}}
\newcommand\realp{\real^+}
\newcommand\realnn{\real^+_0}
\newcommand\calA{\mathcal{A}}
\newcommand\calB{\mathcal{B}}
\newcommand\calC{\mathcal{C}}

\newcommand\calF{\mathcal{F}}

\newcommand\calH{\mathcal{H}}

\newcommand\calP{\mathcal{P}}

\newcommand\calR{\mathcal{R}}

\newcommand\calT{\mathcal{T}}
\newcommand\calU{\mathcal{U}}
\newcommand\calV{\mathcal{V}}

\newcommand\dd{\mathbf{d}}

\newcommand\subseq[2]{#1/#2}
\newcommand\len[1]{\left\vert #1 \right\vert}
\newcommand{\concat}{\cdot}
\newcommand{\Concat}{\prod}

\newcommand\mult{\cdot}

\newcommand{\llbrace}{\{\!\{}
\newcommand{\rrbrace}{\}\!\}}

\makeatletter
\newcommand{\Cxt}[3][C]{
  \ifthenelse {\equal{#2}{b}}
  {#1[#3]}
  { \ifthenelse {\equal{#2}{bb}}
    {#1\llbracket #3 \rrbracket}
    { \ifthenelse {\equal{#2}{a}}
      {#1\langle #3 \rangle}
      { \ifthenelse {\equal{#2}{aa}}
        {#1\llangle #3 \rrangle}
        { \ifthenelse {\equal{#2}{c}}
          {#1\{ #3 \}}
          { \ifthenelse {\equal{#2}{cc}}
            {#1\llbrace #3 \rrbrace}
            {
              \latex@error{Wrong argument}\@ehc
            }
          }
        }
      }
    }
  }
}
\makeatother

\newcommand\nin{\not\in}

\newcommand\limto{\rightarrow}

\newcommand{\setcom}[2]{\set{#1\left|\vphantom{#1}\,#2\right.}}
\newcommand{\set}[1]{\left\{#1\right\}}

\newcommand\fcolon{\colon\,}
\newcommand\funto{\rightarrow}
\newcommand\symb[1]{\mathsf{#1}}
\newcommand\dom[1]{\symb{dom}(#1)}
\newcommand\srank[1]{\symb{ar}(#1)}
\newcommand\prpty[1]{\mathsf{#1}}
\newcommand\eprpty[1]{\ensuremath{\prpty{#1}}}
\newcommand\UN{\eprpty{UN}}

\newcommand\iprpty[1]{\ensuremath{#1^\infty}}

\newcommand\iUN{\iprpty{\UN}}

\newcommand\similar[2]{\symb{sim}(#1,#2)}

\newcommand\trunc[2]{#1|#2}
\newcommand\sdepth[2]{{#2}\mhyphen\symb{depth}(#1)}
\newcommand\atPos[2]{#1|_{#2}}
\newcommand\substAtPos[3]{#1[#3]_{#2}}
\newcommand\calPos{\calP}
\newcommand\pos[1]{\calPos(#1)}
\newcommand\posNonBot[1]{\calPos_{\not\bot}(#1)}
\newcommand\posBot[1]{\calPos_{\bot}(#1)}
\newcommand\posFun[1]{\calPos_{\Sigma}(#1)}
\newcommand\unlab[1]{\left\lVert #1 \right\rVert}

\newcommand\hCol[1][\nothing]
{\ifthenelse{\equal{#1}{\nothing}}{\calH\calC}{\calH\calC_{#1}}}
\newcommand\rAct[1][\nothing]
{\ifthenelse{\equal{#1}{\nothing}}{\calR\calA}{\calR\calA_{#1}}}

\newcommand\prs{p}
\newcommand\mrs{m}

\newcommand\segm[3]{#1|_{[#2,#3)}}
\newcommand\prefix[2]{#1|_{#2}}
\newcommand\traces[3]{\calT r(#1,#2,#3)}
\newcommand\postraces[3]{\calP\calT r(#1,#2,#3)}
\newcommand\paths[3]{\calP(#1,#2,#3)}
\newcommand\trace[2]{\symb{tr}_{#1}(#2)}
\newcommand\symtrace[2]{\symb{sym}_{#1}(#2)}
\newcommand\postrace[1]{\symb{pos}(#1)}
\newcommand\devTerm[3]{\calF(#1,#2,#3)}

\newcommand\emptylab{\emptyset}

\newcommand\edge[1][\nothing]{%
\ifthenelse{\equal{#1}{\nothing}}%
{\rightarrow}%
{\stackrel{#1}{\rightarrow}}%
}
\newcommand\node[1]{%
\def\nodeN{#1}%
\nodeI
}
\newcommand\nodeI[1][\nothing]{%
  \ifthenelse{\equal{#1}{\nothing}}%
  {\nodeN}%
  {{\nodeN}^{#1}}%
}

\newcommand\ipterms[1][\Sigma]{%
\def\itermsSig{#1_\bot}%
\itermsI%
}
\newcommand\iterms[1][\Sigma]{%
\def\itermsSig{#1}%
\itermsI%
}
\newcommand\itermsI[1][\calV]{%
\def\itermsVar{#1}%
\calT^\infty(\itermsSig,\itermsVar)%
}
\newcommand\terms[1][\Sigma]{%
\def\termsSig{#1}%
\termsI%
}
\newcommand\termsI[1][\calV]{%
\def\termsVar{#1}%
\calT(\termsSig,\termsVar)%
}

\newcommand\dEsc[2]{#1/\!/#2}
\newcommand\proj[2]{#1/#2}
\def\nothing{}
\let\oldTo\to

\newcommand\finleftright{\leftrightarrow}
\newcommand\finright{\oldTo}
\newcommand\finleft{\leftarrow}
\newcommand\weakright{\hookrightarrow}
\newcommand\weakleft{\hookleftarrow}
\newcommand\strongright{\twoheadrightarrow}
\newcommand\strongleft{\twoheadleftarrow}

\newcommand\mrsright{\mathrel{\twoheadrightarrow^{\hspace{-11pt}m}\hspace{3pt}}}
\newcommand\mrsleft{\mathrel{\twoheadleftarrow^{\hspace{-7pt}m}\hspace{0pt}}}
\newcommand\prsright{\mathrel{\twoheadrightarrow^{\hspace{-9pt}p}\hspace{4pt}}}
\newcommand\prsleft{\mathrel{\twoheadleftarrow^{\hspace{-5pt}p}\hspace{0pt}}}

\newcommand\mrswright{\mathrel{\hookrightarrow^{\hspace{-10pt}m}\hspace{2pt}}}
\newcommand\mrswleft{\mathrel{\hookleftarrow^{\hspace{-7pt}m}\hspace{0pt}}}
\newcommand\prswright{\mathrel{\hookrightarrow^{\hspace{-9pt}p}\hspace{4pt}}}
\newcommand\prswleft{\mathrel{\hookleftarrow^{\hspace{-5pt}p}\hspace{0pt}}}

\newcommand{\RewArr}[2] {
  \RewStmt{#1}{\nothing}{#2}
}

\newcommand{\RewStmt}[3] {
  \def\RewArrArr{#1}
  \def\RewArrRhs{#2}
  \def\RewArrIter{#3}
  \RewArrI
}

\makeatletter \newcommand{\RewArrI}[1][\nothing] { \def\RewArrCxt{#1}
  \ifthenelse{\equal{\RewArrArr}{\oldTo} \OR
    \equal{\RewArrArr}{\finright} \OR
    \equal{\RewArrArr}{\finleftright} \OR
    \equal{\RewArrArr}{\prsright} \OR
    \equal{\RewArrArr}{\mrsright} \OR
    \equal{\RewArrArr}{\prswright} \OR
    \equal{\RewArrArr}{\mrswright} \OR
    \equal{\RewArrArr}{\strongright} \OR
    \equal{\RewArrArr}{\weakright}} {
    \RewArrArr\ifthenelse{\equal{\RewArrIter}{\nothing}}{}{^{\RewArrIter}}\ifthenelse{\equal{\RewArrCxt}{\nothing}}{}{_{\RewArrCxt}}
  } { \ifthenelse{\equal{\RewArrArr}{\finleft} \OR
      \equal{\RewArrArr}{\prsleft} \OR
      \equal{\RewArrArr}{\mrsleft} \OR
      \equal{\RewArrArr}{\prswleft} \OR
      \equal{\RewArrArr}{\mrswleft} \OR
      \equal{\RewArrArr}{\strongleft} \OR
      \equal{\RewArrArr}{\weakleft}}{
      \RewArrArr\ifthenelse{\equal{\RewArrIter}{\nothing}}{}{^{\RewArrIter}}\ifthenelse{\equal{\RewArrCxt}{\nothing}}{}{_{\RewArrCxt}}
    } { \latex@error{Rewrite arrow not defined}\@ehc } } \RewArrRhs }
\makeatother

\renewcommand{\to}{\RewArr{\finright}{\nothing}}
\newcommand{\fto}{\RewArr{\finright}}
\newcommand{\pto}{\RewArr{\prsright}}
\newcommand{\pato}{\RewArr{\prsright}{\nothing}}
\newcommand{\pafrom}{\RewArr{\prsleft}{\nothing}}

\newcommand{\pwto}{\RewArr{\prswright}}
\newcommand{\pwato}{\RewArr{\prswright}{\nothing}}
\newcommand{\pacont}{\RewStmt{\prsright}{\dots}{\nothing}}

\newcommand{\macont}{\RewStmt{\mrsright}{\dots}{\nothing}}

\newcommand{\pwacont}{\RewStmt{\prswright}{\dots}{\nothing}}

\newcommand{\mwacont}{\RewStmt{\mrswright}{\dots}{\nothing}}

\newcommand{\mto}{\RewArr{\mrsright}}
\newcommand{\mato}{\RewArr{\mrsright}{\nothing}}
\newcommand{\mafrom}{\RewArr{\mrsleft}{\nothing}}

\newcommand{\mwto}{\RewArr{\mrswright}}
\newcommand{\mwato}{\RewArr{\mrswright}{\nothing}}
\newcommand{\mwafrom}{\RewArr{\mrswleft}{\nothing}}

\newcommand{\sato}{\RewArr{\strongright}{\nothing}}
\newcommand{\wato}{\RewArr{\weakright}{\nothing}}
\usepackage{subfig}

\theoremstyle{definition}
\theoremstyle{plain}
\theoremstyle{plain}
\theoremstyle{plain}
\theoremstyle{plain}
\theoremstyle{definition}
\theoremstyle{definition}

\colorlet{termback}{black!10}
\colorlet{termfringe}{black!70}
\colorlet{termconst}{black!30}
\colorlet{termconstfringe}{black!30}
\colorlet{termstable}{black!50}
\colorlet{termstablefringe}{black!50}
\colorlet{redsite}{black!80}

\usetikzlibrary{matrix}
\usetikzlibrary{arrows,backgrounds,positioning,fit,shapes,calc}
\usetikzlibrary{%
  decorations.pathmorphing,%
  decorations.pathreplacing,%
  decorations.markings,%
  decorations.shapes,%
  decorations.footprints,%
  decorations.fractals,%
  decorations.text%
}%
\tikzset{%
  shorten/.style={shorten >=#1, shorten <=#1},%
  level distance=1cm,%
  edge from parent/.style={->,draw},%
  force edge/.style={edge from parent/.style={->,draw}},%
  no edge/.style={edge from parent/.style={}},%
  >=latex',%
  label distance=-3pt,%
  site/.style={%
    shape=circle,%
    minimum size=4mm%
  },%
  red site/.style={%
    site,%
    draw=redsite%
  },%
  term back/.style={%
    draw=termfringe,%
    fill=termback%
  },%
    term const/.style={%
    draw=termconstfringe,%
    fill=termconst%
  },%
  term stable/.style={%
    draw=termstablefringe,%
    fill=termstable%
  },%
  term fringe/.style={%
    draw=termfringe,%
  },%
  fun/.style={>=to,->},%
  single step/.style={>=to,->},%
  conv/.style={>=to,<->},%
  strongly/.style={>=to,->>},%<<
  weakly/.style={>=to,right hook->},%
  weakly left/.style={>=to,left hook->}%
}%

\begin{document}

\title[Partial Order Infinitary Term Rewriting]
      {Partial Order Infinitary Term Rewriting\rsuper*}

\author[Patrick Bahr]{Patrick Bahr}
\address{
  Department of Computer Science, University of Copenhagen
  \newline
  Universitetsparken 5, 2100 Copenhagen, Denmark}
\urladdr{http://www.diku.dk/\~{}paba}
\email{paba@di.ku.dk}

\keywords{infinitary term rewriting, Böhm trees, partial order,
  confluence, normalisation}
\subjclass{F.4.2}
\titlecomment{{\lsuper*}Parts of this paper have appeared in the
  proceedings of RTA 2010 \cite{bahr10rta2}.}
%%%%%%%%%%%%%%%%%%%%%%%%%%%%%%%%%%%%%%%%%%%%%%%%%%%%%%%%%%%%%%%%%%%%%%%%%%%

%% the abstract has to PRECEED the command \maketitle:
%% be sure not to issue the \maketitle command twice!

\begin{abstract}
  We study an alternative model of infinitary term rewriting. Instead
  of a metric on terms, a partial order on partial terms is employed
  to formalise convergence of reductions. We consider both a weak and
  a strong notion of convergence and show that the metric model of
  convergence coincides with the partial order model restricted to
  total terms. Hence, partial order convergence constitutes a
  conservative extension of metric convergence, which additionally
  offers a fine-grained distinction between different levels of
  divergence.

  In the second part, we focus our investigation on strong convergence
  of orthogonal systems. The main result is that the gap between the
  metric model and the partial order model can be bridged by extending
  the term rewriting system by additional rules. These extensions are
  the well-known Böhm extensions. Based on this result, we are able to
  establish that -- contrary to the metric setting -- orthogonal
  systems are both infinitarily confluent and infinitarily normalising
  in the partial order setting. The unique infinitary normal forms
  that the partial order model admits are Böhm trees.
\end{abstract}

\maketitle

\section*{Introduction}
\label{sec:introduction}

Infinitary term rewriting \cite{kennaway03book} extends the theory of
term rewriting by giving a meaning to transfinite rewriting
sequences. Its formalisation \cite{dershowitz91tcs} is chiefly based
on the metric space of terms as studied by Arnold and
Nivat~\cite{arnold80fi}. Other models for transfinite reductions,
using for example general topological spaces \cite{rodenburg98jsyml}
or partial orders \cite{corradini93tapsoft,blom04rta}, were mainly
considered to pursue quite specific purposes and have not seen nearly
as much attention as the metric model. In this paper we introduce a
novel foundation of infinitary term rewriting based on the partially
ordered set of partial terms \cite{goguen77jacm}. We show that this
model of infinitary term rewriting is superior to the metric
model. This assessment is supported by two findings: First, the
partial order model of infinitary term rewriting conservatively
extends the metric model. That is, anything that can be done in the
metric model can be achieved in the partial order model as well by
simply restricting it to the set of total terms. Secondly, unlike the
metric model, the partial order model provides a fine-grained
distinction between different levels of divergence and exhibits nice
properties like infinitary confluence and normalisation of orthogonal
systems.

The defining core of a theory of infinitary term rewriting is its
notion of convergence for transfinite reductions: which transfinite
reductions are ``admissible'' and what is their final outcome. In this
paper we study both variants of convergence that are usually
considered in the established theory of metric infinitary term
rewriting: weak convergence \cite{dershowitz91tcs} and strong
convergence \cite{kennaway95ic}. For both variants we introduce a
corresponding notion of convergence based on the partially ordered set
of partial terms.

The first part of this paper is concerned with comparing the metric
model and the partial order model both in their respective weak and
strong variants. In both cases, the partial order approach constitutes
a conservative extension of the metric approach: a reduction in the
metric model is converging iff it is converging in the partial order
model and only contains total terms.

In the second part we focus on strong convergence in orthogonal
systems. To this end we reconsider the theory of meaningless terms of
Kennaway et al.\ \cite{kennaway99jflp}. In particular, we consider
Böhm extensions. The Böhm extension of a term rewriting system adds
rewrite rules which admit contracting meaningless terms to $\bot$. The
central result of the second part of this paper is that the additional
rules in Böhm extensions close the gap between partial order
convergence and metric convergence. More precisely, we show that
reachability w.r.t.\ partial order convergence in a term rewriting
system coincides with reachability w.r.t.\ metric convergence in the
corresponding Böhm extension.

From this result we can easily derive a number of properties for
strong partial order convergence in orthogonal systems:
\begin{itemize}
\item Infinitary confluence,
\item infinitary normalisation, and
\item compression, i.e.\ each reduction can be compressed to length at
  most $\omega$
\end{itemize}
The first two properties exhibit another improvement over the metric
model which does not have either of these. Moreover, it means that
each term has a unique infinitary normal form -- its Böhm tree.

The most important tool for establishing these results is provided by
a notion of complete developments that we have transferred from the
metric approach to infinitary rewriting \cite{kennaway95ic}. We show,
that the final outcome of a complete development is unique and that,
in contrast to the metric model, the partial order model admits
complete developments for any set of redex occurrences. To this end,
we use a technique similar to paths and finite jumps known from metric
infinitary term rewriting \cite{kennaway03book, ketema11ic}.

\subsubsection*{Outline}

After providing the basic preliminaries for this paper in
Section~\ref{sec:preliminaries}, we will briefly recapitulate the
metric model of infinitary term rewriting including meaningless terms
and Böhm extensions in Section~\ref{sec:metr-infin-term}. In
Section~\ref{sec:part-order-infin}, we introduce our novel approach to
infinitary term rewriting based on the partial order on terms. In
Section~\ref{sec:comp-mrs-conv}, we compare both models and establish
that the partial order model provides a conservative extension of the
metric model. In the remaining part of this paper, we focus on the
strong notion of convergence. In Section~\ref{sec:compl-devel}, we
establish a theory of complete developments in the setting of partial
order convergence. This is then used in
Section~\ref{sec:relation-bohm-trees} to prove the equality of
reachability w.r.t.\ partial order convergence and reachability
w.r.t.\ metric convergence in the Böhm extension. Finally, we evaluate
our results and point to interesting open questions in
Section~\ref{sec:conclusions}.

% \tableofcontents

\section{Preliminaries}
\label{sec:preliminaries}

We assume the reader to be familiar with the basic theory of ordinal
numbers, orders and topological spaces \cite{kelley55book}, as well as
term rewriting \cite{terese03book}. In the following, we
briefly recall the most important notions.

\subsection{Transfinite Sequences}
We use $\alpha, \beta, \gamma, \lambda, \iota$ to denote ordinal
numbers. A \emph{transfinite sequence} (or simply called
\emph{sequence}) $S$ of length $\alpha$ in a set $A$, written
$(a_\iota)_{\iota < \alpha}$, is a function from $\alpha$ to $A$ with
$\iota \mapsto a_\iota$ for all $\iota \in \alpha$. We use $\len{S}$
to denote the length $\alpha$ of $S$. If $\alpha$ is a limit ordinal,
then $S$ is called \emph{open}. Otherwise, it is called
\emph{closed}. If $\alpha$ is a finite ordinal, then $S$ is called
\emph{finite}. Otherwise, it is called \emph{infinite}. For a finite
sequence $(a_i)_{i < n}$ or a sequence $(a_i)_{i < \omega}$ of length
$\omega$, we also use the notation $\seq{a_0,a_1,\dots,a_{n-1}}$
respectively $\seq{a_0,a_1,\dots}$. In particular, $\emptyseq$ denotes
an empty sequence.

The \emph{concatenation} $(a_\iota)_{\iota<\alpha}\concat
(b_\iota)_{\iota<\beta}$ of two sequences is the sequence
$(c_\iota)_{\iota<\alpha+\beta}$ with $c_\iota = a_\iota$ for $\iota <
\alpha$ and $c_{\alpha+\iota} = b_\iota$ for $\iota < \beta$. A
sequence $S$ is a (proper) \emph{prefix} of a sequence $T$, denoted $S
\le T$ (resp.\ $S < T$), if there is a (non-empty) sequence $S'$ with
$S\concat S' = T$. The prefix of $T$ of length $\beta$ is denoted
$\prefix{T}{\beta}$. The binary relation $\le$ forms a complete
semilattice (see Section~\ref{sec:partial-orders} below). Similarly, a
sequence $S$ is a (proper) \emph{suffix} of a sequence $T$ if there is
a (non-empty) sequence $S'$ with $S'\concat S = T$.

Let $S = (a_\iota)_{\iota < \alpha}$ be a sequence. A sequence $T =
(b_{\iota})_{\iota < \beta}$ is called a \emph{subsequence} of $S$ if
there is a monotone function $f\fcolon \beta \to \alpha$ such that
$b_\iota = a_{f(\iota)}$ for all $\iota < \beta$. To indicate this, we
write $\subseq{S}{f}$ for the subsequence $T$. If $f(\iota) = f(0) +
\iota$ for all $\iota < \beta$, then $\subseq{S}{f}$ is called a
\emph{segment} of $S$. That is, $T$ is a segment of $S$ iff there are
two sequences $T_1, T_2$ such that $S = T_1\concat T \concat T_2$. We
write $\segm{S}{\beta}{\gamma}$ for the segment $\subseq{S}{f}$, where
$f\fcolon \alpha' \to \alpha$ is the mapping defined by $f(\iota) =
\beta + \iota$ for all $\iota < \alpha'$, with $\alpha'$ the unique
ordinal with $\gamma = \beta + \alpha'$. Note that in particular
$\segm{S}{0}{\alpha} = \prefix{S}{\alpha}$ for each sequence $S$ and
ordinal $\alpha \le \len{S}$.

\subsection{Metric Spaces}
\label{sec:metric-spaces}

A pair $(M,\dd)$ is called a \emph{metric space} if $\dd \fcolon M
\times M \to \realnn$ is a function satisfying $\dd(x,y) = 0$ iff
$x=y$ (identity), $\dd(x, y) = \dd(y, x)$ (symmetry), and $\dd(x, z)
\le \dd(x, y) + \dd(y, z)$ (triangle inequality), for all $x,y,z\in
M$.  If $\dd$ instead of the triangle inequality, satisfies the
stronger property $\dd(x, z) \le \max \set{ \dd(x, y),\dd(y, z)}$
(strong triangle), then $(M,\dd)$ is called an \emph{ultrametric
  space}. Let $(a_\iota)_{\iota<\alpha}$ be a sequence in a metric
space $(M,\dd)$. The sequence $(a_\iota)_{\iota<\alpha}$
\emph{converges} to an element $a\in M$, written
$\lim_{\iota\limto\alpha} a_\iota$, if, for each $\varepsilon \in
\realp$, there is a $\beta < \alpha$ such that $\dd(a,a_\iota) <
\varepsilon$ for every $\beta < \iota < \alpha$;
$(a_\iota)_{\iota<\alpha}$ is \emph{continuous} if
$\lim_{\iota\limto\lambda} a_\iota = a_\lambda$ for each limit ordinal
$\lambda < \alpha$. The sequence $(a_\iota)_{\iota<\alpha}$ is called
\emph{Cauchy} if, for any $\varepsilon \in \realp$, there is a
$\beta<\alpha$ such that, for all $\beta < \iota < \iota' < \alpha$,
we have that $\dd(m_\iota,m_{\iota'}) < \varepsilon$.  A metric space
is called \emph{complete} if each of its non-empty Cauchy sequences
converges.

\subsection{Partial Orders}
\label{sec:partial-orders}
A \emph{partial order} $\le$ on a set $A$ is a binary relation on $A$
that is \emph{transitive}, \emph{reflexive}, and
\emph{antisymmetric}. The pair $(A,\le)$ is then called a
\emph{partially ordered set}. We use $<$ to denote the strict part of
$\le$, i.e. $a < b$ iff $a \le b$ and $b \not\le a$. A sequence
$(a_\iota)_{\iota<\alpha}$ in $(A,\le)$ is called a (\emph{strict})
\emph{chain} if $a_\iota \le a_\gamma$ (resp.\ $a_\iota < a_\gamma$)
for all $\iota < \gamma < \alpha$. A subset $D$ of the underlying set
$A$ is called \emph{directed} if it is non-empty and each pair of
elements in $D$ has an upper bound in $D$. A partially ordered set
$(A, \le)$ is called a \emph{complete semilattice} if it has a
\emph{least element}, every \emph{directed subset} $D$ of $A$ has a
\emph{least upper bound} (\emph{lub}) $\Lub D$, and every subset of
$A$ having an upper bound also has a least upper bound. Hence,
complete semilattices also admit a \emph{greatest lower bound}
(\emph{glb}) $\Glb B$ for every \emph{non-empty} subset $B$ of $A$. In
particular, this means that for any non-empty sequence
$(a_\iota)_{\iota<\alpha}$ in a complete semilattice, its \emph{limit
  inferior}, defined by $\liminf_{\iota \limto \alpha}a_\iota =
\Lub_{\beta<\alpha} \left(\Glb_{\beta \le \iota < \alpha}
  a_\iota\right)$, exists.

It is easy to see that the limit inferior of closed sequences is
simply the last element of the sequence. This is, however, only a
special case of the following more general proposition:
\begin{prop}[invariance of the limit inferior]
  \label{prop:liminfSuffix}
  % invariance of limit inferior %
  Let $(a_\iota)_{\iota < \alpha}$ be a sequence in a complete
  semilattice and $(b_\iota)_{\iota< \beta}$ a non-empty suffix of
  $(a_\iota)_{\iota < \alpha}$. Then $\liminf_{\iota \limto \alpha}
  a_\iota = \liminf_{\iota \limto \beta} b_\iota$.
\end{prop}
\begin{proof}
  Let $a = \liminf_{\iota \limto \alpha} a_\iota$ and $b =
  \liminf_{\iota \limto \beta} b_\iota$.  Since $(b_\iota)_{\iota<
    \beta}$ is a suffix of $(a_\iota)_{\iota < \alpha}$, there is some
  $\delta < \alpha$ such that $b_\iota = a_{\delta+\iota}$ for all
  $\iota < \beta$. Hence, we know that $a = \Lub_{\gamma<\alpha}
  \Glb_{\gamma \le \iota < \alpha} a_\iota$ and $b = \Lub_{\delta \le
    \gamma<\alpha} \Glb_{\gamma \le \iota < \alpha} a_\iota$.  Let
  $c_\gamma = \Glb_{\gamma \le \iota < \alpha} a_\iota$ for each
  $\gamma < \alpha$, $A = \setcom{c_\gamma}{\gamma < \alpha}$ and $B =
  \setcom{c_\gamma}{\delta \le \gamma < \alpha}$. Note that $a = \Lub
  A$ and $b = \Lub B$. Because $B \subseteq A$, we have that $b \le
  a$. On the other hand, since $c_\gamma \le c_{\gamma'}$ for $\gamma
  \le \gamma'$, we find, for each $c_\gamma \in A$, some $c_{\gamma'}
  \in B$ with $c_\gamma \le c_{\gamma'}$. Hence, $a \le b$. Therefore,
  due to the antisymmetry of $\le$, we can conclude that $a = b$.
\end{proof} 
Note that the limit in a metric space has the same behaviour as the
one for the limit inferior described by the proposition
above. However, one has to keep in mind that -- unlike the limit --
the limit inferior is not invariant under taking cofinal subsequences!

With the prefix order $\le$ on sequences we can generalise
concatenation to arbitrary sequences of sequences: Let
$(S_\iota)_{\iota < \alpha}$ be a sequence of sequences in a common
set. The concatenation of $(S_\iota)_{\iota < \alpha}$, written
$\Concat_{\iota < \alpha} S_\iota$, is recursively defined as the
empty sequence $\emptyseq$ if $\alpha = 0$, $\left(\Concat_{\iota <
    \alpha'} S_\iota\right) \concat S_{\alpha'}$ if $\alpha = \alpha'
+ 1$, and $\Lub_{\gamma < \alpha} \Concat_{\iota < \gamma} S_\iota$ if
$\alpha$ is a limit ordinal. For instance, the concatenation
$\Concat_{i<\omega} \seq{i,i+1}$ yields the sequence
$\seq{0,1,1,2,2,\dots}$ of length $\omega$, and the concatenation
$\Concat_{\iota<\alpha} \seq{\iota}$, for any ordinal $\alpha$, yields
the sequence $(\iota)_{\iota<\alpha}$.

\subsection{Terms}
\label{sec:terms}

Unlike in the traditional -- i.e.\ finitary -- framework of term
rewriting, we consider the set $\iterms$ of \emph{infinitary terms}
(or simply \emph{terms}) over some \emph{signature} $\Sigma$ and a
countably infinite set $\calV$ of variables. A \emph{signature}
$\Sigma$ is a countable set of symbols. Each symbol $f$ is associated
with its arity $\srank{f}\in \nat$, and we write $\Sigma^{(n)}$ for
the set of symbols in $\Sigma$ which have arity $n$. The set $\iterms$
is defined as the \emph{greatest} set $T$ such that, for each element
$t \in T$, we either have $t \in \calV$ or $t = f(t_0,\dots,
t_{k-1})$, where $f \in \Sigma^{(k)}$, and $t_0,\dots,t_{k-1}\in T$. A
symbol $c \in \Sigma^{(0)}$ of arity $0$ is also called a
\emph{constant symbol}, and we use the shorthand $c$ to denote a term
$c()$. We consider $\iterms$ as a superset of the set $\terms$ of
\emph{finite terms}.

For each term $t \in \iterms$, we define the \emph{set of positions}
in $t$, denoted $\pos{t}$, as the smallest set of finite sequences in
$\nat$ such that $\emptyseq \in \pos t$, and $\seq i\concat \pi \in
\pos t$ whenever $t = f(t_0,\dots,t_{k-1})$, $i<k$, and $\pi \in
\pos{t_i}$. Given a position $\pi \in \pos t$, we define the
\emph{subterm} of $t$ at $\pi$, denoted $\atPos{t}{\pi}$, by recursion
on $\pi$ as follows: $\atPos t \emptyseq = t$, and
$\atPos{f(t_0,\dots,t_{k-1})}{\seq i \concat \pi} = \atPos{t_i}{\pi}$.
Moreover, we write $t(\pi)$ for the symbol in $t$ at $\pi$, i.e.\
$t(\pi) = f$ if $\atPos t \pi = f(t_0,\dots,t_{k-1})$ and $t(\pi) = v$
if $\atPos t \pi = v \in \calV$. For terms $s,t \in \iterms$ and a
position $\pi \in \pos{t}$, we write $\substAtPos{t}{\pi}{s}$ for the
term $t$ with the subterm at $\pi$ replaced by $s$, i.e.\
\[
\substAtPos{t}{\emptyseq}{s} = s,\quad\text{ and }\quad
\substAtPos{f(t_0,\dots,t_{k-1})}{\seq i\concat \pi}{s} =
f(t_0,\dots,t_{i-1},\substAtPos{t_i}{\pi}{s},t_{i+1},\dots,t_{k-1}).
\]
Note that while the set of terms $\iterms$ is defined coinductively,
the set of positions of a term is defined inductively. Consequently,
the subterm at a position and substitution at a position are defined
by recursion.

Two terms $s$ and $t$ are said to \emph{coincide} in a set of
positions $P \subseteq \pos{s} \cap \pos{t}$ if $s(\pi) = t(\pi)$ for
all $\pi \in P$. A position is also called an \emph{occurrence} if the
focus lies on the subterm at that position rather than the position
itself. Two positions $\pi_1, \pi_2$ are called \emph{disjoint} if
neither $\pi_1 \le \pi_2$ nor $\pi_2 \le \pi_1$.

A \emph{context} is a ``term with holes'', which are represented by a
distinguished variable $\Box$. We write $\Cxt{b}{,\dots,}$ for a
context with at least one occurrence of $\Box$, and $\Cxt{a}{,\dots,}$
for a context with zero more occurrences of $\Box$. $\Cxt{b}{t_1,
  \dots, t_n}$ denotes the result of replacing the occurrences of
$\Box$ in $C$ (from left to right) by $t_1,\dots,t_n$. $\Cxt{a}{t_1,
  \dots, t_n}$ is defined accordingly.

A \emph{substitution} $\sigma$ is a mapping from $\calV$ to
$\iterms$. Its \emph{domain}, denoted $\dom{\sigma}$, is the set
$\setcom{x \in \calV}{\sigma(x) \neq x}$ of variables not mapped to
itself by $\sigma$. Substitutions are uniquely extended to functions
from $\iterms$ to $\iterms$: $\sigma(f (t_1, \dots, t_n )) = f
(\sigma(t_1 ), \dots ,\sigma(t_n ))$ for $f \in \Sigma^{(n)}$ and
$t_1, \dots, t_n \in \iterms$. Instead of $\sigma(s)$, we shall also
write $s\sigma$.

On $\iterms$ a similarity measure $\similar{\cdot}{\cdot} \in \nat
\cup \set{\infty}$ can be
defined by setting
\[
\similar{s}{t} = \min \setcom{\len{\pi}}{\pi \in \pos{s}\cap\pos{t}, s(\pi) \neq
  t(\pi)} \cup \set{\infty} \qquad \text {for } s,t\in \iterms
\]
That is, $\similar{s}{t}$ is the minimal depth at which $s$ and $t$
differ, or $\infty$ if $s = t$. Based on this, a distance function
$\dd$ can be defined by $\dd(s,t) = 2^{-\similar{s}{t}}$, where we
interpret $2^{-\infty}$ as $0$. Note that $0\le\dd(s,t)\le 1$. In
particular, $\dd(s,t) = 0$ iff $s$ and $t$ coincide, and $\dd(s,t) =
1$ iff $s$ and $t$ differ at the root. The pair $(\iterms, \dd)$ is
known to form a complete ultrametric space
\cite{arnold80fi}. \emph{Partial terms}, i.e.\ terms over signature
$\Sigma_\bot = \Sigma \uplus \set{\bot}$ with $\bot$ a fresh constant
symbol, can be endowed with a binary relation $\lebot$ by defining $s
\lebot t$ iff $s$ can be obtained from $t$ by replacing some subterm
occurrences in $t$ by $\bot$. Interpreting the term $\bot$ as denoting
``undefined'', $\lebot$ can be read as ``is less defined than''. The
pair $(\ipterms,\lebot)$ is known to form a complete semilattice
\cite{goguen77jacm}. For a partial term $t\in \ipterms$ we use the
notation $\posNonBot{t}$ and $\posFun{t}$ for the set $\setcom{\pi\in
  \pos{t}}{t(\pi) \neq \bot}$ of non-$\bot$ positions resp.\ the set
$\setcom{\pi\in\pos{t}}{t(\pi) \in \Sigma}$ of positions of function
symbols. With this, $\lebot$ can be characterised alternatively by $s
\lebot t$ iff $s(\pi) = t(\pi)$ for all $\pi \in \posNonBot{s}$. To
explicitly distinguish them from partial terms, we call terms in
$\iterms$ \emph{total}.

\subsection{Term Rewriting Systems}
\label{sec:abstr-reduct-syst}

A \emph{term rewriting system} (TRS) $\calR$ is a pair $(\Sigma, R)$
consisting of a signature $\Sigma$ and a set $R$ of \emph{term rewrite
  rules} of the form $l \to r$ with $l \in \iterms \setminus \calV$
and $r \in \iterms$ such that all variables occurring in $r$ also
occur in $l$. Note that this notion of a TRS deviates slightly from
the standard notion of TRSs in the literature on infinitary rewriting
\cite{kennaway03book} in that it allows infinite terms on the
left-hand side of rewrite rules! This generalisation will be necessary
to accommodate Böhm extensions, which are introduced later in
Section~\ref{sec:mean-terms-bohm}. TRSs having only finite left-hand
sides are called \emph{left-finite}.

As in the finitary setting, every TRS $\calR$ defines a rewrite
relation $\to[\calR]$:
\[
s \to[\calR] t \iff \exists \pi \in \pos{s}, l\to r \in
R, \sigma\colon\; \atPos{s}{\pi} = l\sigma, t = \substAtPos{s}{\pi}{r\sigma}
\]
Instead of $s \to[\calR] t$, we sometimes write $s \to[\pi,\rho] t$ in
order to indicate the applied rule $\rho$ and the position $\pi$, or
simply $s \to t$. The subterm $\atPos{s}{\pi}$ is called a
\emph{$\rho$-redex} or simply \emph{redex}, $r\sigma$ its
\emph{contractum}, and $\atPos{s}{\pi}$ is said to be
\emph{contracted} to $r\sigma$.

Let $\rho\fcolon l \to r$ be a term rewrite rule. The \emph{pattern}
of $\rho$ is the context $l\sigma$, where $\sigma$ is the substitution
$\setcom{x \mapsto \Box}{x \in \calV}$ that maps all variables to
$\Box$. If $t$ is a $\rho$-redex, then the pattern $P$ of $\rho$ is
also called the \emph{redex pattern} of $t$ w.r.t.\ $\rho$. When
referring to the occurrences in a pattern, occurrences of the symbol
$\Box$ are neglected.

Let $\rho_1\fcolon l_1 \to r_1$, $\rho_2\fcolon l_2 \to r_2$ be rules
in a TRS $\calR$. The rules $\rho_1,\rho_2$ are said to \emph{overlap}
if there is a non-variable position $\pi$ in $l_1$ such that
$\atPos{l_1}{\pi}$ and $l_2$ are unifiable and $\pi$ is not the root
position $\emptyseq$ in case $\rho_1,\rho_2$ are renamed copies of the
same rule. A TRS is called \emph{non-overlapping} if none of its rules
overlap. A term $t$ is called \emph{linear} if each variable occurs at
most once in $t$. The TRS $\calR$ is called \emph{left-linear} if the
left-hand side of every rule in $\calR$ is linear. It is called
\emph{orthogonal} if it is left-linear and non-overlapping.

\section{Metric Infinitary Term Rewriting}
\label{sec:metr-infin-term}

In this section we briefly recall the metric model of infinitary term
rewriting \cite{kennaway95ic} and some of its properties. We will use
the metric model in two ways: Firstly, it will serve as a yardstick to
compare the partial order model to. But most importantly, we will use
known results for metric infinitary rewriting and transfer them to the
partial order model. In order to accomplish the latter, we shall
develop correspondence theorems (Theorem~\ref{thr:weakExt} and
Theorem~\ref{thr:strongExt}) that relate convergence in the metric
model and convergence in the partial order model. Specifically, these
correspondence results show that the two notions of convergence
coincide if we restrict ourselves to total terms.

At first we have to make clear what a \emph{reduction} in our setting
of infinitary rewriting is:
\begin{defi}[reduction (step)]
  \label{def:red}
  Let $\calR$ be a TRS. A \emph{reduction step} $\phi$ in $\calR$ is a
  tuple $(s,\pi,\rho,t)$ such that $s \to[\pi,\rho] t$; we also write
  $\phi\fcolon s \to[\pi,\rho] t$. A \emph{reduction} $S$ in $\calR$
  is a sequence $(\phi_\iota)_{\iota < \alpha}$ of reduction steps in
  $\calR$ such that there is a sequence $(t_\iota)_{\iota <
    \wsuc\alpha}$ of terms, with $\wsuc\alpha = \alpha$ if $S$ is open
  and $\wsuc\alpha = \alpha + 1$ if $S$ is closed, such that
  $\phi_\iota\fcolon t_\iota \to t_{\iota+1}$. If $S$ is finite, we
  write $S\fcolon t_0 \fto{*} t_\alpha$.
\end{defi}

This definition of reductions is a straightforward generalisation of
finite reductions. As an example consider the TRS with the single rule
$a \to f(a)$. In this system we get a reduction $S\fcolon a \fto{*}
f(f(f(a)))$ of length $3$:
\[
S =
\seq{\phi_0\fcolon a \to f(a),\phi_1\fcolon f(a) \to
  f(f(a)),\phi_2\fcolon f(f(a)) \to f(f(f(a)))}
\]
In a more concise notation we write
\[
S\fcolon a \to f(a) \to f(f(a)) \to f(f(f(a)))
\]
Clearly, we can extend this reduction arbitrarily often which results
in the following infinite reduction of length $\omega$:
\[
T\fcolon a \to f(a) \to f^2(a) \to f^3(a) \to f^4(a) \to \dots
\]
However, this is as far we can go with this simple definition of
reductions. As soon as we go beyond $\omega$, we get reductions which
do not make sense. For example, consider the following reduction:
\[
T\concat S\fcolon a \to f(a) \to f^2(a) \to f^3(a) \to f^4(a) \to
\dots\; a \to f(a) \to f(f(a)) \to f(f(f(a)))
\]
The reduction $T$ of length $\omega$ can be extended by an arbitrary
reduction, e.g.\ by the reduction $S$. The notion of reductions
according to Definition~\ref{def:red} is only meaningful if restricted
to reductions of length at most $\omega$. The problem is that the
$\omega$-th step in the reduction, viz.\ the second step of the form
$a \to f(a)$ in the example above, is completely independent of all
previous steps since it does not have an immediate predecessor. This
issue occurs at each limit ordinal number. An appropriate definition
of a reduction of length beyond $\omega$ requires a notion of
continuity to bridge the gaps that arise at limit ordinals. In the
next section we will present the well-know metric approach to
this. Later in Section~\ref{sec:part-order-infin}, we will introduce a
novel approach using partial orders.

\subsection{Metric Convergence}
\label{sec:metric-convergence}

In this section we consider two notions of \emph{convergence} based on
the metric on terms as defined in Section~\ref{sec:terms}. We consider
both the weak \cite{dershowitz91tcs} and the strong
\cite{kennaway95ic} variant known from the literature. Related to this
notion of convergence is a corresponding notion of
\emph{continuity}. In order to distinguish both from the partial order
model that we will introduce in Section~\ref{sec:part-order-infin} we
will use the names \emph{weak} resp.\ \emph{strong $\mrs$-convergence}
and \emph{weak} resp.\ \emph{strong $\mrs$-continuity}.

It is important to understand that a reduction is a \emph{sequence of
  reduction steps} rather than just a sequence of terms. This is
crucial for a proper definition of strong convergence resp.\
continuity, which does not only depend on the sequence of terms that
are derived within the reduction but does also depend on the positions
where the contractions take place:
\begin{defi}[$\mrs$-continuity/-convergence]
  Let $\calR$ be a TRS and $S = (\phi_\iota\fcolon t_\iota
  \to[\pi_\iota] t_{\iota+1})_{\iota < \alpha}$ a non-empty reduction
  in $\calR$. The reduction $S$ is called
  \begin{enumerate}[label=(\roman*)]
  \item \emph{weakly $\mrs$-continuous} in $\calR$, written $S\fcolon
    t_0 \mwacont[\calR]$, if $\lim_{\iota \limto \lambda} t_\iota =
    t_\lambda$ for each limit ordinal $\lambda < \alpha$.
  \item \emph{strongly $\mrs$-continuous} in $\calR$, written
    $S\fcolon t_0 \macont[\calR]$, if it is weakly $\mrs$-continuous
    and for each limit ordinal $\lambda < \alpha$, the sequence
    $(\len{\pi_\iota})_{\iota < \lambda}$ of contraction depths tends
    to infinity.
  \item \emph{weakly $\mrs$-converging} to $t$ in $\calR$, written
    $S\fcolon t_0 \mwato[\calR] t$, if it is weakly $\mrs$-continuous
    and $t = \lim_{\iota \limto \wsuc\alpha} t_\iota$.
  \item \emph{strongly $\mrs$-converging} to $t$ in $\calR$, written
    $S\fcolon t_0 \mato[\calR] t$, if it is strongly
    $\mrs$-continuous, weakly $\mrs$-converges to $t$ and, in case
    that $S$ is open, $(\len{\pi_\iota})_{\iota<\alpha}$ tends to
    infinity.
  \end{enumerate}
  Whenever $S\fcolon t_0 \mwato[\calR] t$ or $S\fcolon t_0
  \mato[\calR] t$, we say that $t$ is weakly resp.\ strongly
  \emph{$\mrs$-reachable} from $t_0$ in $\calR$. By abuse of notation
  we use $\mwato[\calR]$ and $\mato[\calR]$ as a binary relation to
  indicate weakly resp.\ strongly $\mrs$-reachability. In order to
  indicate the length of $S$ and the TRS $\calR$, we write $S\fcolon
  t_0 \mwto{\alpha}[\calR] t$ resp.\ $S\fcolon t_0 \mto{\alpha}[\calR]
  t$. The empty reduction $\emptyseq$ is considered weakly/strongly
  $\mrs$-continuous and $\mrs$-convergent for any identical start and
  end term, i.e.\ $\emptyseq\fcolon t\mato[\calR] t$ for all $t \in
  \terms$.
\end{defi}

From the above definition it is clear that strong $\mrs$-convergence
implies both weak $\mrs$-convergence and strong $\mrs$-continuity and
that both weak $\mrs$-convergence and strong $\mrs$-continuity imply
weak $\mrs$-continuity, respectively. This is indicated in
Figure~\ref{fig:contConvRel}.
\begin{figure}
  \centering
  \begin{tikzpicture}[on grid,node distance = 3cm]
    \node (sconv) {strong $\mrs$-convergence};%
    \node[below left=of sconv] (wconv) {weak $\mrs$-convergence};%
    \node[below right=of sconv] (scont) {strong $\mrs$-continuity};%
    \node[below right=of wconv] (wcont) {weak $\mrs$-continuity};%
    
    \draw[every edge/.style={draw,double, double distance=1mm,-implies}]%
    (sconv) edge (wconv)%
    (sconv) edge (scont)%
    (wconv) edge (wcont)%
    (scont) edge (wcont);%
  \end{tikzpicture}
  \caption{Relation between continuity and convergence properties.}
  \label{fig:contConvRel}
\end{figure}
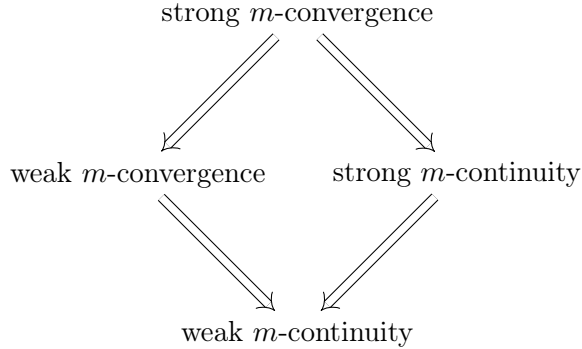
It is important to recognise that $\mrs$-convergence implies
$\mrs$-continuity. Hence, only meaningful, i.e.\ $\mrs$-continuous,
reductions can be $\mrs$-convergent.

For a reduction to be weakly $\mrs$-continuous, each open
\emph{proper} prefix of the underlying sequence
$(t_\iota)_{\iota<\wsuc\alpha}$ of terms must converge to the term
following next in the sequence -- or, equivalently,
$(t_\iota)_{\iota<\wsuc\alpha}$ must be continuous. For strong
$\mrs$-continuity, additionally, the depth at which contractions take
place has to tend to infinity for each of the reduction's open proper
prefixes. The convergence properties do only differ from the
continuity properties in that they require the above conditions to
hold for \emph{all} open prefixes, i.e.\ including the whole reduction
itself unless it is closed. For example, considering the rule $a \to
f(a)$, the reduction $g(a) \to g(f(a)) \to g(f(f(a))) \to \dots$
strongly $\mrs$-converges to the infinite term $g(f^\omega)$. The
first step takes place at depth $1$, the next step at depth $2$ and so
on. Having the rule $g(x) \to g(f(x))$ instead, the reduction $g(a)
\to g(f(a)) \to g(f(f(a))) \to \dots$ is trivially strongly
$\mrs$-continuous but is now not strongly $\mrs$-convergent since
every step in this reduction takes place at depth $0$, i.e. the
sequence of reduction depths does not tend to infinity. However, the
reduction still weakly $\mrs$-converges to $g(f^\omega)$.

In contrast to the strong notions of continuity and convergence, the
corresponding weak variants are independent from the rules that are
applied during the reduction. What makes strong $\mrs$-convergence
(and -continuity) strong is the fact that it employs a conservative
overapproximation of the differences between consecutive terms in the
reduction. For weak $\mrs$-convergence the distance
$\dd(t_\iota,t_{\iota+1})$ between consecutive terms in a reduction
$(t_\iota \to_{\pi_\iota} t_{\iota+1})_{\iota<\lambda}$ has to tend to
$0$. For strong $\mrs$-convergence the depth $\len{\pi_\iota}$ of the
reduction steps has to tend to infinity. In other words,
$2^{-\len{\pi_\iota}}$ has to tend to $0$. Note that
$2^{-\len{\pi_\iota}}$ is a conservative overapproximation of
$\dd(t_\iota,t_{\iota+1})$, i.e.\ $2^{-\len{\pi_\iota}} \ge
\dd(t_\iota,t_{\iota+1})$. So strong $\mrs$-convergence is simply weak
$\mrs$-convergence w.r.t.\ this overapproximation of $\dd$
\cite{bahr10rta}.  If this approximation is actually precise,
i.e. coincides with the actual value, both notions of
$\mrs$-convergence coincide.

\begin{rem}
  \label{rem:mcont}
  The notion of $\mrs$-continuity can be defined solely in terms of
  $\mrs$-convergence \cite{bahr10rta}. More precisely, we have for
  each reduction $S = (t_\iota \to t_{\iota+1})_{\iota < \alpha}$ that
  $S$ is weakly $\mrs$-continuous iff every (open) proper prefix of
  $\prefix{S}{\beta}$ weakly $\mrs$-converges to $t_\beta$.
  Analogously, strong $\mrs$-continuity can be characterised in terms
  of strong $\mrs$-convergence. An easy consequence of this is that
  $\mrs$-converging reductions are closed under concatenation, i.e.\
  $S\fcolon s \mwato t$, $T\fcolon t \mwato u$ implies $S\concat T
  \fcolon s \mwato u$ and likewise for strong $\mrs$-convergence.
\end{rem}

For the most part our focus in this paper is set on strong
$\mrs$-convergence and its partial order correspondent that we will
introduce in Section~\ref{sec:part-order-infin}. Weak
$\mrs$-convergence is well-known to be rather unruly
\cite{simonsen04ipl}. Strong convergence is far more well-behaved
\cite{kennaway95ic}. Most prominently, we have the following
Compression Lemma \cite{kennaway95ic} which in general does not hold
for weak $\mrs$-convergence:
\begin{thm}[Compression Lemma]
  \label{thr:mrsCompr}
  % Compression Lemma %
  % [kennaway95ic, Lem. 5.1]
  For each left-linear, left-finite TRS, $s \mato t$ implies $s
  \mto{\le \omega} t$.
\end{thm}
% left-finiteness is necessary: Consider a -> f(a), f^\omega -> b and
% reduction a -> f(a) -> ... -> f^\omega -> b

As an easy corollary we obtain that the final term of a strongly
$\mrs$-converging reduction can be approximated arbitrarily accurately
by a finite reduction:
\begin{cor}[finite approximation]
  \label{cor:mrsFinApprox}
  Let $\calR$ be a left-linear, left-finite TRS and $s \mato t$. Then,
  for each depth $d \in \nat$, there is a finite reduction $s \fto{*}
  t'$ such that $t$ and $t'$ coincide up to depth $d$, i.e.\
  $\dd(t,t') < 2^{-d}$.
\end{cor}
\begin{proof}
  Assume $s \mato t$. By Theorem~\ref{thr:mrsCompr}, there is a
  reduction $S\fcolon s \mto{\le \omega} t$. If $S$ is of finite
  length, then we are done. If $S\fcolon s \mto{\omega} t$, then, by
  strong $\mrs$-convergence, there is some $n < \omega$ such that all
  reductions steps in $S$ after $n$ take place at a depth greater than
  $d$. Consider $\prefix{S}{n}\fcolon s \fto{*} t'$. It is clear that
  $t$ and $t'$ coincide up to depth $d$.
\end{proof}
As a special case of the above corollary, we obtain that $s \mato t$
implies $s \fto* t$ whenever $t$ is a finite term.

% left-finiteness is necessary: Consider a -> f(a), f^\omega -> b and
% reduction a -> f(a) -> ... -> f^\omega -> b

An important difference between $\mrs$-converging reductions and
finite reductions is the confluence of orthogonal systems. In contrast
to finite reachability, $\mrs$-reachability of orthogonal TRSs -- even
in its strong variant -- does not necessarily have the diamond
property, i.e.\ orthogonal systems are confluent but not infinitarily
confluent \cite{kennaway95ic}:
\begin{exa}[failure of infinitary confluence]
  \label{ex:mconfl}
  Consider the orthogonal TRS consisting of the \emph{collapsing}
  rules $\rho_1\fcolon f(x) \to x$ and $\rho_2\fcolon g(x) \to x$ and
  the infinite term $t = g(f(g(f(\dots))))$. We then obtain the
  reductions $S\fcolon t \mato g^\omega$ and $T\fcolon t \mato
  f^\omega$ by successively contracting all $\rho_1$- resp.\
  $\rho_2$-redexes. However, there is no term $s$ such that $g^\omega
  \mato s \mafrom f^\omega$ (or $g^\omega \mwato s \mwafrom f^\omega$)
  as both $g^\omega$ and $f^\omega$ can only be rewritten to
  themselves, respectively.
\end{exa}

In the following section we discuss a method for obtaining an
appropriate notion of transfinite reachability based on strong
$\mrs$-reachability which actually has the diamond property.

\subsection{Meaningless Terms and Böhm Trees}
\label{sec:mean-terms-bohm}
At the end of the previous section we have seen that orthogonal TRSs
are in general not infinitarily confluent. However, as Kennaway et
al.\ \cite{kennaway95ic} have shown, orthogonal TRSs are infinitarily
confluent modulo so-called \emph{hyper-collapsing} terms -- in the
sense that two forking strongly $\mrs$-converging reductions $t \mato
t_1, t \mato t_2$ can always be extended by two strongly
$\mrs$-converging reductions $t_1 \mato t_3, t_2 \mato t_3'$ such that
the resulting terms $t_3, t_3'$ only differ in the hyper-collapsing
subterms they contain.

This result was later generalised by Kennaway et al.\
\cite{kennaway99jflp} to develop an axiomatic theory of
\emph{meaningless terms}. Intuitively, a set of meaningless terms in
this setting consists of terms that are deemed meaningless since, from
a term rewriting perspective, they cannot be distinguished from one
another and they do not contribute any information to any
computation. Kennaway et al.\ capture this by a set of axioms that
characterise a set of meaningless terms. For orthogonal TRSs, one such
set of terms, in fact the least such set, is the set of
\emph{root-active} terms \cite{kennaway99jflp}:
\begin{defi}[root-activeness]
  Let $\calR$ be a TRS and $t \in \iterms$. Then $t$ is called
  \emph{root-active} if for each reduction $t \fto{*} t'$, there is a
  reduction $t' \fto{*} s$ to a redex $s$. The set of all root-active
  terms of $\calR$ is denoted $\rAct[\calR]$ or simply $\rAct$ if
  $\calR$ is clear from the context.
\end{defi}
Intuitively speaking, as the name already suggests, root-active terms
are terms that can be contracted at the root arbitrarily often, e.g.\
the terms $f^\omega$ and $g^\omega$ from Example~\ref{ex:mconfl}.

In this paper we are only interested in this particular set of
meaningless terms. So for the sake of brevity we restrict our
discussion in this section to the set $\rAct$ instead of the original
more general axiomatic treatment by Kennaway et al.\
\cite{kennaway99jflp}.

Since, denotationally, root-active terms cannot be distinguished from
each other it is appropriate to equate them
\cite{kennaway99jflp}. This can be achieved by introducing a new
constant symbol $\bot$ and making each root-active term equal to
$\bot$. By adding rules which enable rewriting root-active terms to
$\bot$, this can be encoded into an existing TRS
\cite{kennaway99jflp}:
\begin{defi}[Böhm extension]
  Let $\calR = (\Sigma, R)$ be a TRS, and $\calU \subseteq \iterms$.
  \begin{enumerate}[label=(\roman*)]
  \item A term $t \in \iterms$ is called a
    \emph{$\bot,\calU$-instance} of a term $s \in
    \ipterms$ if $t$ can be obtained from $s$ by
    replacing each occurrence of $\bot$ in $s$ with some term in
    $\calU$.
  \item $\calU_\bot$ is the set of terms in
    $\ipterms$ that have a $\bot,\calU$-instance in $\calU$.
  \item The \emph{Böhm extension}{} of $\calR$ w.r.t.\ $\calU$ is the
    TRS $\calB_{\calR,\calU} = (\Sigma_\bot,R \cup B)$, where
    \[
    B = \setcom{t \to \bot}{t \in \calU_\bot\setminus\set{\bot}}
    \]
    We write $s \to[\calU,\bot] t$ for a reduction step using a rule
    in $B$. If $\calR$ and $\calU$ are clear from the context, we
    simply write $\calB$ and $\to[\bot]$ instead of
    $\calB_{\calR,\calU}$ and $\to[\calU,\bot]$, respectively.
  \end{enumerate}
  A reduction that is strongly $\mrs$-converging in the Böhm extension
  $\calB$ is called \emph{Böhm-converging}. A term $t$ is called
  \emph{Böhm-reachable} from $s$ if there is a Böhm-converging
  reduction from $s$ to $t$.
\end{defi}

The definition of $\calU_\bot$ is quite subtle and deserves further
attention before we move on. According to the definition, a term $t$
is in $\calU_\bot$ iff the term obtained from $t$ by replacing
occurrences of $\bot$ in $t$ by terms from $\calU$ is also in
$\calU$. More illuminating, however, is the converse view, i.e.\ how
to construct a term in $\calU_\bot$ from a term in $\calU$. First of
all, any term in $\calU$ is also in $\calU_\bot$. Secondly, we may
obtain a term in $\calU_\bot$ by taking a term $t \in \calU$ and
replacing any number of subterms of $t$ that are in $\calU$ by
$\bot$. For B\"ohm extensions, this means that we may contract any
term $t \in \calU$ to $\bot$, even if we already contracted some
proper subterms of $t$ to $\bot$ before.

It is at this point where we, in fact, need the generality of allowing
infinite terms on the left-hand side of rewrite rules: The additional
rules of a Böhm extension allow possibly infinite terms $t \in
\calU_\bot\setminus\set{\bot}$ on the left-hand side.

\begin{rem}[closure under substitution]
  \label{rem:cloSub}
  Note that, for orthogonal TRSs, $\rAct$ is closed under
  substitutions and, hence, so is $\rAct_\bot$
  \cite{kennaway99jflp}. Therefore, whenever $\Cxt{b}{t}
  \to[\rAct,\bot] \Cxt{b}{\bot}$, we can assume that $t \in
  \rAct_\bot$.
\end{rem}

With the additional rules provided by the Böhm extension, we gain
infinitary confluence of orthogonal systems:
\begin{thm}[infinitary confluence of Böhm-converging reductions,
  {\cite{kennaway99jflp}}]
  %[kennaway99jflp, Thm. 2]%
  \label{thr:bohmCR}
  % confluence of Böhm reduction %
  Let $\calR$ be an orthogonal, left-finite TRS. Then the Böhm
  extension $\calB$ of $\calR$ w.r.t.\ $\rAct$ is infinitarily
  confluent, i.e.\ $s_1 \mafrom[\calB] t \mato[\calB] s_2$ implies
  $s_1 \mato[\calB] t' \mafrom[\calB] s_2$.
% TODO: is the restriction to left-finite TRS necessary?
\end{thm}
The lack of confluence for strongly $\mrs$-converging reductions is
resolved in Böhm extensions by allowing (sub-)terms, which where
previously not joinable, to be contracted to $\bot$. Returning to
Example~\ref{ex:mconfl}, we can see that $g^\omega$ and $f^\omega$ can
be rewritten to $\bot$ as both terms are root-active.

In fact, w.r.t.\ Böhm-convergence, every term of an orthogonal TRS has
a normal form:
\begin{thm}[infinitary normalisation of Böhm-converging reductions,
  {\cite{kennaway99jflp}}]
  \label{thr:bohmWn}
  %[kennaway99jflp, Thm. 1]%
  Let $\calR$ be an orthogonal, left-finite TRS. Then the Böhm
  extension $\calB$ of $\calR$ w.r.t.\ $\rAct$ is infinitarily
  normalising, i.e.\ for each term $t$ there is a $\calB$-normal form
  Böhm-reachable from $t$.
% TODO: is the restriction to left-finite TRS necessary?
\end{thm}
This means that each term $t$ of an orthogonal, left-finite TRS
$\calR$ has a unique normal form in $\calB_{\calR, \rAct}$. This
normal form is called the \emph{Böhm tree} of $t$ (w.r.t.\ $\rAct$)
\cite{kennaway99jflp}.

The rest of this paper is concerned with establishing an alternative
to the metric notion of convergence based on the partial order on
terms that is equivalent to the Böhm extension approach.

\section{Partial Order Infinitary Rewriting}
\label{sec:part-order-infin}

In this section we introduce an alternative model of infinitary term
rewriting which uses the partial order on terms to formalise
convergence of transfinite reductions. To this end we will turn to
partial terms which, like in the setting of Böhm extensions, have an
additional constant symbol $\bot$. The result will be a more
fine-grained notion of convergence in which, intuitively speaking, a
reduction can be diverging in some positions but at the same time
converging in other positions. The ``diverging parts'' are then
indicated by a $\bot$-occurrence in the final term of the reduction:
\begin{exa}
  \label{ex:prsConv}
  Consider the TRS consisting of the rules $h(x) \to h(g(x)), b \to
  g(b)$ and the term $t= f(h(a),b)$. In this system, we have the
  reduction
  \[
  S\fcolon f(h(a),b) \to f(h(g(a)),b) \to f(h(g(a)),g(b)) \to
  f(h(g(g(a))),g(b)) \to \dots
  \]
  which alternately contracts the redex in the left and in the right
  argument of $f$.
\end{exa}
The reduction $S$ weakly $\mrs$-converges to the term
$f(h(g^\omega),g^\omega)$. But it does not \emph{strongly}
$\mrs$-converge as the depth at which contractions are performed does
not tend to infinity. However, this does only happen in the left
argument of $f$, not in the other one. Within the partial order model
we will still be able to obtain that $S$ weakly converges to
$f(h(g^\omega),g^\omega)$ but we will also obtain that it strongly
converges to the term $f(\bot,g^\omega)$. That is, we will be able to
identify that the reduction $S$ strongly converges except at position
$\seq{0}$, the first argument of $f$.

\subsection{Partial Order Convergence}
\label{sec:part-order-conv}

In order to formalise continuity and convergence in terms of the
complete semilattice $(\ipterms,\lebot)$ instead of the complete
metric space $(\iterms, \dd)$, we move from the limit of the metric
space to the limit inferior of the complete semilattice:
\begin{defi}[$\prs$-continuity/-convergence]
  Let $\calR = (\Sigma,R)$ be a TRS and $S = (\phi_\iota\fcolon
  t_\iota \to[\pi_\iota] t_{\iota+1})_{\iota<\alpha}$ a non-empty
  reduction in $\calR_\bot = (\Sigma_\bot,R)$. The reduction $S$ is
  called
  \begin{enumerate}[label=(\roman*)]
  \item \emph{weakly $\prs$-continuous} in $\calR$, written $S\fcolon
    t_0 \pwacont[\calR]$, if $\liminf_{\iota \limto \lambda} t_\iota =
    t_\lambda$ for each limit ordinal $\lambda < \alpha$.
  \item \emph{strongly $\prs$-continuous} in $\calR$, written
    $S\fcolon t_0 \pacont[\calR]$, if $\liminf_{\iota \limto \lambda}
    c_\iota = t_\lambda$ for each limit ordinal $\lambda < \alpha$,
    where $c_\iota = \substAtPos{t_\iota}{\pi_\iota}{\bot}$. Each
    $c_\iota$ is called the \emph{context} of the reduction step
    $\phi_\iota$, which we indicate by writing $\phi_\iota\fcolon
    t_\iota \to[c_\iota] t_{\iota+1}$.
  \item \emph{weakly $\prs$-converging} to $t$ in $\calR$, written
    $S\fcolon t_0 \pwato[\calR] t$, if it is weakly $\prs$-continuous
    and $t = \liminf_{\iota \limto \wsuc\alpha} t_\iota$.
  \item \emph{strongly $\prs$-converging} to $t$ in $\calR$, written
    $S\fcolon t_0 \pato[\calR] t$, if it is strongly
    $\prs$-continuous and $S$ is closed with $t=t_{\alpha+1}$
    or $t = \liminf_{\iota \limto \alpha} c_\iota$.
  \end{enumerate}
  Whenever $S\fcolon t_0 \pwato[\calR] t$ or $S\fcolon t_0
  \pato[\calR] t$, we say that $t$ is weakly resp.\ strongly
  \emph{$\prs$-reachable} from $t_0$ in $\calR$. By abuse of notation
  we use $\pwato[\calR]$ and $\pato[\calR]$ as a binary relation to
  indicate weak resp.\ strong $\prs$-reachability. In order to
  indicate the length of $S$ and the TRS $\calR$, we write $S\fcolon
  t_0 \pwto{\alpha}[\calR] t$ resp.\ $S\fcolon t_0 \pto{\alpha}[\calR]
  t$. The empty reduction $\emptyseq$ is considered weakly/strongly
  $\prs$-continuous and $\prs$-convergent for any start and end term,
  i.e.\ $\emptyseq\fcolon t\pato[\calR] t$ for all $t \in \terms$.
\end{defi}

The definitions of weak $\prs$-continuity and weak $\prs$-convergence
are straightforward ``translations'' from the metric setting to the
partial order setting replacing the limit $\lim_{\iota \limto \alpha}$
by the limit inferior $\liminf_{\iota \limto \alpha}$. On the other
hand, the definitions of the strong counterparts seem a bit different
compared to the metric model: Whereas strong $\mrs$-convergence simply
adds a side condition regarding the depth $\len{\pi_\iota}$ of the
reduction steps, strong $\prs$-convergence is defined in a different
way compared to the weak variant. Instead of the terms $t_\iota$ of
the reduction, it considers the contexts $c_\iota =
\substAtPos{t_\iota}{\bot}{\pi_\iota}$. However, one can surmise some
similarity due to the fact that the partial order model of strong
convergence indirectly takes into account the position $\pi_\iota$ of
each reduction step as well. Moreover, for the sake of understanding
the intuition of strong $\prs$-convergence it is better to compare the
contexts $c_\iota$ rather with the glb of two consecutive terms
$t_\iota \glbbot t_{\iota +1}$ instead of the term $t_\iota$
itself. The following proposition allows precisely that.
\begin{prop}[limit inferior of open sequences]
  \label{prop:liminfOpen}
  % limit inferior of open sequences %
  Let $(a_\iota)_{\iota < \lambda}$ be an open sequence in a complete
  semilattice. Then it holds that $\liminf_{\iota \limto \lambda} a_\iota =
  \liminf_{\iota \limto \lambda} (a_\iota \glb a_{\iota+1})$.
\end{prop}
\begin{proof}
  Let $\ol a = \liminf_{\iota \limto \lambda} a_\iota$ and $\oh a =
  \liminf_{\iota \limto \lambda} (a_\iota \glb a_{\iota+1})$. Since
  $a_\iota \glb a_{\iota+1} \le a_\iota$ for each $\iota < \lambda$, we
  have $\oh a \le \ol a$. On the other hand, consider the sets $\ol
  A_\alpha = \setcom{a_\iota}{\alpha \le \iota < \lambda}$ and $\oh
  A_\alpha = \setcom{a_\iota\glb a_{\iota+1}}{\alpha \le \iota <
    \lambda}$ for each $\alpha < \lambda$. Of course, we then have
  $\Glb\ol A_\alpha \le a_\iota$ for all $\alpha\le\iota<\lambda$, and
  thus also $\Glb\ol A_\alpha \le a_\iota \glb a_{\iota+1}$ for all
  $\alpha\le\iota<\lambda$. Hence, $\Glb \ol A_\alpha$ is a lower
  bound of $\oh A_\alpha$ which implies that $\Glb \ol A_\alpha \le
  \Glb \oh A_\alpha$. Consequently, $\ol a \le \oh a$ and, due to the
  antisymmetry of $\le$, we can conclude that $\ol a = \oh a$.
\end{proof}
With this in mind we can replace $\liminf_{\iota \limto \lambda}
t_\iota$ in the definition of weak $\prs$-convergence resp.\
$\prs$-continuity with $\liminf_{\iota \limto \lambda} t_\iota \glbbot
t_{\iota+1}$. From there it is easier to see the intention of moving
from $t_\iota \glbbot t_{\iota+1}$ to the context
$\substAtPos{t_\iota}{\pi_\iota}{\bot}$ in order to model strong
convergence:

What makes the notion of strong $\prs$-convergence (and
$\prs$-continuity) \emph{strong}, similar to the notion of strong
$\mrs$-convergence (resp.\ $\mrs$-continuity), is the choice of taking
the contexts $\substAtPos{t_\iota}{\pi_\iota}{\bot}$ for defining the
limit behaviour of reductions instead of the whole terms
$t_\iota$. The context $\substAtPos{t_\iota}{\pi_\iota}{\bot}$
provides a conservative underapproximation of the shared structure
$t_\iota \glbbot t_{\iota + 1}$ of two consecutive terms $t_\iota$ and
$t_{\iota+1}$ in a reduction step $\phi_\iota\fcolon t_\iota
\to[\pi_\iota] t_{\iota+1}$. More specifically, we have that
$\substAtPos{t_\iota}{\pi_\iota}{\bot} \lebot t_\iota \glbbot t_{\iota
  + 1}$. That is, as in the metric model of strong convergence, the
difference between two consecutive terms is overapproximated by using
the position of the reduction step as an indicator. Likewise, strong
$\prs$-convergence is simply weak $\prs$-convergence w.r.t.\ this
underapproximation of $t_\iota \glbbot t_{\iota+1}$
\cite{bahr10rta}. If this approximation is actually precise, i.e.\
coincides with the actual value, both notions of $\prs$-convergence
coincide.

\begin{rem}
  \label{rem:pcont}
  As for the metric model, also in the partial order model, continuity
  can be defined solely in terms of convergence \cite{bahr10rta}. More
  precisely, we have for each reduction $S = (t_\iota \to
  t_{\iota+1})_{\iota < \alpha}$ that $S$ is weakly $\prs$-continuous
  iff every (open) proper prefix of $\prefix{S}{\beta}$ weakly
  $\prs$-converges to $t_\beta$. Analogously, strong $\prs$-continuity
  can be characterised in terms of strong $\prs$-convergence.  An easy
  consequence of this is that $\prs$-converging reductions are closed
  under concatenation, i.e.\ $S\fcolon s \wato t$, $T\fcolon t \wato
  u$ implies $S\concat T \fcolon s \wato u$ and likewise for strong
  $\prs$-convergence.
\end{rem}

In order to understand the difference between weak and strong
$\prs$-convergence let us look at a simple example:
\begin{exa}
\label{ex:weakVsStrong}
  Consider the TRS with the single rule $f(x,y) \to f(y,x)$. This rule
  induces the following reduction:
  \[
  S\fcolon f(a,f(g(a),g(b))) \to f(a,f(g(b),g(a))) \to
  f(a,f(g(a),g(b))) \to\;\dots
  \]
  $S$ simply alternates between the terms $f(a,f(g(a),g(b)))$ and
  $f(a,f(g(b),g(a)))$ by swapping the arguments of the inner $f$
  occurrence. The reduction is depicted in
  Figure~\ref{fig:pconv}. The picture illustrates the parts of the
  terms that remain \emph{unchanged} and those that remain completely
  \emph{untouched} by the corresponding reduction step by using a
  lighter resp.\ a darker shade of grey. The unchanged part
  corresponds to the glb of the two terms of a reduction step, viz.\
  for the first step
  \[
  f(a,f(g(a),g(b))) \glbbot f(a,f(g(b),g(a))) =
  f(a,f(g(\bot),g(\bot)))
  \]
  By symmetry, the glb of the terms of the second step is the same
  one. It is depicted in Figure~\ref{fig:pconvWeak}. Let $(t_i)_{i <
    \omega}$ be the sequence of terms of the reduction $S$. By
  definition, $S$ weakly $\prs$-converges to $\liminf_{i\limto\omega}
  t_i$. According to Proposition~\ref{prop:liminfOpen}, this is equal
  to $\liminf_{i\limto\omega} (t_i \glbbot t_{i+1})$. Since $t_i \glbbot
  t_{i+1}$ is constantly $f(a,f(g(\bot),g(\bot)))$, the reduction
  sequence weakly $\prs$-converges to $f(a,f(g(\bot),g(\bot)))$.

  Similarly, the part of the term that remains untouched by the
  reduction step corresponds to the context. For the first step, this
  is $f(a,\bot)$. It is depicted in Figure~\ref{fig:pconvStrong}. By
  definition, $S$ strongly $\prs$-converges to
  $\liminf_{i\limto\omega} c_i$ for $(c_i)_{i<\omega}$ the sequence of
  contexts of $S$. As one can see in Figure~\ref{fig:pconv}, the
  context constantly remains $f(a,\bot)$. Hence, $S$ strongly
  $\prs$-converges to $f(a,\bot)$.  The example sequence is a
  particularly simple one as both the glbs $t_i \glbbot t_{i+1}$ and
  the contexts $c_i$ remain stable. In general, this is not necessary,
  of course.
\end{exa}
\begin{figure}
  \centering
  \begin{tikzpicture}[node distance=3cm]
  \node (a) {$f$}%
  child {%
    node (a1) {$a$}%
    child[no edge] { node (a11) {} child { node (a111) {}}}%
  } child {%
    node (a2) {$f$}%
    child {%
      node (a21) {$g$}%
      child {%
        node (a211) {$a$}%
      }%
    } child {%
      node (a22) {$g$}%
      child {%
        node (a221) {$b$}%
      }%
    }%
  };%

  \node (b) [right=of a] {$f$}%
  child {%
    node (b1) {$a$}%
    child[no edge] { node (b11) {} child { node (b111) {}}}%
  } child {%
    node (b2) {$f$}%
    child {%
      node (b21) {$g$}%
      child {%
        node (b211) {$b$}%
      }%
    } child {%
      node (b22) {$g$}%
      child {%
        node (b221) {$a$}%
      }%
    }%
  };%

  \node (c) [right=of b] {$f$}%
  child {%
    node (c1) {$a$}%
    child[no edge] { node (c11) {} child { node (c111) {}}}%
  } child {%
    node (c2) {$f$}%
    child {%
      node (c21) {$g$}%
      child {%
        node (c211) {$a$}%
      }%
    } child {%
      node (c22) {$g$}%
      child {%
        node (c221) {$b$}%
      }%
    }%
  };%
  \foreach \l in {a,b,c} {%
    \foreach \n in {, 1, 11, 111, 2, 21,21, 211, 22, 221} {%
      \node[shape=circle,minimum size=10mm] (\l\n) at (\l\n) {};%
    }%
  }%
    
  \begin{scope}[rounded corners=.3cm]
    \begin{pgfonlayer}{background} %
      \foreach \l in {a,b,c} {%
        \path[fill=termback]%
        (\l.north) -- (\l 1.170) -- (\l 211.south west) -- (\l
        221.south east) -- (\l 22.north east) -- cycle;%
      }%
      \foreach \l in {b,c} {%
        \path[term const]%
        (\l.north) -- (\l 1.170) -- (\l 1.250) -- (\l 21.south west)
        -- (\l 22.south east) -- (\l 22.north east) -- cycle;%
      }%
      \foreach \l in {b,c} {%
        \path[term stable]%
        (\l.north) -- (\l 1.170) -- (\l 1.250) -- (\l 1.330) -- (\l.-10) -- cycle;%
      }%
      \foreach \l in {a,b,c} {%
        \path[term fringe]%
        (\l.north) -- (\l 1.170) -- (\l 211.south west) -- (\l
        221.south east) -- (\l 22.north east) -- cycle;%
      }%
    \end{pgfonlayer}
  \end{scope}

  \node[node distance=2cm,right=of c2]
  (d2) {};
  \node[node distance=1cm,right=of d2]
  (e2) {};
    \draw[thick,loosely dotted]
  (d2) -- (e2);
  \foreach \l in {a,b,c} {%
    \node[red site] (\l 2) at (\l 2) {};%
  }%
  \draw[single step,every edge/.style={bend left=20,draw=black!80}]%
  (a2) edge (b2)%
  (b2) edge (c2)%
  (c2) edge (d2);%
\end{tikzpicture}
\caption{Reduction with stable context.}
\label{fig:pconv}
\end{figure}
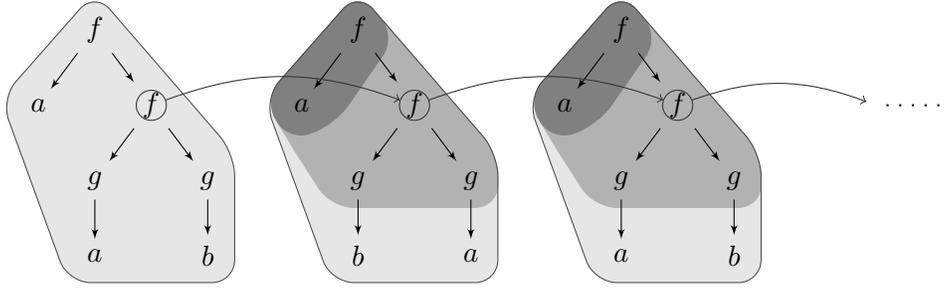

\begin{figure}
  \centering
  \subfloat[Limit w.r.t.\ weak $\prs$-convergence.]{
    \begin{minipage}{.4\linewidth}
      \centering
      \begin{tikzpicture}
        \node (a) {$f$}%
        child {%
          node (a1) {$a$}%
          child[no edge] { node (a11) {} child { node (a111) {}}}%
        } child {%
          node (a2) {$f$}%
          child {%
            node (a21) {$g$}%
            child {%
              node (a211) {$\bot$}%
            }%
          } child {%
            node (a22) {$g$}%
            child {%
              node (a221) {$\bot$}%
            }%
          }%
        };%

        \foreach \l in {a} {%
          \foreach \n in {, 1, 11, 111, 2, 21,21, 211, 22, 221} {%
            \node[shape=circle,minimum size=10mm] (\l\n) at (\l\n)
            {};%
          }%
        }%
        \begin{scope}[rounded corners=.3cm]
          \begin{pgfonlayer}{background} %
            \foreach \l in {a} {%
              \path[term const]%
              (\l.north) -- (\l 1.170) -- (\l 1.250) -- (\l 21.south
              west) -- (\l 22.south east) -- (\l 22.north east) --
              cycle;%
            }%
          \end{pgfonlayer}
        \end{scope}
      \end{tikzpicture}
    \end{minipage}
    \label{fig:pconvWeak}}
  \hspace{2cm}
  \subfloat[Limit w.r.t.\ strong $\prs$-convergence.]{
    \begin{minipage}{.4\linewidth}
      \centering
      \begin{tikzpicture}
        \node (a) {$f$}%
        child {%
          node (a1) {$a$}%
          child[no edge] { node (a11) {} child { node (a111) {}}}%
        } child {%
          node (a2) {$\bot$}%
          child[no edge] {%
            node (a21) {}%
            child {%
              node (a211) {}%
            }%
          } child[no edge] {%
            node (a22) {}%
            child {%
              node (a221) {}%
            }%
          }%
        };%
        \foreach \l in {a} {%
          \foreach \n in {, 1, 11, 111, 2, 21,21, 211, 22, 221} {%
            \node[shape=circle,minimum size=10mm] (\l\n) at (\l\n)
            {};%
          }%
        }%
        \begin{scope}[rounded corners=.3cm]
          \begin{pgfonlayer}{background} %
            \foreach \l in {a} {%
              \path[term stable]%
              (\l.north) -- (\l 1.170) -- (\l 1.250) -- (\l 1.330) --
              (\l.-10) -- cycle;%
            }%
          \end{pgfonlayer}
        \end{scope}
      \end{tikzpicture}
    \end{minipage}

    \label{fig:pconvStrong}}
  \caption{Limits of a $\prs$-converging reduction.}
  \label{fig:pconvLimits}
\end{figure}
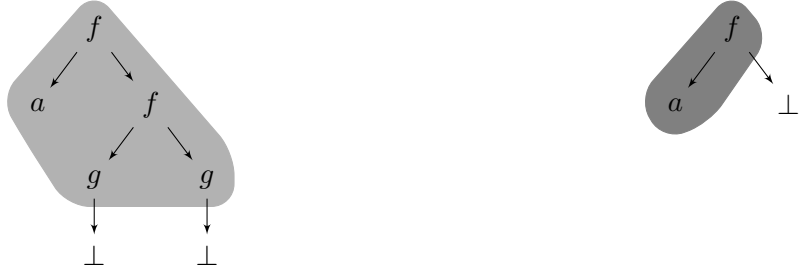

%%% Local Variables: 
%%% mode: latex
%%% TeX-master: "paper"
%%% End: 

One can clearly see from the definition that, as for their metric
counterparts, weak resp.\ strong $\prs$-convergence implies weak
resp.\ strong $\prs$-continuity. In contrast to the metric model,
however, also the converse implication holds! Since the partial order
$\lebot$ on partial terms forms a complete semilattice, the limit
inferior is defined for any non-empty sequence of partial
terms. Hence, any weakly resp.\ strongly $\prs$-continuous reduction
is also weakly resp.\ strongly $\prs$-convergent. This is a major
difference to $\mrs$-convergence/-continuity. Nevertheless,
$\prs$-convergence constitutes a meaningful notion of convergence: The
final term of a $\prs$-convergent reduction contains a $\bot$ subterm
at each position at which the reduction is ``locally diverging'' as we
have seen in Example~\ref{ex:prsConv} and
Example~\ref{ex:weakVsStrong}. In fact, as we will show in
Section~\ref{sec:comp-mrs-conv}, whenever there are no '$\bot$'s
involved, i.e.\ if there is no ``local divergence'',
$\mrs$-convergence and $\prs$-convergence coincide -- both in the weak
and the strong variant.

Recall that strong $\mrs$-continuity resp.\ $\mrs$-convergence implies
weak $\mrs$-continuity resp.\ $\mrs$-convergence. This is not the case
in the partial order setting. The reason for this is that strong
$\prs$-convergence resp.\ $\prs$-continuity is defined differently
compared to its weak variant. It uses the contexts instead of the
terms in the reduction, whereas in the metric setting the strong
notion of convergence is a mere restriction of the weak counterpart as
we have observed earlier.

\begin{exa}
  \label{ex:strWeakPconv}
  Consider the TRS consisting of the rules $\rho_1\fcolon h(x) \to h(g(x)),
  \rho_2\fcolon f(x) \to g(x)$ and the reductions
  \[
  S\fcolon f(h(a)) \to[\rho_1] f(h(g(a))) \to[\rho_1] f(h(g(g(a)))) \to[\rho_1] \dots
  \quad \text{ and } \quad T\fcolon f(\bot) \to[\rho_2] g(\bot)
  \]
  Then the reduction
  \[
  S\concat T\fcolon f(h(a)) \to[\rho_1] f(h(g(a))) \to[\rho_1]
  f(h(g(g(a)))) \to[\rho_1]
  \dots \; f(\bot) \to[\rho_2] g(\bot)
  \]
  is clearly both strongly $\prs$-continuous and -convergent. On the
  other hand it is neither weakly $\prs$-continuous nor -convergent
  for the simple fact that $S$ does not weakly $\prs$-converge to
  $f(\bot)$ but to $f(h(g^\omega))$.
\end{exa}

Nevertheless, by observing that $\liminf_{\iota \limto \alpha} c_\iota
\lebot \liminf_{\iota \limto \alpha} t_\iota$ since $c_\iota \lebot
t_\iota$ for each $\iota < \alpha$, we obtain the following weaker
relation between weak and strong $\prs$-convergence:
\begin{prop}
  Let $\calR$ be a left-linear TRS with $s \pato[\calR] t$. Then there
  is a term $t' \gebot t$ with $s \pwato[\calR] t'$.
\end{prop}
\proof
  Let $S = (\phi_\iota\fcolon t_\iota \to[\rho_\iota]
  t_{\iota+1})_{\iota<\alpha}$ be a reduction strongly
  $\prs$-converging to $t_\alpha$. By induction we construct for each
  prefix $\prefix{S}{\beta}$ of $S$ a reduction $S'_\beta =
  (\phi'_\iota\fcolon t'_\iota \to[\rho_\iota]
  t'_{\iota+1})_{\iota<\beta}$ weakly $\prs$-converging to a term
  $t'_\beta$ such that $t_\iota \lebot t'_\iota$ for each $\iota \le
  \alpha$. The proposition then follows from the case where $\beta =
  \alpha$.

  The case $\beta = 0$ is trivial. If $\beta = \gamma + 1$, then by
  induction hypothesis we have a reduction $S'_\gamma\fcolon t'_0
  \pwato[\calR] t'_\gamma$. Since $t_\gamma \lebot t'_\gamma$ and $
  t_\gamma$ is a $\rho_\gamma$-redex, also $t'_\gamma$ is a
  $\rho_\gamma$-redex due to the left-linearity of $\calR$. Hence,
  there is a reduction step $\phi'_\gamma\fcolon t'_\gamma \to
  t'_\beta$. One can easily see that then $t_\beta \lebot
  t'_\beta$. Hence, $S'_\beta = S'_\gamma \concat \seq{\phi'_\gamma}$
  satisfies desired conditions.

  If $\beta$ is a limit ordinal, we can apply the induction hypothesis
  to obtain for each $\gamma < \beta$ a reduction $S'_\gamma =
  (\phi'_\iota\fcolon t'_\iota \to[\rho_\iota]
  t'_{\iota+1})_{\iota<\gamma}$ that weakly $\prs$-converges to
  $t'_\gamma \gebot t_\gamma$. Hence, according to
  Remark~\ref{rem:pcont}, $S'_\beta = (\phi'_\iota\fcolon t'_\iota
  \to[\rho_\iota] t'_{\iota+1})_{\iota<\beta}$ is weakly
  $\prs$-continuous. Therefore, we obtain that $S'_\beta$ weakly
  $\prs$-converges to $t'_\beta = \liminf_{\iota \limto \beta}
  t'_\iota$. Moreover, since $c_\iota \lebot t_\iota$ and $t_\iota \lebot
  t'_\iota$ for each $\iota < \beta$, we can conclude that 
  \[
  t_\beta = \liminf_{\iota \limto \beta} c_\iota \lebot \liminf_{\iota
    \limto \beta} t_\iota \lebot \liminf_{\iota \limto \beta} t'_\iota =
  t'_\beta.
  \eqno{\qEd}\]

And indeed, returning to Example~\ref{ex:strWeakPconv}, we can see
that there is a reduction
\[
f(h(a)) \to[\rho_1] f(h(g(a))) \to[\rho_1] f(h(g(g(a)))) \to[\rho_1]
\dots \; f(h(g^\omega)) \to[\rho_2] g(h(g^\omega))
\]
that, starting from $f(h(a))$, weakly $\prs$-converges to
$g(h(g^\omega))$ which is strictly larger than $g(\bot)$.

A simple example shows that left-linearity is crucial for the above
proposition:
\begin{exa}
  Let $\calR$ be a TRS consisting of the rules
  \[
  \rho_1\fcolon a\to a,\quad \rho_2\fcolon b\to b, \quad \rho_3\fcolon
  f(x,x)\to c.
  \]
  We then get the strongly $\prs$-converging reduction
  \[
  f(a,b) \to[\rho_1] f(a,b) \to[\rho_2] f(a,b) \to[\rho_1] f(a,b) \to[\rho_2]\dots\;
  f(\bot,\bot) \to[\rho_3] c
  \]
  Yet, there is no reduction in $\calR$ that, starting from $f(a,b)$,
  weakly $\prs$-converges to $c$.
\end{exa}

\subsection{Strong \texorpdfstring{$\prs$}{p}-Convergence}
\label{sec:strong-prs-conv}

In this paper we are mainly focused on the strong notion of
convergence. To this end, the rest of this section will be concerned
exclusively with strong $\prs$-convergence. We will, however, revisit
weak $\prs$-convergence in Section~\ref{sec:comp-mrs-conv} when
comparing it to weak $\mrs$-convergence.

Note that in the partial order model we have to consider reductions
over the extended signature $\Sigma_\bot$, i.e.\ reductions containing
partial terms. Thus, from now on, we assume reductions in a TRS over
$\Sigma$ to be implicitly over $\Sigma_\bot$. When we want to make it
explicit that a reduction $S$ contains only total terms, we say that
$S$ is \emph{total}. When we say that a strongly $\prs$-convergent
reduction $S\fcolon s \pato t$ is total, we mean that both the
reduction $S$ and the final term $t$ are total.\footnote{Note that if
  $S$ is open, the final term $t$ is not explicitly contained in
  $S$. Hence, the totality of $S$ does not necessarily imply the
  totality of $t$.}

In order to understand the behaviour strong $\prs$-convergence, we
need to look at how the lub and glb of a set of terms looks like. The
following two lemmas provide some insight.
\begin{lem}[lub of terms]
  \label{lem:lubbot}
  For each $T \subseteq \ipterms$ and $t = \Lubbot T$, the following
  holds
  \begin{enumerate}[label=(\roman*)]
  \item $\pos{t} = \bigcup_{s\in T} \pos{s}$
    \label{item:lubbotI}
  \item $t(\pi) = f$ iff there is some $s \in T$ with $s(\pi) = f$ \quad for
    each $f \in \Sigma\cup\calV$, and position $\pi$.
  \label{item:lubbotII}
  \end{enumerate}
\end{lem}
\begin{proof}
  Clause~(\ref{item:lubbotI}) follows straightforwardly from
  clause~(\ref{item:lubbotII}). The ``if'' direction of
  (\ref{item:lubbotII}) follows from the fact that if $s \in T$, then
  $s\lebot t$ and, therefore, $s(\pi) = f$ implies $t(\pi) = f$. For
  the ``only if'' direction assume that no $s \in T$ satisfies $s(\pi)
  = f$. Since, $s\lebot t$ for each $s\in T$, we have $\pi \nin
  \posNonBot{s}$ for each $s \in T$. But then $t' =
  \substAtPos{t}{\pi}{\bot}$ is an upper bound of $T$ with $t' \lbot
  t$. This contradicts the assumption that $t$ is the least upper
  bound of $T$.
\end{proof}

\begin{lem}[glb of terms]
  \label{lem:glbbot}
  Let $T \subseteq \ipterms$ and $P$ a set of positions closed under
  prefixes such that all terms in $T$ coincide in all positions in
  $P$, i.e.\ $s(\pi) = t(\pi)$ for all $\pi \in P$ and $s,t \in
  T$. Then the glb $\Glbbot T$ also coincides with all terms in $T$ in
  all positions in $P$.
\end{lem}
\begin{proof}
  Construct a term $s$ such that it coincides with all terms in $T$ in
  all positions in $P$ and has $\bot$ at all other positions. That is,
  given an arbitrary term $t \in T$, we define $s$ as the unique term
  with $s(\pi) = t(\pi)$ for all $\pi \in P$, and $s(\pi\concat\seq i)
  = \bot$ for all $\pi \in P$ with $\pi\concat \seq i \in \pos t
  \setminus P$. Then $s$ is a lower bound of $T$. By construction, $s$
  coincides with all terms in $T$ in all positions in $P$. Since $s
  \lebot \Glbbot T$, this property carries over to $\Glbbot T$.
\end{proof}

Following the two lemmas above, we can observe that -- intuitively
speaking -- the limit inferior $\liminf_{\iota \limto \alpha} t_\iota$
of a sequence of terms is the term that contains those parts that
become \emph{eventually stable} in the sequence. Remaining holes in
the term structure are filled with '$\bot$'s. Let us see what this
means for strongly $\prs$-converging reductions:
\begin{lem}[non-$\bot$ symbols in open reductions]
  \label{lem:nonBotLimRed}
  % non-$\bot$ symbols in open reductions %
  Let $\calR = (\Sigma,R)$ be a TRS and $S\fcolon s
  \pto{\lambda}[\calR] t$ an open reduction with $S = (t_\iota
  \to[\pi_\iota,c_\iota] t_{\iota + 1})_{\iota<\lambda}$. Then the
  following statements are equivalent for all positions $\pi$:
  \begin{enumerate}[label=(\alph*)]
  \item $t(\pi)\neq \bot$.
    \label{item:nonBotLimRed1}
  \item there is some $\alpha < \lambda$ such that $c_\iota(\pi) =
    t(\pi) \neq \bot$ for all $\alpha \le \iota < \lambda$.
    \label{item:nonBotLimRed2}
  \item there is some $\alpha < \lambda$ such that $t_\alpha(\pi) =
    t(\pi) \neq \bot$ and $\pi_\iota \not\le \pi$ for all $\alpha \le
    \iota < \lambda$.
    \label{item:nonBotLimRed3}
  \item there is some $\alpha < \lambda$ such that $t_\alpha(\pi)\neq
    \bot$ and $\pi_\iota \not\le \pi$ for all $\alpha \le \iota <
    \lambda$.
    \label{item:nonBotLimRed4}
  \end{enumerate}
\end{lem}
\begin{proof}
  \def\claima{(\ref{item:nonBotLimRed1})}
  \def\claimb{(\ref{item:nonBotLimRed2})}
  \def\claimc{(\ref{item:nonBotLimRed3})}
  \def\claimd{(\ref{item:nonBotLimRed4})}
  At first consider the implication from \claima{} to \claimb{}. To
  this end, let $t(\pi)\neq \bot$ and $s_\gamma = \Glbbot_{\gamma \le
    \iota < \lambda} c_\iota$ for each $\gamma < \lambda$. Note that
  then $t = \Lubbot_{\gamma < \lambda} s_\gamma$. Applying
  Lemma~\ref{lem:lubbot} yields that there is some $\alpha < \lambda$
  such that $s_{\alpha}(\pi) = t(\pi)$. Moreover, for each $\alpha \le
  \iota < \lambda$, we have $s_\alpha \lebot c_\iota$ and, therefore,
  $s_\alpha(\pi) = c_\iota(\pi)$. Consequently, we obtain
  $c_\iota(\pi) = t(\pi)$ for all $\alpha \le \iota < \lambda$.

  Next consider the implication from \claimb{} to \claimc{}. Let
  $\alpha < \lambda$ be such that $c_\iota(\pi) = t(\pi) \neq \bot$
  for all $\alpha \le \iota < \lambda$. Recall that $c_\iota =
  \substAtPos{t_\iota}{\pi_\iota}{\bot}$ for all $\iota <
  \lambda$. Hence, the fact that $c_\iota(\pi)\neq\bot$ for all
  $\alpha \le \iota < \lambda$ implies that $t_\alpha(\pi) =
  c_\alpha(\pi)$ and that $\pi_\iota \not\le \pi$ for all $\alpha \le
  \iota < \lambda$. Since $c_\alpha(\pi) = t(\pi) \neq \bot$, we also
  have $t_\alpha(\pi) = t(\pi) \neq \bot$.

  The implication from \claimc{} to \claimd{} is trivial.

  Finally, consider the implication from \claimd{} to \claima{}. For
  this purpose, let $\alpha < \lambda$ be such that (1) $\pi \in
  \posNonBot{t_\alpha}$ and (2) $\pi_\iota \not\le \pi$ for all
  $\alpha \le \iota < \lambda$. Consider the set $P$ consisting of all
  positions in $t_\alpha$ that are prefixes of $\pi$. $P$ is obviously
  closed under prefixes and, because of (2), all terms in the set $T =
  \setcom{c_\iota}{\alpha\le \iota <\lambda}$ coincide in all
  positions in $P$. According to Lemma~\ref{lem:glbbot}, also
  $s_\alpha = \Glbbot T$ coincides with all terms in $T$ in all
  positions in $P$. Since $\pi \in P$ and $c_\alpha \in T$, we thereby
  obtain that $c_\alpha(\pi) = s_\alpha(\pi)$. As we also have
  $t_\alpha(\pi) = c_\alpha(\pi)$, due to (2), and $\pi \in
  \posNonBot{t_\alpha}$ we can infer that $\pi \in
  \posNonBot{s_\alpha}$. Since $s_\alpha \lebot t$, we can then
  conclude $\pi \in \posNonBot{t}$.
\end{proof}

The above lemma is central for dealing with strongly $\prs$-convergent
reductions. It also reveals how the final term of a strongly
$\prs$-convergent reduction is constructed. According to the equality
of (\ref{item:nonBotLimRed1}) and (\ref{item:nonBotLimRed3}), the
final term has the non-$\bot$ symbol $f$ at some position $\pi$ iff
some term $t_\alpha$ in the reduction also had this symbol $f$ at this
position $\pi$ and no reduction after that term occurred at $\pi$ or
above. In this way, the final outcome of a strongly $\prs$-convergent
reduction consists of precisely those parts of the intermediate terms
which become \emph{eventually persistent} during the reduction, i.e.\
are from some point on not subjected to contraction any more.

Now we turn to a characterisation of the parts that are not included
in the final outcome of a strongly $\prs$-convergent reduction, i.e.\
those that do not become persistent. These parts are either omitted or
filled by the placeholder $\bot$. We will call these positions
\emph{volatile}:
\begin{defi}[volatility]
  \label{def:volatile}
  % volatile position %
  Let $\calR$ be a TRS and $S = (t_\iota \to[\pi_\iota] t_{\iota +
    1})_{\iota < \lambda}$ an open $\prs$-converging reduction in
  $\calR$. A position $\pi$ is said to be \emph{volatile} in $S$ if,
  for each ordinal $\beta < \lambda$, there is some $\beta \le \gamma
  < \lambda$ such that $\pi_\gamma = \pi$. If $\pi$ is volatile in $S$
  and no proper prefix of $\pi$ is volatile in $S$, then $\pi$ is
  called \emph{outermost-volatile}.
\end{defi}
In Example~\ref{ex:prsConv} the position $\seq{0}$ is
outermost-volatile in the reduction $S$.

\begin{exa}[volatile positions]
  \label{ex:volPos}
  Consider the TRS $\calR$ consisting of the rules
  \[
  \rho_1\fcolon h(x) \to g(x),\qquad \rho_2\fcolon s(g(x)) \to
  s(h(s(x)))
  \]
  $\calR$ admits the following reduction $S$ of length $\omega$:
  \begin{align*}
    S\fcolon f(s(0),s(h(0))) &\to[\rho_1] f(s(0),s(g(0))) \to[\rho_2]
    f(s(0),s(h(s(0))))
    \\&\to[\rho_1] f(s(0),s(g(s(0)))) \to[\rho_2] f(s(0),s(h(s(s(0)))))
  \end{align*}
  The reduction $S$ $\prs$-converges to $f(s(0),\bot)$, i.e.\ we have
  $S\fcolon f(s(0),s(h(0)))
  \pto\omega[\calR]f(s(0),\bot)$. Figure~\ref{fig:volPos} illustrates
  the reduction indicating the position of each reduction step by two
  circles and a reduction arrow in between. One can clearly see that
  both $\pi_1 = \seq{1}$ and $\pi_2 = \seq{1,0}$ are volatile in
  $S$. Again and again reductions take place at $\pi_1$ and
  $\pi_2$. Since these are the only volatile positions and $\pi_1$ is
  a prefix of $\pi_2$, we have that $\pi_1$ is an outermost-volatile
  position in $S$.
\end{exa}
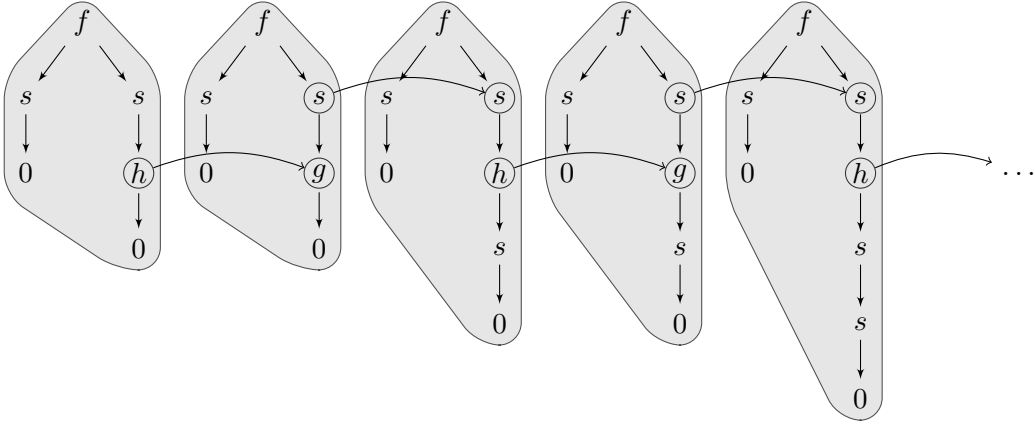
\begin{figure}
  \centering
  \begin{tikzpicture}[node distance=1.9cm]
    \node[alias=a_1] (a) {$f$}%
    child {%
      node[alias=a_2] (a1) {$s$}%
      child {%
        node[alias=a_3] (a11) {$0$}%
      }%
    } child {%
      node[alias=a_5] (a2) {$s$}%
      child {%
        node (a21) {$h$}%
        child {%
          node[alias=a_4] (a211) {$0$}%
        }%
      }%
    };%
    \node [right=of a, alias=b_1] (b) {$f$}%
    child {%
      node[alias=b_2] (b1) {$s$}%
      child {%
        node[alias=b_3] (b11) {$0$}%
      }%
    } child {%
      node[alias=b_5] (b2) {$s$}%
      child {%
        node (b21) {$g$}%
        child {%
          node[alias=b_4] (b211) {$0$}%
        }%
      }%
    };%
    \node [right=of b,alias=c_1] (c) {$f$}%
    child {%
      node[alias=c_2] (c1) {$s$}%
      child {%
        node[alias=c_3] (c11) {$0$}%
      }%
    } child {%
      node[alias=c_5] (c2) {$s$}%
      child {%
        node (c21) {$h$}%
        child {%
          node (c211) {$s$}%
          child {%
            node[alias=c_4] (c2111) {$0$}%
          }%
        }%
      }%
    };%
    \node [right=of c,alias=d_1] (d) {$f$}%
    child {%
      node[alias=d_2] (d1) {$s$}%
      child {%
        node[alias=d_3] (d11) {$0$}%
      }%
    } child {%
      node[alias=d_5] (d2) {$s$}%
      child {%
        node (d21) {$g$}%
        child {%
          node (d211) {$s$}%
          child {%
            node[alias=d_4] (d2111) {$0$}%
          }%
        }%
      }%
    };%
    \node [right=of d,alias=e_1] (e) {$f$}%
    child {%
      node[alias=e_2] (e1) {$s$}%
      child {%
        node[alias=e_3] (e11) {$0$}%
      }%
    } child {%
      node[alias=e_5] (e2) {$s$}%
      child {%
        node (e21) {$h$}%
        child {%
          node (e211) {$s$}%
          child {%
            node (e2111) {$s$}%
            child {%
              node[alias=e_4] (e21111) {$0$}%
            }%
          }%
        }%
      }%
    };%
    \node [node distance=1.5cm,right=of e21] (f21) {$\dots$};%

  \begin{scope}[every node/.style={red site}]
    \foreach \n in {a21,b21,b2,c2,c21,d21,d2,e2,e21} {%
      \node (\n) at (\n) {};%
    }%

    \draw[single step, every edge/.style={draw,bend left=20}]%
    (a21) edge (b21)%
    (b2) edge (c2)%
    (c21) edge (d21)%
    (d2) edge (e2)%
    (e21) edge (f21);%
  \end{scope}
  
  \foreach \l in {a_,b_,c_,d_,e_} {%
    \foreach \n in {1,2,3,4,5} {%
      \node[shape=circle,minimum size=8mm] (\l\n) at (\l\n) {};%
    }%
  }%
    
  \begin{scope}[rounded corners=.3cm]
    \begin{pgfonlayer}{background} %
      \foreach \l in {a_,b_,c_,d_,e_} {%
        \path[term back]%
        (\l 1.north) -- (\l 2.north west) -- (\l 3.south west) -- (\l
        4.south west) -- (\l 4.south east) -- (\l 5.north east) --
        cycle;%
      }%
    \end{pgfonlayer}
  \end{scope}

\end{tikzpicture}

\caption{Reduction with two nested volatile positions.}
\label{fig:volPos}
\end{figure}
%%% Local Variables: 
%%% mode: latex
%%% TeX-master: "paper"
%%% End: 
%
As we shall see later in Section~\ref{sec:relation-bohm-trees},
volatility is closely related to root-active terms: if a reduction has
a volatile position $\pi$, then we find a term in the reduction with a
root-active subterm at $\pi$. Conversely, from each root-active term
starts a reduction with volatile position $\emptyseq$ (cf.\
Proposition~\ref{prop:rootAct}). This connection between volatility
and root-activeness is the cornerstone of the correspondence between
$\prs$-convergence and B\"ohm-convergence that we prove in
Section~\ref{sec:relation-bohm-trees}.

The following lemma shows that $\bot$ symbols are produced precisely
at outermost-volatile positions in open reduction.
\begin{lem}[$\bot$ subterms in open reductions]
  \label{lem:botLimRed}
  % $\bot$-symbols in open reductions %
  Let $S = (t_\iota \to[\pi_\iota] t_{\iota + 1})_{\iota < \alpha}$
  an open reduction $\prs$-converging to $t_\alpha$ in some
  TRS. Then, for every position $\pi$, we have the following:
  \begin{enumerate}[label=(\roman*)]
  \item If $\pi$ is volatile in $S$, then $\pi \nin
    \posNonBot{t_\alpha}$.
    \label{item:botLimRed1}
  \item $t_\alpha(\pi) = \bot$ iff
    \begin{enumerate}[label=(\alph*)]
    \item $\pi$ is outermost-volatile in $S$, or
      \label{item:botLimRed2a}
    \item there is some $\beta < \alpha$ such that $t_\beta(\pi) = \bot$
      and $\pi_\iota \not\le \pi$ for all $\beta \le \iota < \alpha$.
      \label{item:botLimRed2b}
    \end{enumerate}
    \label{item:botLimRed2}
  \item Let $t_\iota$ be total for all $\iota < \alpha$. Then
    $t_\alpha(\pi) = \bot$ iff $\pi$ is outermost-volatile in $S$.
    \label{item:botLimRed3}
  \end{enumerate}
\end{lem}
\begin{proof}
%  \todo{Can this proof be shortened?} 
  \def\itemI{(\ref{item:botLimRed1})}
  \def\itemII{(\ref{item:botLimRed2})}
  \def\itemIIa{(\ref{item:botLimRed2a})}
  \def\itemIIb{(\ref{item:botLimRed2b})}
  \def\itemIII{(\ref{item:botLimRed3})}
  \itemI{} This follows from Lemma~\ref{lem:nonBotLimRed}, in particular
  the equivalence of (\ref{item:nonBotLimRed1}) and
  (\ref{item:nonBotLimRed3}).

  \itemII{} At first consider the ``only if'' direction. To this end,
  suppose that $t_\alpha(\pi) = \bot$. In order to show that then
  \itemIIa{} or \itemIIb{} holds, we will prove that \itemIIb{} must
  hold true whenever \itemIIa{} does not hold. For this purpose, we
  assume that $\pi$ is not outermost-volatile in $S$. Note that no
  proper prefix $\pi'$ of $\pi$ can be volatile in $S$ as this would
  imply, according to clause \itemI{}, that $\pi' \nin
  \posNonBot{t_\alpha}$ and, therefore, $\pi \nin
  \pos{t_\alpha}$. Hence, $\pi$ is also not volatile in $S$. In sum,
  no prefix of $\pi$ is volatile in $S$. Consequently, there is an
  upper bound $\beta < \alpha$ on the indices of reduction steps
  taking place at $\pi$ or above. But then $t_\beta(\pi) = \bot$ since
  otherwise Lemma~\ref{lem:nonBotLimRed} would imply that
  $t_\alpha(\pi) \neq \bot$. This shows that \itemIIb{} holds.
  
  For the converse direction, we will show that both \itemIIa{} and
  \itemIIb{} independently imply that $t_\alpha(\pi) = \bot$:

  \itemIIa{} Let $\pi$ be outermost-volatile in $S$. By clause
  \itemI{}, this implies $\pi\nin\posNonBot{t_\alpha}$. Hence, it
  remains to be shown that $\pi \in \pos{t_\alpha}$. If $\pi =
  \emptyseq$, then this is trivial. Otherwise, $\pi$ is of the form
  $\pi'\concat i$. Since all proper prefixes of $\pi$ are not
  volatile, there is some $\beta < \alpha$ such that $\pi_\beta = \pi$
  and $\pi_\iota \not\le \pi'$ for all $\beta \le \iota <
  \alpha$. This implies that $\pi \in \pos{t_\beta}$. Hence,
  $t_\beta(\pi') = f$ is a symbol having an arity of at least
  $i+1$. Consequently, according to Lemma~\ref{lem:nonBotLimRed}, also
  $t_\alpha(\pi') = f$. Since $f$'s arity is at least $i+1$, also $\pi
  = \pi'\concat i \in \pos{t_\alpha}$.

  \itemIIb{} Let $\beta < \alpha$ such that $t_\beta(\pi) = \bot$ and
  $\pi_\iota \not\le \pi$ for all $\beta \le \iota <
  \alpha$. According to Proposition~\ref{prop:liminfSuffix}, we have
  that $t_\alpha = \Lubbot_{\beta\le \gamma < \alpha} \Glbbot_{\gamma
    \le \iota < \alpha} c_\iota$. Define $s_\gamma = \Glbbot_{\gamma
    \le \iota < \alpha} c_\iota$ for each $\gamma < \alpha$. Since
  from $\beta$ onwards no reduction takes place at $\pi$ or above, it
  holds that all $c_\iota$, for $\beta \le \iota < \alpha$, coincide
  in all prefixes of $\pi$. By Lemma~\ref{lem:glbbot}, this also holds
  for all $s_\iota$ and $c_\iota$ with $\beta \le \iota <
  \alpha$. Since $c_\beta(\pi) = t_\beta(\pi) = \bot$, this means that
  $s_\iota(\pi) = \bot$ for all $\beta \le \iota < \alpha$. Recall
  that $t_\alpha = \Lubbot_{\beta\le \gamma < \alpha}
  s_\gamma$. Hence, according to Corollary~\ref{lem:lubbot},
  we can conclude that $t_\alpha(\pi) = \bot$.

  \itemIII{} is a special case of \itemII{}: If each $t_\iota$, $\iota
  < \alpha$, is total, then \itemIIb{} cannot be true.
\end{proof}

Clause~(\ref{item:botLimRed2}) shows that a $\bot$ subterm in the
final term can only have its origin either in a preceding term which
already contains this $\bot$ which then becomes stable, or in an
outermost-volatile position. That is, it is exactly the
outermost-volatile positions that generate '$\bot$'s.

We can apply this lemma to Example~\ref{ex:volPos}: As we have seen,
the position $\pi_1 = \seq{1}$ is outermost-volatile in the reduction
$S$ mentioned in the example. Hence, $S$ strongly $\prs$-converges to
a term that has, according to Lemma~\ref{lem:botLimRed}, the symbol
$\bot$ at position $\pi_1$. That is, $S$ strongly $\prs$-converges to
$f(s(0),\bot)$.

This characterisation of the final outcome of a $\prs$-converging
reduction clearly shows that the partial order model captures the
intuition of strong convergence in transfinite reductions even though
it allows that every continuous reduction is also convergent: The
final outcome only represents the parts of the reduction that
\emph{are} converging. Locally diverging parts are cut off and
replaced by $\bot$.

In fact, the absence of such local divergence, or volatility, as we
call it here, is equivalent to the absence of $\bot$:
\begin{lem}[total reductions]
  \label{lem:totalRed}
  % total reductions %
  Let $\calR$ be a TRS, $s$ a total term in $\calR$, and $S\fcolon s
  \pato[\calR] t$. $S\fcolon s \pato[\calR] t$ is total iff no prefix of $S$
  has a volatile position.
\end{lem}
\begin{proof}
  The ``only if'' direction follows straightforwardly from
  Lemma~\ref{lem:botLimRed}.

  We prove the ``if'' direction by induction on the length of $S$. If
  $\len{S} = 0$, then the totality of $S$ follows from the assumption of
  $s$ being total. If $\len{S}$ is a successor ordinal, then the
  totality of $S$ follows from the induction hypothesis since single
  reduction steps preserve totality. If $\len{S}$ is a limit ordinal,
  then the totality of $S$ follows from the induction hypothesis using
  Lemma~\ref{lem:botLimRed}.
\end{proof}

Moreover, as we shall show in the next section, if local divergences
are excluded, i.e.\ if total reductions are considered, both the
metric model and the partial order model coincide.

\section{Comparing \texorpdfstring{$\mrs$}{m}-Convergence and 
  \texorpdfstring{$\prs$}{p}-Convergence}
\label{sec:comp-mrs-conv}

In this section we want to compare the metric and the partial order
model of convergence. In particular, we shall show that the partial
order model is only a conservative extension of the metric model: If
we only consider total reductions, i.e.\ reductions over terms in
$\iterms$, then $\mrs$-convergence and $\prs$-convergence coincide
both in their weak and strong variant.

The first and rather trivial observation to this effect is that
already on the level of single reduction steps the partial order model
conservatively extends the metric model:
\begin{fact}
  \label{fact:step}
  Let $\calR = (\Sigma,R)$ be a TRS, $\calR_\bot = (\Sigma_\bot, R)$,
  and $s, t \in \ipterms$. Then we have 

  \[s \to[\calR,\pi] t \quad \text{ iff } \quad s \to[\calR_\bot,\pi]
  t \text{ and } s \text{ is total}.
  \]
\end{fact}

The next step is to establish that the underlying structures that are
used to formalise convergence exhibit this behaviour as well. That is,
the limit inferior in the complete semilattice $(\ipterms,\lebot)$ is
conservative extension of the limit in the complete metric space
$(\iterms,\dd)$. More precisely, we want to have that for a sequence
$(t_\iota)_{\iota<\alpha}$ in $\iterms$ 
\begin{gather*}
  \liminf_{\iota \limto \alpha} t_\iota = \lim_{\iota \limto \alpha}
  t_\iota \qquad \text{ whenever}\quad
  \begin{aligned}
    &\lim_{\iota \limto \alpha} t_\iota \text{ is defined, or}\\
    &\liminf_{\iota \limto \alpha} t_\iota \text{ is a total term.}
  \end{aligned}
\end{gather*}
Note that, as a corollary, the above property implies that
$\lim_{\iota \limto \alpha} t_\iota$ is defined iff $\liminf_{\iota
  \limto \alpha} t_\iota$ is a total term.  In
Section~\ref{sec:limit-infer-cons} we shall establish the above
property. This result is then used in Section~\ref{sec:prs-conv-cons}
in order to show the desired property that $\prs$-convergence is a
conservative extension of $\mrs$-convergence in both their respective
weak and strong variant.

\subsection{Complete Semilattice vs.\ Complete Metric Space}
\label{sec:limit-infer-cons}

In order to compare the complete semilattice of partial terms with the
complete metric space of term, it is convenient to have an alternative
characterisation of the similarity $\similar{s}{t}$ of two terms
$s,t$, which in turn provides an alternative characterisation of the
metric $\dd$ on terms. To this end we use the \emph{truncation} of a
term at a certain depth. This notion was originally used by Arnold and
Nivat~\cite{arnold80fi} to show that the $\dd$ is a complete
ultrametric on terms:
\begin{defi}[truncation]
  \label{def:trunc}
  Let $d \in \nat \cup \set{\infty}$ and $t \in \ipterms$. The
  \emph{truncation} $\trunc{t}{d}$ of $t$ at depth $d$ is defined
  inductively on $d$ as follows
  \begin{align*}
    \trunc{t}{0} &= \bot \hspace{40pt} \trunc{t}{\infty} = t \\
    \trunc{t}{d + 1} &=
    \begin{cases}
      t &\text{ if } t \in \calV\\
      f(\trunc{t_1}{d},\dots,\trunc{t_k}{d}) &\text{ if } t = f(t_1,\dots, t_k)
    \end{cases}
  \end{align*}
\end{defi}
More concisely we can say that the truncation of a term $t$ at depth
$d$ replaces all subterms at depth $d$ with $\bot$. From this we can
easily establish the following two properties of the truncation:
\begin{prop}[truncation]
  \label{prop:trunc}
  For each two $s,t \in \ipterms$ we
  have
  \begin{enumerate}[label=(\roman*)]
  \item $\trunc{t}{d} \lebot t$ for all $d\in \nat\cup \set\infty$.
  \item $\trunc{s}{d} \lebot t$ implies $\trunc{s}{d} = \trunc{t}{d}$
    for all $d\in \nat\cup \set\infty$ given $s$ is total.
  \item $\trunc{s}{d} = \trunc{t}{d}$ for all $d\in \nat$ iff $s = t$.
  \end{enumerate}
\end{prop}
\begin{proof}
  Straightforward.
\end{proof}

Recall that the similarity of two terms is the minimal depth at which
they differ resp.\ $\infty$ if they are equal. However, saying that
two terms differ at a certain minimal depth $d$ is the same as saying
that the truncation of the two terms at that depth $d$ coincide. This
provides an alternative characterisation of similarity:
\begin{prop}[characterisation of similarity]
  \label{prop:simTrunc}
  For each pair $s,t \in \iterms$ we have
  \[
  \similar{s}{t} = \max \setcom{d\in \nat\cup\set\infty}{\trunc{s}{d}
    = \trunc{t}{d}}
  \]
\end{prop}
\begin{proof}
  Straightforward.
\end{proof}

We can use this characterisation to show the first part of the
compatibility of the metric and the partial order:
\begin{lem}[metric limit equals limit inferior]
  \label{lem:limLiminf}
  Let $(t_\iota)_{\iota < \alpha}$ be a convergent sequence in
  $(\iterms,\dd)$. Then $\lim_{\iota \limto \alpha} t_\iota =
  \liminf_{\iota \limto \alpha} t_\iota$.
\end{lem}
\begin{proof}
  If $\alpha$ is a successor ordinal, this is trivial. Let $\alpha$ be
  a limit ordinal, $\oh t = \lim_{\iota \limto \alpha} t_\iota$, and
  $\ol t = \liminf_{\iota \limto \alpha} t_\iota$. Then for each
  $\epsilon \in \realp$ there is a $\beta < \alpha$ such that $\dd(\oh
  t, t_\iota) < \epsilon$ for all $\beta \le \iota < \alpha$. Hence,
  for each $d \in \nat$ there is a $\beta < \alpha$ such that
  $\similar{\oh t}{t_\iota} > d$ for all $\beta \le \iota <
  \alpha$. According to Proposition~\ref{prop:simTrunc}, $\similar{\oh
    t}{t_\iota} > d$ implies $\trunc{\oh t}{d} = \trunc{t_\iota}{d}$,
  which, according to Proposition~\ref{prop:trunc}, implies
  $\trunc{\oh t}{d} \lebot t_\iota$. Therefore, $\trunc{\oh t}{d}$ is
  a lower bound of $T_\beta = \setcom{t_\iota}{\beta \le \iota <
    \alpha}$, i.e.\ $\trunc{\oh t}{d} \lebot \Glbbot T_\beta$. Since
  $\ol t = \Lubbot_{\beta<\alpha} \Glbbot T_\beta$, we also have that
  $\Glbbot T_\beta \lebot \ol t$. By transitivity, we obtain
  $\trunc{\oh t}{d} \lebot \ol t$ for each $d \in \nat$. Since $\oh t$
  is total, we can thus conclude, according to
  Proposition~\ref{prop:trunc}, that $\oh t = \ol t$.
\end{proof}

Before we continue, we want introduce another characterisation of
similarity which bridges the gap to the partial order $\lebot$. In
order to follow this approach, we need the to define the
\emph{$\bot$-depth} of a term $t \in \ipterms$. It is the minimal
depth of an occurrence of the subterm $\bot$ in $t$:
\[
\sdepth{t}{\bot} = \min \setcom{\len\pi}{t(\pi) = \bot}\cup \set\infty
\]

Intuitively, the glb $s \glbbot t$ of two terms $s,t$ represents the
common structure that both terms share. The similarity
$\similar{s}{t}$ is a much more condensed measure. It only provides
the depth up two which the terms share a common structure. Using the
$\bot$-depth we can directly condense the glb $s \glbbot t$ to the
similarity $\similar{s}{t}$:
\begin{prop}[characterisation of similarity]
  \label{prop:simDepth}
  For each pair $s,t \in \iterms$ we have
  \[
  \similar{s}{t} = \sdepth{s \glbbot t}{\bot}
  \]
\end{prop}
\begin{proof}
  Follows from Lemma~\ref{lem:glbbot}.
\end{proof}

We can employ this alternative characterisation of similarity to show
the second part of the compatibility of the metric and the partial
order:
\begin{lem}[total limit inferior implies Cauchy]
  \label{lem:liminfCauchy}
  Let $(t_\iota)_{\iota<\alpha}$ be a sequence in $\iterms$ such that
  $\liminf_{\iota\limto\alpha} t_\iota$ is total. Then
  $(t_\iota)_{\iota<\alpha}$ is Cauchy.
\end{lem}
\begin{proof}
  For $\alpha$ a successor ordinal this is trivial. For the case that
  $\alpha$ is a limit ordinal, suppose that $(t_\iota)_{\iota<\alpha}$
  is not Cauchy. That is, there is an $\epsilon \in \realp$ such that
  for all $\beta < \alpha$ there is a pair $\beta < \iota,\iota' <
  \alpha$ with $\dd(t_\iota,t_{\iota'}) \ge \epsilon$. Hence, there is
  a $d \in \nat$ such that for all $\beta < \alpha$ there is a pair
  $\beta < \iota,\iota' < \alpha$ with $\similar{t_\iota}{t_{\iota'}}
  \le d$, which, according to Proposition~\ref{prop:simDepth}, is
  equivalent to $\sdepth{t_\iota \glbbot t_{\iota'}}{\bot} \le
  d$. That is,
  \begin{gather}
    \label{eq:liminfCauchy}
    \text{for each } \beta < \alpha \text{ there are } \beta < \iota ,
    \iota' < \alpha \text{ with } \sdepth{t_\iota \glbbot
      t_{\iota'}}{\bot} \le d
    \tag{1}
  \end{gather}

  Let $s_\beta = \Glbbot_{\beta \le \iota < \alpha} t_\iota$. Then
  $s_\beta \lebot t_\iota \glbbot t_{\iota'}$ for all $\beta \le
  \iota, \iota' < \alpha$, which implies $\sdepth{s_\beta}{\bot} \le
  \sdepth{t_\iota \glbbot t_{\iota'}}{\bot}$. By combining this with
  \eqref{eq:liminfCauchy}, we obtain $\sdepth{s_\beta}{\bot} \le
  d$. More precisely, we have that
  \begin{gather}
    \label{eq:liminfCauchyI}
    \text{for each } \beta < \alpha \text{ there is a } \pi \in
    \pos{s_\beta} \text{ with } \len\pi \le d \text{ and }
    s_\beta(\pi) = \bot.
    \tag{2}
  \end{gather}
  Let $\ol t = \liminf_{\iota \limto \alpha} t_\iota$. Note that $\ol t
  = \Lubbot_{\beta < \alpha} s_\beta$. Since, according to
  Lemma~\ref{lem:lubbot}, $\pos{\ol t} = \bigcup_{\beta <
    \alpha}\pos{s_\beta}$ we can reformulate \eqref{eq:liminfCauchyI}
  as follows:
  \begin{gather}
    \label{eq:liminfCauchyIp}
    \text{for each } \beta < \alpha \text{ there is a } \pi \in
    \pos{\ol t} \text{ with } \len\pi \le d \text{ and }
    s_\beta(\pi) = \bot.
    \tag{2'}
  \end{gather}
  Since there are only finitely many positions in $\ol t$ of length at
  most $d$, there is some $\pi^* \in \pos{\ol t}$ such that
  \begin{gather}
    \label{eq:liminfCauchyII}
    \text{for each } \beta < \alpha \text{ there is a } \beta \le
    \gamma < \alpha \text{ with } s_\gamma(\pi^*)
    = \bot.  \tag{3}
  \end{gather}
  Since $s_\beta \lebot s_\gamma$, whenever $\beta \le \gamma$, we can
  rewrite \eqref{eq:liminfCauchyII} as follows:
  \begin{gather}
    \label{eq:liminfCauchyIIp}
    s_\beta(\pi^*) = \bot \text{ for all } \beta < \alpha \text{ with } \pi^*
    \in \pos{s_\beta}.  \tag{3'}
  \end{gather}
  Since $\pi^* \in \pos{\ol t}$, we can employ Lemma~\ref{lem:lubbot}
  to obtain from \eqref{eq:liminfCauchyIIp} that $\ol t(\pi^*) =
  \bot$. This contradicts the assumption that $\ol t =
  \liminf_{\iota\limto\alpha} t_\iota$ is total.
\end{proof}

The following proposition combines Lemma~\ref{lem:limLiminf} and
Lemma~\ref{lem:liminfCauchy} in order to obtain the desired property
that the metric and the partial order are compatible:
\begin{prop}[partial order conservatively extends metric]
  \label{prop:poMetric}
  For every sequence $(t_\iota)_{\iota<\alpha}$ in $\iterms$ the
  following holds:
  \begin{gather*}
    \liminf_{\iota \limto \alpha} t_\iota = \lim_{\iota \limto \alpha}
    t_\iota \qquad \text{ whenever}\quad
    \begin{aligned}
      &\lim_{\iota \limto \alpha} t_\iota \text{ is defined, or}\\
      &\liminf_{\iota \limto \alpha} t_\iota \text{ is a total term.}
    \end{aligned}
  \end{gather*}
\end{prop}
\begin{proof}
  If $\lim_{\iota \limto \alpha}$ is defined, the equality follows
  from Lemma~\ref{lem:limLiminf}. If $\liminf_{\iota \limto \alpha}
  t_\iota$ is total, the sequence $(t_\iota)_{\iota<\alpha}$ is Cauchy
  by Lemma~\ref{lem:liminfCauchy}. Then, as the metric space
  $(\iterms,\dd)$ is complete, $(t_\iota)_{\iota<\alpha}$ converges
  and we can apply Lemma~\ref{lem:limLiminf} to conclude the equality.
\end{proof}

\subsection{\texorpdfstring{$\prs$}{p}-Convergence vs.\ 
  \texorpdfstring{$\mrs$}{m}-Convergence}
\label{sec:prs-conv-cons}

In the previous section we have established that the metric and the
partial order on (partial) terms are compatible in the sense that the
corresponding notions of limit and limit inferior coincide whenever
the limit is defined or the limit inferior is a total term. As weak
$\mrs$-convergence and weak $\prs$-convergence are solely based on the
limit in the metric space resp.\ the limit inferior in the partially
ordered set, we can directly apply this result to show that both
notions of convergence coincide on total reductions:
\begin{thm}[total weak $\prs$-convergence = weak
  $\mrs$-convergence]
  \label{thr:weakExt}
  For every reduction $S$ in a TRS the following equivalences hold:
%  \begin{center}
  \begin{enumerate}[label=(\roman*)]
%    \begin{inparaenum}[(i)]
    \item $S\fcolon s \pwacont$ is total iff $S\fcolon s \mwacont$,
      and \label{item:weakExtI}%
      \quad%
    \item $S\fcolon s \pwato t$ is total iff $S\fcolon s
      \mwato t$. \label{item:weakExtII}
  \end{enumerate}
%    \end{inparaenum}
%  \end{center}
\end{thm}
\begin{proof}
  Both equivalences follow directly from
  Proposition~\ref{prop:poMetric} and Fact~\ref{fact:step}, both of
  which are applicable as we presuppose that each term in the
  reduction is total.
\end{proof}

In order to replicate Theorem~\ref{thr:weakExt} for the strong notions
of convergence, we first need the following two lemmas that link the
property of increasing contraction depth to volatile positions and the
limit inferior, respectively:
\begin{lem}[strong $\mrs$-convergence]
  \label{lem:strongConvPos}
  Let $S = (t_\iota \to[\pi_\iota] t_{\iota+1})_{\iota < \lambda}$ be
  an open reduction. Then $(\len{\pi_\iota})_{\iota < \lambda}$ tends
  to infinity iff, for each position $\pi$, there is an ordinal $\alpha
  < \lambda$ such that $\pi_\iota \neq \pi$ for all $\alpha \le \iota <
  \lambda$.
\end{lem}
\begin{proof}
  The ``only if'' direction is trivial. For the converse direction,
  suppose that $\len{\pi_\iota}$ does not tend to infinity as $\iota$
  approaches $\lambda$. That is, there is some depth $d \in \nat$ such
  that there is no upper bound on the indices of reduction steps
  taking place at depth $d$. Let $d^*$ be the minimal such depth. That
  is, there is some $\alpha < \lambda$ such that all reduction steps in
  $\segm{S}{\alpha}{\lambda}$ are at depth at least $d^*$, i.e.\
  $\len{\pi_\iota} \ge d^*$ holds for all $\alpha \le \iota <
  \lambda$. Of course, also in $\segm{S}{\alpha}{\lambda}$ the indices of
  steps at depth $d^*$ are not bounded from above.  As all reduction
  steps in $\segm{S}{\alpha}{\lambda}$ take place at depth $d^*$ or
  below, $\trunc{t_\iota}{d^*} = \trunc{t_{\iota'}}{d^*}$ holds for
  all $\alpha \le \iota,\iota' < \lambda$. That is, all terms in
  $\segm{S}{\alpha}{\lambda}$ have the same set of positions of length
  $d^*$. Let $P^* = \setcom{\pi \in \pos{t_n}}{\len{\pi} = d^*}$ be
  this set. Since there is no upper bound on the indices of steps in
  $\segm{S}{\alpha}{\lambda}$ taking place at a position in $P^*$, yet,
  $P^*$ is finite, there has to be some position $\pi^*\in P^*$ for
  which there is also no such upper bound. This contradicts the
  assumption that there is always such an upper bound.
\end{proof}

\begin{lem}[limit inferior of truncations]
  \label{lem:limInfTrunc}
  % limit inferior of truncations %
  Let $(t_\iota)_{\iota<\lambda}$ be a sequence in $\ipterms$ and
  $(d_\iota)_{\iota <\lambda}$ a sequence in $\nat$ such that
  $\lambda$ is a limit ordinal and $(d_\iota)_{\iota<\lambda}$ tends
  to infinity. Then $\liminf_{\iota \limto \lambda} t_\iota =
  \liminf_{\iota \limto \lambda} \trunc{t_\iota}{d_\iota}$.
\end{lem}
\begin{proof}
  Let $\ol t = \liminf_{\iota \limto \lambda}
  \trunc{t_\iota}{d_\iota}$ and $\oh t = \liminf_{\iota \limto
    \lambda} t_\iota$. Since, according to
  Proposition~\ref{prop:trunc}, $\trunc{t_\iota}{d_\iota} \lebot
  t_\iota$ for each $\iota < \lambda$, we have that $\ol t \lebot \oh
  t$. Thus, it remains to be shown that also $\oh t \lebot \ol t$
  holds. That is, we have to show that $\oh t(\pi) = \ol t(\pi)$ holds
  for all $\pi \in \posNonBot{\oh t}$.

  Let $\pi \in \posNonBot{\oh t}$. That is, $\oh t(\pi) = f \neq
  \bot$. Hence, by Lemma~\ref{lem:lubbot}, there is some $\alpha <
  \lambda$ with $(\Glbbot_{\alpha\le\iota<\lambda} t_\iota)(\pi) =
  f$. Let $P = \setcom{\pi'}{\pi' \le \pi}$ be the set of all prefixes
  of $\pi$. Note that $\Glbbot_{\alpha\le\iota<\lambda} t_\iota \lebot
  t_\gamma$ for all $\alpha \le \gamma < \lambda$. Hence,
  $\Glbbot_{\alpha\le\iota<\lambda} t_\iota$ and $t_\gamma$ coincide
  in all occurrences in $P$ for all $\alpha \le \gamma <
  \lambda$. Because $(d_\iota)_{\iota < \lambda}$ tends to infinity,
  there is some $\alpha \le \beta < \lambda$ such that $d_\gamma >
  \len{\pi}$ for all $\beta \le \gamma < \lambda$. Consequently, since
  $\trunc{t_\gamma}{d_\gamma}$ and $t_\gamma$ coincide in all
  occurrences of length smaller than $d_\gamma$ for all $\gamma <
  \lambda$, we have that $\trunc{t_\gamma}{d_\gamma}$ and $t_\gamma$
  coincide in all occurrences in $P$ for all $\beta \le \gamma <
  \lambda$. Hence, $\trunc{t_\gamma}{d_\gamma}$ and
  $\Glbbot_{\alpha\le\iota<\lambda} t_\iota$ coincide in all
  occurrences in $P$ for all $\beta \le \gamma < \lambda$. Hence,
  according to Lemma~\ref{lem:glbbot},
  $\Glbbot_{\alpha\le\iota<\lambda} t_\iota$ and
  $\Glbbot_{\beta\le\iota<\lambda} \trunc{t_\iota}{d_\iota}$ coincide
  in all occurrences in $P$. Particularly, it holds that
  $(\Glbbot_{\beta\le\iota<\lambda} \trunc{t_\iota}{d_\iota})(\pi) =
  f$ which in turn implies by Lemma~\ref{lem:lubbot} that $\ol t(\pi)
  = f$.
\end{proof}

We now can prove the counterpart of Theorem~\ref{thr:weakExt} for
strong convergences:
\begin{thm}[total strong $\prs$-convergence = strong
  $\mrs$-convergence]
  \label{thr:strongExt}
  For every reduction $S$ in a TRS the following equivalences hold:
%  \begin{center}
%    \begin{inparaenum}[(i)]
  \begin{enumerate}[label=(\roman*)]
    \item $S\fcolon s \pacont$ is total iff $S\fcolon s \macont$,
      and \label{item:strongExtI}%
      \quad%
    \item $S\fcolon s \pato t$ is total iff $S\fcolon s \mato
      t$. \label{item:strongExtII}
  \end{enumerate}
%    \end{inparaenum}
%  \end{center}
\end{thm}
\proof
  It suffices to only prove (\ref{item:strongExtII}) since
  (\ref{item:strongExtI}) follows from (\ref{item:strongExtII})
  according to Remark~\ref{rem:pcont} resp.\ Remark~\ref{rem:mcont}.

  Let $S = (\phi_\iota\fcolon t_\iota \to[\pi_\iota,c_\iota] t_{\iota+1})_{\iota<\alpha}$ be a reduction in
  a TRS $\calR_\bot$. We continue the proof by induction on
  $\alpha$. The case $\alpha = 0$ is trivial. If $\alpha$ is a
  successor ordinal $\beta + 1$, we can reason as follows
  \begin{align*}
    S\fcolon t_0 \pato t_\alpha \text{ total } &\text{ iff }\;
    \prefix{S}{\beta}\fcolon t_0 \pato t_\beta \text{ and }
    t_\beta\to[\calR] t_\alpha \tag{Remark~\ref{rem:pcont},
      Fact~\ref{fact:step}}\\%
    &\text{ iff }\; \prefix{S}{\beta}\fcolon t_0 \mato t_\beta \text{
      and } t_\beta \to[\calR] t_\alpha \tag{ind.\ hyp.}\\%
    &\text{ iff }\; S\fcolon t_0 \mato t_\alpha
    \tag{Remark~\ref{rem:mcont}}
  \end{align*}

  Let $\alpha$ be a limit ordinal. At first consider the ``only if''
  direction. That is, we assume that $S\fcolon t_0 \pato t_\alpha$ is
  total. According to Remark~\ref{rem:pcont}, we have that
  $\prefix{S}{\beta}\fcolon t_0 \pato t_\beta$ for each $\beta <
  \alpha$. Applying the induction hypothesis yields
  $\prefix{S}{\beta}\fcolon t_0 \mato t_\beta$ for each $\beta <
  \alpha$. That is, following Remark~\ref{rem:mcont}, we have
  $S\fcolon t_0 \macont$. Since $c_\iota \lebot
  t_\iota$ for all $\iota < \alpha$, we have that $t_ \alpha =
  \liminf_{\iota \limto \alpha} c_\iota \lebot \liminf_{\iota \limto
    \alpha} t_\iota$. Because $t_\alpha$ is total and, therefore,
  maximal w.r.t.\ $\lebot$, we can conclude that $t_\alpha =
  \liminf_{\iota \limto \alpha} t_\iota$. According to
  Proposition~\ref{prop:poMetric}, this also means that $t_\alpha =
  \lim_{\iota \limto \alpha} t_\iota$. For strong $\mrs$-convergence
  it remains to be shown that $(\len{\pi_\iota})_{\iota<\alpha}$ tends
  to infinity. So let us assume that this is not the case. By
  Lemma~\ref{lem:strongConvPos}, this means that there is a position
  $\pi$ such that, for each $\beta < \alpha$, there is some $\beta \le
  \gamma < \alpha$ such that the step $\phi_\gamma$ takes place at
  position $\pi$. By Lemma~\ref{lem:botLimRed}, this contradicts the
  fact that $t_\alpha$ is a total term.

  Now consider the converse direction and assume that $S \fcolon t_0
  \mato t_\alpha$. Following Remark~\ref{rem:mcont} we obtain
  $\prefix{S}{\beta} \fcolon t_0 \mato t_\beta$ for all $\beta <
  \alpha$, to which we can apply the induction hypothesis in order to
  get $\prefix{S}{\beta} \fcolon t_0 \pato t_\beta$ for all $\beta <
  \alpha$ so that we have $S \fcolon t_0 \pacont$, according to
  Remark~\ref{rem:pcont}. It remains to be shown that $t_\alpha=
  \liminf_{\iota \limto \alpha} c_\iota$. Since $S$ strongly
  $\mrs$-converges to $t_\alpha$, we have that
  \begin{enumerate}[label=(\alph*)]
%  \begin{inparaenum}[(a)]
  \item $t_\alpha= \lim_{\iota
      \limto \alpha} t_\iota$, and that
    \label{item:strongExtA}
  \item the sequence of depths $(d_\iota = \len{\pi_\iota})_{\iota<\alpha}$ tends to
    infinity.
    \label{item:strongExtB}
  \end{enumerate}
%  \end{inparaenum}
  Using Proposition~\ref{prop:poMetric} we can deduce from
  (\ref{item:strongExtA}) that $t_\alpha= \liminf_{\iota \limto \alpha}
  t_\iota$. Due to (\ref{item:strongExtB}), we can apply
  Lemma~\ref{lem:limInfTrunc} to obtain
  \[
  \liminf_{\iota \limto \alpha} t_\iota = \liminf_{\iota \limto
    \alpha} \trunc{t_\iota}{d_\iota} \quad\text{ and }\quad \liminf_{\iota \limto
    \alpha} c_\iota = \liminf_{\iota \limto \alpha}
  \trunc{c_\iota}{d_\iota}.
  \]
  Since $\trunc{t_\iota}{d_\iota} = \trunc{c_\iota}{d_\iota}$ for all
  $\iota < \alpha$, we can conclude that
  \[
  t_\alpha = \liminf_{\iota \limto \alpha} t_\iota = \liminf_{\iota \limto
    \alpha} \trunc{t_\iota}{d_\iota} = \liminf_{\iota \limto \alpha}
  \trunc{c_\iota}{d_\iota} = \liminf_{\iota \limto \alpha} c_\iota.
  \eqno{\qEd}\]

The main result of this section is that we do not loose anything when
switching from the metric model to the partial order model of
infinitary term rewriting. Restricted to the domain of the metric
model, i.e.\ total terms, both models coincide in the strongest
possible sense as Theorem~\ref{thr:weakExt} and
Theorem~\ref{thr:strongExt} confirm.

At the same time, however, the partial order model provides more
structure. Whenever the metric model can only conclude divergence, the
partial order model can qualify the degree of divergence. If a
reduction $\prs$-converges to $\bot$, it can be considered completely
divergent. If it $\prs$-converges to a term that only contains $\bot$
as proper subterms, it can be recognised as being only partially
divergent with the diverging parts of the reduction indicated by
'$\bot$'s, whereas complete absence of '$\bot$'s then indicates
complete convergence.

In the rest of this paper we will put our focus on strong
convergence. Theorem~\ref{thr:strongExt} will be one of the central
tools in Section~\ref{sec:relation-bohm-trees} where we shall discover
that Böhm-reachability coincides with strong $\prs$-reachability in
orthogonal systems. The other crucial tool that we will leverage is
the existence and uniqueness of complete developments. This is the
subject of the subsequent section.

\section{Strongly \texorpdfstring{$\prs$}{p}-Converging Complete Developments}
\label{sec:compl-devel}

The purpose of this section is to establish a theory of residuals and
complete developments in the setting of strongly $\prs$-convergent
reductions. Intuitively speaking, the residuals of a set of redexes
are the remains of this set of redexes after a reduction, and a
complete development of a set of redexes is a reduction which only
contracts residuals of these redexes and ends in a term with no
residuals.

Complete developments are a well-known tool for proving (finitary)
confluence of orthogonal systems \cite{terese03book}. It has also been
lifted to the setting of strongly $\mrs$-convergent reductions in
order to establish (restricted forms of) infinitary confluence of
orthogonal systems \cite{kennaway95ic}. As we have seen in
Example~\ref{ex:mconfl}, $\mrs$-convergence in general does not have
this property.

After introducing residuals and complete developments in
Section~\ref{sec:residuals}, we will show in
Section~\ref{sec:complete-development} resp.\
Section~\ref{sec:uniqueness} that complete developments do always
exist and that their final terms are uniquely determined. We then use
this in Section~\ref{sec:results} to show the Infinitary Strip Lemma
for strongly $\prs$-converging reductions which is a crucial tool for
proving our main result in Section~\ref{sec:relation-bohm-trees}.

\subsection{Residuals}
\label{sec:residuals}

At first we need to formalise the notion of residuals. It is virtually
equivalent to the definition for strongly $\mrs$-convergent reduction
by Kennaway et al.\ \cite{kennaway95ic}:
\begin{defi}[descendants, residuals]
  % generalisation of [kennaway95ic, Def. 4.1]
  \label{def:desc}
  % descendants, residuals %
  Let $\calR$ be a TRS, $S\fcolon t_0 \pto{\alpha}[\calR] t_\alpha$,
  and $U \subseteq \posNonBot{t_0}$. The \emph{descendants} of $U$ by
  $S$, denoted $\dEsc{U}{S}$, is the set of positions in $t_\alpha$
  inductively defined as follows:
  \begin{enumerate}[label=(\alph*)]
  \item If $\alpha = 0$, then $\dEsc{U}{S} = U$.
    \label{item:descA}
  \item If $\alpha = 1$, i.e.\ $S\fcolon t_0 \to[\pi,\rho] t_1$ for
    some $\rho\fcolon l \to r$, take any $u\in U$ and define the set
    $R_u$ as follows: If $\pi \not\le u$, then $R_u = \set{u}$. If $u$
    is in the pattern of the $\rho$-redex, i.e.\ $u = \pi\concat\pi'$
    with $\pi' \in \posFun{l}$, then $R_u = \emptyset$. Otherwise,
    i.e.\ if $u = \pi \concat w \concat x$, with $\atPos{l}{w} \in
    \calV$, then $R_u = \setcom{\pi \concat w' \concat
      x}{\atPos{r}{w'} = \atPos{l}{w}}$. Define $\dEsc{U}{S} =
    \bigcup_{u \in U} R_u$.
    \label{item:descB}
  \item If $\alpha = \beta + 1$, then $\dEsc{U}{S} =
    \dEsc{(\dEsc{U}{\prefix{S}{\beta}})}{\segm{S}{\beta}{\alpha}}$
    \label{item:descC}
  \item If $\alpha$ is a limit ordinal, then \quad
    $\dEsc{U}{S} = \posNonBot{t_\alpha} \cap \liminf_{\iota \limto \alpha}
    \dEsc{U}{\prefix{S}{\iota}}$
    \\
    That is, $u \in \dEsc{U}{S} \quad \text{ iff } \quad u \in
    \posNonBot{t_\alpha} \text{ and } \exists \beta < \alpha \forall
    \beta \le \iota < \alpha\fcolon u \in
    \dEsc{U}{\prefix{S}{\iota}}$
    \label{item:descD}
  \end{enumerate}
  If, in particular, $U$ is a set of redex occurrences, then
  $\dEsc{U}{S}$ is also called the set of \emph{residuals} of $U$ by
  $S$. Moreover, by abuse of notation, we write $\dEsc{u}{S}$ instead
  of $\dEsc{\set{u}}{S}$.
\end{defi}
Clauses (\ref{item:descA}), (\ref{item:descB}) and (\ref{item:descC})
are as in the finitary setting. Clause (\ref{item:descD}) lifts the
definition to the infinitary setting. However, the only difference to
the definition of Kennaway et al.\ is, that we consider partial terms
here. Yet, for technical reasons, the notion of descendants has to be
restricted to non-$\bot$ occurrences. Since $\bot$ cannot be a redex,
this is not a restriction for residuals, though.

\begin{rem}
  \label{rem:desc}
  One can easily see that the descendants of a set of
  non-$\bot$-occurrences is again a set of non-$\bot$-occurrences. The
  restriction to non-$\bot$-occurrences has to be made explicit for
  the case of open reductions. In fact, without this explicit
  restriction the definition would yield descendants which might not
  even be occurrences in the final term $t_\alpha$ of the
  reduction. For example, consider the system with the single rule
  $f(x) \to x$ and the strongly $\prs$-convergent reduction
  \[
  S\fcolon f^\omega \to f^\omega \to\; \dots \;\bot
  \]
  in which each reduction step contracts the redex at the root of
  $f^\omega$. Consider the set $U =\set{\emptyseq,
    \seq{0},\seq{0,0},\seq{0,0,0},\dots}$ of all positions in
  $t^\omega$. Without the abovementioned restriction, the descendants
  of $U$ by $S$ would be $U$ itself as the descendants of $U$ by each
  proper prefix of $S$ is also $U$. However, none of the positions
  $\seq{0},\seq{0,0},\seq{0,0,0},\dots \in U$ is even a position in
  the final term $\bot$. The position $\emptyseq \in U$ occurs in
  $\bot$, but only as a $\bot$-occurrence. With the restriction to
  non-$\bot$-occurrences we indeed get the expected result
  $\dEsc{U}{S} = \emptyset$.
\end{rem}

The definition of descendants of open reductions is quite subtle which
makes it fairly cumbersome to use in proofs. The lemma below
establishes an alternative characterisation which will turn out to be
useful in later proofs:
\begin{lem}[descendants of open reductions]
  \label{lem:descLimRed}
  % descendants of open reductions %
  Let $\calR$ be a TRS, $S\fcolon s \pto{\lambda}[\calR] t$ and $U
  \subseteq \posNonBot{s}$, with $\lambda$ a limit ordinal and $S =
  ({t_\iota\to[\pi_\iota,c_\iota]t_{\iota + 1}})_{\iota <
    \lambda}$. Then it holds that for each position $\pi$
  \[
  \pi \in \dEsc{U}{S} \quad \text{iff} \quad \text{there is some } \beta <
  \lambda \text{ with } \pi \in \dEsc{U}{\prefix{S}{\beta}} \text{
    and }\forall \beta \le \iota < \lambda\;\; \pi_\iota \not\le \pi.
  \]
\end{lem}
\begin{proof}
  \def\claima{(\ref{eq:descLimRed1})}
  \def\claimb{(\ref{eq:descLimRed2})}
  We first prove the ``only if'' direction. To this end, assume that
  $\pi \in \dEsc{U}{S}$. Hence, it holds that
  \begin{gather*}
    \pi \in \posNonBot{t} \text{ and there is some } \gamma_1 < \lambda
    \text{ such that } \pi \in \dEsc{U}{\prefix{S}{\iota}}
    \text{ for all } \gamma_1 \le \iota < \lambda \tag{1}
    \label{eq:descLimRed1}
  \end{gather*}
    Particularly, we have that $t(\pi) \neq \bot$. Applying
  Lemma~\ref{lem:nonBotLimRed} then yields that
  \begin{gather*}
    \text{there is some } \gamma_2 < \lambda \text{ such that }
    \pi_\iota \not\le \pi \text{ for all }  \gamma_2 \le \iota < \lambda
    \tag{2}
    \label{eq:descLimRed2}
  \end{gather*}
  Now take $\beta = \max \set{\gamma_1,\gamma_2}$. Then it holds that
  $\pi \in \dEsc{U}{\prefix{S}{\beta}}$ and that $\pi_\iota
  \not\le \pi$ for all $\beta \le \iota < \lambda$ due to \claima{}
  and \claimb{}, respectively.

  Next, consider the converse direction of the statement: Let $\beta <
  \lambda$ be such that $\pi \in \dEsc{U}{\prefix{S}{\beta}}$ and
  $\pi_\iota \not\le \pi$ for all $\beta \le \iota < \lambda$. We will
  show that $\pi \in \dEsc{U}{S}$ by proving the stronger statement
  that $\pi \in \dEsc{U}{\prefix{S}{\gamma}}$ for all $\beta \le
  \gamma \le \lambda$. We do this by induction on $\gamma$.
  
  For $\gamma = \beta$, this is trivial. Let $\gamma = \gamma' + 1 >
  \beta$. Note that, by definition, $\dEsc{U}{\prefix{S}{\gamma}} =
  \dEsc{\left(\dEsc{U}{\prefix{S}{\gamma'}}\right)}
  {\segm{S}{\gamma'}{\gamma}}$. Hence, since for the $\gamma'$-th step
  we have, by assumption, $\pi_{\gamma'} \not\le \pi$ and for the
  preceding reduction we have, by induction hypothesis, that $\pi \in
  \dEsc{U}{\prefix{S}{\gamma'}}$, we can conclude that $\pi \in
  \dEsc{U}{\prefix{S}{\gamma}}$.

  Let $\gamma > \beta$ be a limit ordinal. By induction hypothesis, we
  have that $\pi \in \dEsc{U}{\prefix{S}{\iota}}$ for each $\beta \le
  \iota < \gamma$. Particularly, this implies that $\pi \in
  \posNonBot{t_\beta}$. Together with the assumption that $\pi_\iota
  \not\le \pi$ for all $\beta \le \iota < \gamma$, this yields that
  $\pi \in \posNonBot{t_\gamma}$ according to
  Lemma~\ref{lem:nonBotLimRed}. Hence, $\pi \in
  \dEsc{U}{\prefix{S}{\gamma}}$.
\end{proof}

The following lemma confirms the expected monotonicity of descendants:
\begin{lem}[monotonicity of descendants]
  \label{lem:descMon}
  % monotonicity of descendants %
  Let $\calR$ be a TRS, $S\fcolon s \pato[\calR] t$ and $U,V \subseteq
  \posNonBot{s}$. If $U\subseteq V$, then $\dEsc{U}{S} \subseteq
  \dEsc{V}{S}$.
\end{lem}
\begin{proof}
  Straightforward induction on the length of $S$.
\end{proof}

This lemma can be generalised such that we can see that descendants
are defined ``pointwise'':
\begin{prop}[pointwise definition of descendants]
  \label{prop:descPoint}
  % pointwise definition of descendants %
  Let $\calR$ be a TRS, $S\fcolon s \pato[\calR] t$ and $U \subseteq
  \posNonBot{s}$. Then it holds that $\dEsc{U}{S} = \bigcup_{u \in U}
  \dEsc{u}{S}$.
\end{prop}
\begin{proof}
  Let $S = (t_\iota \to[\pi_\iota,c_\iota] t_{\iota + 1})_{\iota <
    \alpha}$. For $\alpha = 0$ and $\alpha = 1$, the statement is
  trivially true. If $\alpha = \alpha' + 1 > 1$, then abbreviate
  $\prefix{S}{\alpha'}$ and $\segm{S}{\alpha'}{\alpha}$ by $S_1$ and
  $S_2$, respectively, and reason as follows:
  \begin{align*}
    \dEsc{U}{S} & =\dEsc{(\dEsc{U}{S_1})}{S_2}
    \stackrel{IH}{=} \dEsc{\underbrace{(\bigcup_{u \in
          U} \overbrace{\dEsc{u}{S_1}}^{V_u})}_V}{S_2}
    \stackrel{IH}= \bigcup_{u\in V} \dEsc{u}{S_2}
    \\
    &= \bigcup_{u \in U} \bigcup_{v \in V_u}
    \dEsc{v}{S_2}
    \stackrel{IH}= \bigcup_{u\in U} \dEsc{V_u}{S_2}
    = \bigcup_{u\in U} \dEsc{(\dEsc{u}{S_1})}{S_2}
    = \bigcup_{u \in U} \dEsc{u}{S}
  \end{align*}

  Let $\alpha$ be a limit ordinal. The ``$\supseteq$'' direction of
  the equation follows from Lemma~\ref{lem:descMon}. For the converse
  direction, assume that $\pi \in \dEsc{U}{S}$. By
  Lemma~\ref{lem:descLimRed}, there is some $\beta < \alpha$ such that
  $\pi_\iota \not\le \pi$ for all $\beta \le \iota < \alpha$ and $\pi
  \in \dEsc{U}{\prefix{S}{\beta}}$. Applying the induction hypothesis
  yields that $\pi \in \bigcup_{u \in U}
  \dEsc{u}{\prefix{S}{\beta}}$, i.e.\ there is some $u^* \in U$ such
  that $\pi \in \dEsc{u^*}{\prefix{S}{\beta}}$. By employing
  Lemma~\ref{lem:descLimRed} again, we can conclude that $\pi \in
  \dEsc{u^*}{S}$ and, therefore, that $\pi \in \bigcup_{u \in U}
  \dEsc{u}{S}$.
\end{proof}

Note that the above proposition fails if we would include
$\bot$-occurrences in our definition of descendants: Reconsider the
example in Remark~\ref{rem:desc} and assume we would drop the
restriction to non-$\bot$-occurrences. Then the residuals
$\dEsc{u}{S}$ of each occurrence $u\in U$ would be empty, whereas the
residuals $\dEsc{U}{S}$ of all occurrences would be the root
occurrence $\seq{}$.

\begin{prop}[uniqueness of ancestors]
  \label{prop:descUnique}
  % uniqueness of descendants %
  Let $\calR$ be TRS, $S\fcolon s \pato[\calR] t$ and $U,V \subseteq
  \posNonBot{s}$. If $U \cap V = \emptyset$, then
  $\dEsc{U}{S}\cap\dEsc{V}{S} = \emptyset$.
\end{prop}
\begin{proof}
  We will prove the contraposition of the statement. To this end,
  suppose that there is some occurrence $w \in \dEsc{U}{S} \cap
  \dEsc{V}{S}$. By Proposition~\ref{prop:descPoint}, there are
  occurrences $u \in U$ and $v \in V$ such that $w \in
  \dEsc{u}{S}\cap\dEsc{v}{S}$. We will show by induction on the length
  of $S$ that then $u = v$ and, therefore, $U \cap V
  \neq\emptyset$. If $S$ is empty, then this is trivial. If $S$ is of
  successor ordinal length or open, then $u=v$ follows from the
  induction hypothesis.
\end{proof}

\begin{rem}
  \label{rem:prsAncestor}
  The two propositions above imply that each descendant $u' \in
  \dEsc{U}{S}$ of a set $U$ of occurrences is the descendant of a
  uniquely determined occurrence $u \in U$, i.e.\ $u' \in \dEsc{u}{S}$
  for exactly one $u\in U$. This occurrence $u$ is also called the
  \emph{ancestor} of $u'$ by $S$.
\end{rem}

The following proposition confirms a property of descendants that one
expects intuitively: The descendants of descendants are again
descendants. That is, the concept of descendants is composable.
\begin{prop}[descendants of sequential reductions]
  \label{prop:descSeqRed}
  % descendants of sequential reductions %
  Let $\calR$ be a TRS, $S\fcolon t_0 \pato[\calR] t_1$,
  $T\fcolon t_1 \pato[\calR] t_2$, and $U \subseteq
  \posNonBot{t_0}$. Then $\dEsc{U}{S\concat T} =
  \dEsc{(\dEsc{U}{S})}{T}$.
\end{prop}
\begin{proof}
  Straightforward proof by induction on the length of $T$.
\end{proof}
% \begin{proof}
%   We conduct the proof by induction on $\beta$. If $\beta = 0$, then
%   the statement is trivially true. Suppose that $\beta$ is a successor
%   ordinal $\beta' + 1$. Let $T_1$ and $T_2$ denote
%   $\prefix{T}{\beta'}$ and $\segm{T}{\beta'}{\beta}$,
%   respectively. Then we have the following equations:
%   \begin{align*}
%     \dEsc{U}{S\concat T} &=
%     \dEsc{U}{S\concat T_1\concat T_2} =
%     \dEsc{(\dEsc{U}{S\concat T_1})}{T_2}
%     \stackrel{IH}= \dEsc{(\dEsc{(\dEsc{U}{S})}{T_1})}{T_2}
%     = \dEsc{(\dEsc{U}{S})}{T_1\concat T_2} =
%     \dEsc{(\dEsc{U}{S})}{T}
%   \end{align*}
  
%   Suppose $\beta$ is a limit ordinal. Hence, also $\alpha + \beta$,
%   the length of $S\concat T$, is a limit ordinal. Therefore, we can
%   reason as follows:
%   \begin{align*}
%     u \in \dEsc{U}{S\concat T} \quad &\text{iff} \quad u \in
%     \posNonBot{t_2} \text{ and } \exists \gamma < \alpha + \beta \forall
%     \gamma \le \iota < \alpha + \beta\;\; u \in
%     \dEsc{U}{\prefix{(S\concat T)}{\iota}}
%     \\
%     &\text{iff} \quad u \in \posNonBot{t_2} \text{ and } \exists \gamma <
%     \beta \forall \gamma \le \iota < \beta\;\; u \in
%     \dEsc{U}{S\concat(\prefix{T}{\iota})}
%     \\
%     &\text{iff} \quad u \in \posNonBot{t_2} \text{ and } \exists \gamma <
%     \beta \forall \gamma \le \iota < \beta\;\; u \in
%     \dEsc{(\dEsc{U}{S})}{\prefix{T}{\iota}}
%     \tag{ind.\ hyp.}\\
%     &\text{iff} \quad u \in \dEsc{(\dEsc{U}{S})}{T}
%   \end{align*}
% \end{proof}
The following proposition confirms that the disjointness of
occurrences is propagated through their descendants:
\begin{prop}[disjoint descendants]
  \label{prop:disjDesc}
  % disjoint descendants %
  The descendants of a set of pairwise disjoint occurrences are
  pairwise disjoint as well.
\end{prop}
\begin{proof}
  Let $S\fcolon s \pto{\alpha} t$ and let $U$ be a set of pairwise
  disjoint occurrences in $s$. We show that $\dEsc{U}{S}$ is also a
  set of pairwise disjoint occurrences by induction on $\alpha$.

  For $\alpha$ being $0$, the statement is trivial, and, for $\alpha$
  being a successor ordinal, the statement follows straightforwardly
  from the induction hypothesis. Let $\alpha$ be limit ordinal and
  suppose that there are two occurrences $u,v \in \dEsc{U}{S}$ which
  are not disjoint. By definition, there are ordinals $\beta_1,\beta_2
  < \alpha$ such that $u \in \dEsc{U}{\prefix{S}{\iota}}$ for all
  $\beta_1\le\iota<\alpha$, and $v \in \dEsc{U}{\prefix{S}{\iota}}$
  for all $\beta_2\le\iota<\alpha$. Let $\beta =
  \max\set{\beta_1,\beta_2}$. Then we have that $u,v \in
  \dEsc{U}{\prefix{S}{\beta}}$. This, however, contradicts the
  induction hypothesis which, in particular, states that
  $\dEsc{U}{\prefix{S}{\beta}}$ is a set of pairwise disjoint
  occurrences.
\end{proof}

For the definition of complete developments it is important that the
descendants of redex occurrences are again redex occurrences:
\begin{prop}[residuals]
  \label{prop:residual}
  % disjoint redexes %
  Let $\calR$ be an orthogonal TRS, $S\fcolon s \pato[\calR] t$ and
  $U$ a set of redex occurrences in $s$. Then $\dEsc{U}{S}$ is a set
  of redex occurrences in $t$.
\end{prop}
\begin{proof}
  \def\claima{(\ref{eq:prsDisjRedex1})}
  \def\claimb{(\ref{eq:prsDisjRedex2})}
  \def\claimc{(\ref{eq:prsDisjRedex3})}
  \def\claimbp{(\ref{eq:prsDisjRedex2p})}
  \def\claimcp{(\ref{eq:prsDisjRedex3p})}
  Let $S = (t_\iota \to[\pi_\iota,c_\iota] t_{\iota + 1})_{\iota <
    \alpha}$. We proceed by induction on $\alpha$. For $\alpha$ being
  $0$, the statement is trivial, and, for $\alpha$ a successor
  ordinal, the statement follows straightforwardly from the induction
  hypothesis.

  So assume that $\alpha$ is a limit ordinal and that $\pi \in
  \dEsc{U}{S}$. We will show that $\atPos{t}{\pi}$ is a redex. From
  Lemma~\ref{lem:descLimRed} we obtain that
  \begin{gather*}
    \text{there is some } \beta < \alpha \text{ with } \pi \in
    \dEsc{U}{\prefix{S}{\beta}} \text{ and } \pi_\iota
    \not\le \pi \text{ for all } \beta \le \iota < \alpha.  \tag{1}
    \label{eq:prsDisjRedex1}
  \end{gather*}
  By applying the induction hypothesis, we get that $\pi$ is a redex
  occurrence in $t_\beta$. Hence, there is some rule $l\to r \in R$
  such that $\atPos{t_\beta}{\pi}$ is an instance of $l$.

  We continue this proof by showing the following stronger claim:
  \begin{align*}
    \text{for all } \beta \le \gamma \le \alpha
    &&&\atPos{t_\gamma}{\pi} \text{ is an instance of } l, \text{ and}
    \tag{2}
    \label{eq:prsDisjRedex2}
    \\
    &&&\atPos{c_\iota}{\pi} \text{ is an instance of } l \text{ for
      all } \beta \le \iota < \gamma \tag{3}
    \label{eq:prsDisjRedex3}
  \end{align*}
  For the special case $\gamma = \alpha$ the above claim \claimb{}
  implies that $\atPos{t}{\pi}$ is a redex.

  We proceed by an induction on $\gamma$. For $\gamma = \beta$, part
  \claimb{} of the claim has already been shown and \claimc{} is
  vacuously true. Let $\gamma = \gamma' + 1 > \beta$. According to the
  induction hypothesis, \claimb{} and \claimc{} hold for
  $\gamma'$. Hence, it remains to be shown that both
  $\atPos{t_\gamma}{\pi}$ and $\atPos{c_{\gamma'}}{\pi}$ are instances
  of $l$. At first consider $\atPos{c_{\gamma'}}{\pi}$. Recall that
  $c_{\gamma'} = \substAtPos{t_{\gamma'}}{\pi_{\gamma'}}{\bot}$. At
  first consider the case where $\pi$ and $\pi_{\gamma'}$ are
  disjoint. Then $\atPos{c_{\gamma'}}{\pi} =
  \atPos{t_{\gamma'}}{\pi}$. Since, by induction hypothesis,
  $\atPos{t_{\gamma'}}{\pi}$ is an instance of $l$, so is
  $\atPos{c_{\gamma'}}{\pi}$. Next, consider the case where $\pi$ and
  $\pi_{\gamma'}$ are not disjoint. Because of \claima{}, we then have
  that $\pi < \pi_{\gamma'}$, i.e.\ there is some non-empty $\pi'$
  with $\pi_{\gamma'} = \pi \concat \pi'$. Since $\calR$ is
  non-overlapping, $\pi'$ cannot be a position in the pattern of the
  redex $\atPos{t_{\gamma'}}{\pi}$ w.r.t.\ $l$. Therefore, also
  $\atPos{c_{\gamma'}}{\pi}$ is an instance of $l$. So in either case
  $\atPos{c_{\gamma'}}{\pi}$ is an instance of $l$. Since $c_{\gamma'}
  \lebot t_\gamma$, also $\atPos{t_\gamma}{\pi}$ is an instance of
  $l$.

  Let $\gamma > \beta$ be a limit ordinal. Part \claimc{} of the claim
  follows immediately from the induction hypothesis. Hence,
  $\atPos{c_\iota}{\pi}$ is an instance of $l$ for all $\beta \le
  \iota < \gamma$. This and \claima{} implies that all terms in the
  set $T = \setcom{c_\iota}{\beta \le \iota < \gamma}$ coincide in all
  occurrences in the set
  \[
  P = \setcom{\pi'}{\pi'\le \pi} \cup \setcom{\pi\concat\pi'}{\pi' \in
    \posFun{l}}
  \]
  $P$ is obviously closed under prefixes. Therefore, we can apply
  Lemma~\ref{lem:glbbot} in order to obtain that $\Glbbot T$
  coincides with all terms in $T$ in all occurrences in $P$. Since
  $\Glbbot T \lebot t_\gamma$, this property carries over to
  $t_\gamma$. Consequently, also $\atPos{t_\gamma}{\pi}$ is an
  instance of $l$.
\end{proof}

Next we want to establish an alternative characterisation of
descendants based on labellings. This is a well-known technique
\cite{terese03book} that keeps track of descendants by labelling the
symbols at the relevant positions in the initial term. In order to
formalise this idea, we need to extend a given TRS such that it can
also deal with terms that contain labelled symbols:
\begin{defi}[labelled TRSs/terms]
  Let $\calR = (\Sigma,R)$ be a TRS.
  \begin{enumerate}[label=(\roman*)]
  \item The \emph{labelled signature} $\Sigma^\lab$ is defined as
    $\Sigma \cup \setcom{f^\lab}{f \in \Sigma}$. The arity of the
    function symbol $f^\lab$ is the same as that of $f$. The symbols
    $f^\lab$ are called \emph{labelled}; the symbols $f \in \Sigma$
    are called \emph{unlabelled}. Terms over $\Sigma^\lab$ are called
    \emph{labelled terms}. Note that the symbol $\bot \in \Sigma_\bot$
    has no corresponding labelled symbol $\bot^\lab$ in the labelled
    signature $\Sigma^\lab_\bot$. Likewise, there are no labelled
    variables.
  \item Labelled terms can be projected back to the original
    unlabelled ones by removing the labels via the projection function
    $\unlab{\cdot}\fcolon\iterms[\Sigma^\lab_\bot] \funto \ipterms$:
    \begin{align*}
      \unlab{\bot} &= \bot%
      \qquad \qquad \unlab{x} = x  &&\text{for all } x \in \calV, \text{ and} \\%
      \unlab{f^\lab(t_1,\dots,t_k)} &= \unlab{f(t_1,\dots,t_k)} =
      f(\unlab{t_1},\dots,\unlab{t_k}) &&\text{for all } f \in
      \Sigma^{(k)}
    \end{align*}
  \item The \emph{labelled TRS} $\calR^\lab$ is defined as
    $(\Sigma^\lab,R^\lab)$, where $R^\lab = \setcom{l \to r}{\unlab{l} \to r \in R}$.
  \item For each rule $l \to r \in R^\lab$, we define its unlabelled
    original $\unlab{l \to r} = \unlab{l} \to r$ in $R$.
  \item Let $t \in \ipterms$ and $U \subseteq \posFun{t}$. The term
    $t^{(U)} \in \iterms[\Sigma_\bot^\lab]$ is defined by
    \begin{gather*}
      t^{(U)}(\pi) = 
      \begin{cases}
        t(\pi) &\text{if } \pi \nin U\\
        t(\pi)^\lab &\text{if } \pi \in U
      \end{cases}
    \end{gather*}
    That is, $\unlab{t^{(U)}} = t$ and the labelled symbols in
    $t^{(U)}$ are exactly those at positions in $U$.
  \end{enumerate}
\end{defi}

\noindent The key property which is needed in order to make the labelling
approach work is that any reduction in a left-linear TRS that starts
in some term $t$ can be lifted for any labelling $t'$ of $t$ to a
unique equivalent reduction in the corresponding labelled TRS that
starts in $t'$:
\begin{prop}[lifting reductions to labelled TRSs]
  \label{prop:liftLabelled}
  % lifting of reductions to labelled TRS %
  Let $\calR = (\Sigma,R)$ be a left-linear TRS, $S = (s_\iota
  \to[\rho_\iota,\pi_\iota] s_{\iota + 1})_{\iota<\alpha}$ a reduction
  strongly $\prs$-converging to $s_\alpha$ in $\calR$ , and $t_0 \in
  \iterms[\Sigma_\bot^\lab]$ a labelled term with $\unlab{t_0} =
  s_0$. Then there is a unique reduction $T = (t_\iota
  \to[\rho'_\iota,\pi_\iota] t_{\iota + 1})_{\iota< \alpha }$ strongly
  $\prs$-converging to $t_\alpha$ in $\calR^\lab$ such that
  \begin{enumerate}[label=(\alph*)]
  \item $\unlab{t_\iota} = s_\iota$, $\unlab{\rho'_\iota} =
    \rho_\iota$, for all $\iota < \alpha$, and
    \label{item:liftLabelled1}
  \item $\unlab{t_\alpha} = s_\alpha$.
    \label{item:liftLabelled2}
  \end{enumerate}
\end{prop}
\begin{proof}
  \def\itema{(\ref{item:liftLabelled1})}
  \def\itemb{(\ref{item:liftLabelled2})}
  We prove this by an induction on $\alpha$. For the case of $\alpha$
  being zero, the statement is trivially true. For the case of
  $\alpha$ being a successor ordinal, the statement follows
  straightforwardly from the induction hypothesis (the argument is the
  same as for finite reductions; e.g.\ consult \cite{terese03book}).

  Let $\alpha$ be a limit ordinal. By induction hypothesis, for each
  proper prefix $\prefix{S}{\gamma}$ of $S$ there is a uniquely
  defined strongly $\prs$-convergent reduction $T_\gamma$ in $\calR^\lab$
  satisfying \itema{} and \itemb{}. Since the sequence
  $(\prefix{S}{\iota})_{\iota < \alpha}$ forms a chain w.r.t.\ the
  prefix order $\le$, so does the corresponding sequence
  $(T_\iota)_{\iota < \alpha}$. Hence the sequence $T = \Lub_{\iota <
    \alpha} T_\iota$ is well-defined. By construction, $T_\gamma \le
  T$ holds for each $\gamma < \alpha$, and we can use the induction
  hypothesis to obtain part \itema{} of the proposition. In order to
  show $s_\alpha = \unlab{t_\alpha}$, we prove the two inequalities
  $s_\alpha \lebot \unlab{t_\alpha}$ and $s_\alpha \gebot
  \unlab{t_\alpha}$:

  To show $\unlab{t_\alpha} \lebot s_\alpha$, we take some $\pi \in
  \posNonBot{\unlab{t_\alpha}}$ and show that $\unlab{t_\alpha}(\pi) =
  s_\alpha(\pi)$. Let $f = \unlab{t_\alpha}(\pi)$. That is, either
  $t_\alpha(\pi) = f$ or $t_\alpha(\pi) = f^\lab$. In either case, we can
  employ Lemma~\ref{lem:nonBotLimRed} to obtain some $\beta < \alpha$
  such that $t_\beta(\pi) = f$ resp.\ $t_\beta(\pi) = f^\lab$ and
  $\pi_\iota \not\le \pi$ for all $\beta \le \iota < \alpha$. Since,
  by \itema{}, $s_\beta = \unlab{t_\beta}$, we have in both cases that
  $s_\beta(\pi) = f$. By applying Lemma~\ref{lem:nonBotLimRed} again,
  we get that $s_\alpha(\pi) = f$, too.

  Lastly, we show the converse inequality $s_\alpha \lebot
  \unlab{t_\alpha}$. For this purpose, let $\pi \in
  \posNonBot{s_\alpha}$ and $f = s_\alpha(\pi)$. By
  Lemma~\ref{lem:nonBotLimRed}, there is some $\beta < \alpha$ such
  that $s_\beta(\pi) = f$ and $\pi_\iota \not\le \pi$ for all $\beta
  \le \iota < \alpha$. Since, by \itema{}, $s_\beta =
  \unlab{t_\beta}$, we have that $t_\beta(\pi) \in
  \set{f,f^\lab}$. Applying Lemma~\ref{lem:nonBotLimRed} again then
  yields that $t_\alpha(\pi) \in \set{f,f^\lab}$ and, therefore,
  $\unlab{t_\alpha}(\pi) = f$.
\end{proof}

Having this, we can establish an alternative characterisation of
descendants using labellings:
\begin{prop}[alternative characterisation of descendants]
  \label{prop:chaDesc}
  % alternative characterisation of descendants % 
  Let $\calR$ be a left-linear TRS, $S\fcolon s_0 \pato[\calR]
  s_\alpha$, and $U \subseteq \posNonBot{s_0}$. Following
  Proposition~\ref{prop:liftLabelled}, let $T\fcolon t_0 \pato[\calR]
  t_\alpha$ be the unique lifting of $S$ to $\calR^\lab$ starting with
  the term $t_0 = s_0^{(U)}$. Then it holds that $t_\alpha =
  s_\alpha^{(\dEsc{U}{S})}$. That is, for all $\pi \in
  \posNonBot{s_\alpha}$, it holds that $t_\alpha(\pi)$ is labelled iff
  $\pi \in \dEsc{U}{S}$.
\end{prop}
\begin{proof}
  Let $S = (s_\iota \to[\pi_\iota] s_{\iota + 1})_{\iota<\alpha}$ and
  $T = (t_\iota \to[\pi_\iota] t_{\iota + 1})_{\iota< \alpha }$.  We
  prove the statement by an induction on the length $\alpha$ of
  $S$. If $\alpha = 0$, then the statement is trivially true. If
  $\alpha$ is a successor ordinal, then a straightforward argument
  shows that the statement follows from the induction hypothesis. Here
  the restriction to left-linear systems is vital.

  Let $\alpha$ be a limit ordinal and let $\pi \in
  \posNonBot{s_\alpha}$. We can then reason as follows:
  \begin{align*}
    t_\alpha(\pi) \text{ is labelled} \quad &\text{iff} \quad
    \exists \beta < \alpha\fcolon\;
    t_\beta(\pi) \text{ is labelled and } \forall  \beta \le \iota <
    \alpha\colon \;\; \pi_\iota \not\le \pi
    \tag{Lem.~\ref{lem:nonBotLimRed}}
    \\
    &\text{iff} \quad 
    \pi \in \dEsc{U}{\prefix{S}{\beta}} \text{ and } \forall  \beta \le \iota <
    \alpha\colon \;\; \pi_\iota \not\le \pi
    \tag{ind.\ hyp.}
    \\
    &\text{iff} \quad
    \pi \in \dEsc{U}{S}
    \tag{Lem.~\ref{lem:descLimRed}}
  \end{align*}
\end{proof}

\subsection{Constructing Complete Developments}
\label{sec:complete-development}

Complete developments are usually defined for (almost) orthogonal
systems. This ensures that the residuals of redexes are again redexes.
Since we are going to use complete developments for potentially
overlapping systems as well, we need to make restrictions on the set
of redex occurrences instead:
\begin{defi}[conflicting redex occurrences]
  Two distinct redex occurrences $u,v$ in a term $t$ are called
  \emph{conflicting} if there is a position $\pi$ such that $v =
  u\concat \pi$ and $\pi$ is a pattern position of the redex at $u$,
  or, vice versa, $u = v\concat \pi$ and $\pi$ is a pattern position
  of the redex at $v$. If this is not the case, then $u$ and $v$ are
  called \emph{non-conflicting}.
\end{defi}
One can easily see that in an orthogonal TRS any pair of redex
occurrences is non-conflicting.

\begin{defi}[(complete) development]
  \label{def:devel}
  % (complete) development %
  Let $\calR$ be a left-linear TRS, $s$ a partial term in $\calR$, and
  $U$ a set of pairwise non-conflicting redex occurrences in $s$.
  \begin{enumerate}[label=(\roman*)]
  \item A \emph{development} of $U$ in $s$ is a strongly
    $\prs$-converging reduction $S\fcolon s \pto{\alpha} t$ in which
    each reduction step $\phi_\iota\fcolon t_\iota \to[\pi_\iota]
    t_{\iota + 1}$ contracts a redex at $\pi_\iota \in
    \dEsc{U}{\prefix{S}{\iota}}$.
  \item A development $S\fcolon s \pato t$ of $U$ in $s$ is called
    \emph{complete},
    denoted $S\fcolon s \pato[U] t$, if $\dEsc{U}{S} = \emptyset$.
  \end{enumerate}
\end{defi}
This is a straightforward generalisation of complete developments
known from the finitary setting and coincides with the corresponding
formalisation for metric infinitary rewriting \cite{kennaway95ic} if
restricted to total terms.

The restriction to non-conflicting redex occurrences is essential in
order guarantee that the redex occurrences are independent from each
other:
\begin{prop}[non-conflicting residuals]
  Let $\calR$ be a left-linear TRS, $s$ a partial term in $\calR$, $U$
  a set of pairwise non-conflicting redex occurrences in $s$, and
  $S\fcolon s \sato[U] t$ a development of $U$ in $s$. Then also
  $\dEsc{U}{S}$ is a set of pairwise non-conflicting redex
  occurrences.
\end{prop}
\begin{proof}
  This can be proved by induction on the length of $S$. The part
  showing that the descendants are again redex occurrences can be
  copied almost verbatim from Proposition~\ref{prop:residual}. Instead
  of referring to the non-overlappingness of the system one can refer
  to the non-conflictingness of the preceding residuals which can be
  assumed by the induction hypothesis. The part of the induction proof
  that shows non-conflictingness is analogous to
  Proposition~\ref{prop:disjDesc}.
\end{proof}

It is relatively easy to show that complete developments of sets of
non-conflicting redex occurrences do always exists in the partial
order setting. The reason for this is that strongly $\prs$-continuous
reductions do always strongly $\prs$-converge as well. This means that
as long as there are (residuals of) redex occurrences left after an
incomplete development, one can extend this development arbitrarily by
contracting some of the remaining redex occurrences. The only thing
that remains to be shown is that one can devise a reduction strategy
which eventually contracts (all residuals of) all redexes. The
proposition below shows that a parallel-outermost reduction strategy
will always yield a complete development in a left-linear system.
\begin{prop}[complete developments]
  \label{prop:exComplDev}
  % generalisation of [kennaway95ic, Prop. 4.6]
  Let $\calR$ be a left-linear TRS, $t$ a partial term in $\calR$, and
  $U$ a set of pairwise non-conflicting redex occurrences in $t$. Then
  $U$ has a complete development in $t$.
\end{prop}
\begin{proof}
  \def\claima{(\ref{eq:compDev1})}
  Let $t_0 = t$, $U_0 = U$ and $V_0$ the set of outermost occurrences
  in $U_0$. Furthermore, let $S_0\fcolon t_0 \pato[V_0] t_1$ be some
  complete development of $V_0$ in $t_0$. $S_0$ can be constructed by
  contracting the redex occurrences in $V_0$ in a left-to-right
  order. This step can be continued for each $i < \omega$ by taking
  $U_{i+1} = \dEsc{U_i}{S_i}$, where $S_{i}\fcolon t_{i} \pato[V_{i}]
  t_{i+1}$ is some complete development of $V_{i}$ in $t_{i}$ with
  $V_{i}$ the set of outermost redex occurrences in $U_{i}$.

  Note that then, by iterating Proposition~\ref{prop:descSeqRed}, it
  holds that
  \begin{gather}
    \dEsc{U}{S_0\concat \dots \concat S_{n-1}} = U_n \quad
    \text{ for all } n < \omega \tag{1}
    \label{eq:compDev1}
  \end{gather}
  If there is some $n < \omega$ for which $U_n = \emptyset$, then
  $S_0 \concat \dots \concat S_{n-1}$ is a complete
  development of $U$ according to \claima{}.

  If this is not the case, consider the reduction $S = \Concat_{i <
    \omega} S_i$, i.e.\ the concatenation of all '$S_i$'s. We claim
  that $S$ is a complete development of $U$. Suppose that this is not
  the case, i.e.\ $\dEsc{U}{S} \neq \emptyset$. Hence, there is some
  $u \in \dEsc{U}{S}$. Since all '$U_i$'s are non-empty, so are the
  '$V_i$'s. Consequently, all '$S_i$'s are non-empty reductions which
  implies that $S$ is a reduction of limit ordinal length, say
  $\lambda$. Therefore, we can apply Lemma~\ref{lem:descLimRed} to
  infer from $u \in \dEsc{U}{S}$ that there is some $\alpha < \lambda$
  such that $u \in \dEsc{U}{\prefix{S}{\alpha}}$ and all reduction
  steps beyond $\alpha$ do not take place at $u$ or above. This is not
  possible due to the parallel-outermost reduction strategy that $S$
  adheres.
\end{proof}

This shows that complete developments of any set of redex occurrences
do always exist in any (almost) orthogonal system. This is already an
improvement over strongly $\mrs$-converging reductions, which only
allow this if no collapsing rules are present or the considered set of
redex occurrences does not contain an infinite set of nested
collapsing redexes -- also known as an \emph{infinite collapsing
  tower}.

We shall discuss the issue of collapsing rules as well as infinite
collapsing towers in more detail in the subsequent section, where we
will show that complete developments are also unique in the sense that
the final outcome is uniquely determined by the initial set of redexes
occurrences.

\subsection{Uniqueness of Complete Developments}
\label{sec:uniqueness}

The goal of this section is to show that the final term of a complete
development is uniquely determined by the initial set of redex
occurrences $U$. There are several techniques to show that in the
metric model. One of these approaches, introduced by Kennaway and de
Vries~\cite{kennaway03book} and detailed by Ketema and
Simonsen~\cite{ketema10lmcs,ketema05lpar} for infinitary combinatory
reduction systems, uses so-called \emph{paths}. Paths are constructed
such that they, starting from the root, run through the initial term
$t$ of the complete development, and whenever a redex occurrence of
the development is encountered, the path jumps to the root of the
right-hand side of the corresponding rule and jumps back to the term
$t$ when it reaches a variable in the right-hand side.
\def\thepath{%
  \draw[]%
  (ta) .. controls ($(ta)+(0,-.7)$) and ($(s1a)+(0,.7)$) .. (s1a);%
  \draw[]%
  (s1o) .. controls ($(s1o)+(0,-.3)$) and ($(s2a)+(0,.3)$) .. (s2a);%
  \draw[-|]%
  (s2o) .. controls ($(s2o)+(0,-.3)$) and ($(to)+(0,.3)$) .. (to);%
  \draw[]%
  (r1a) .. controls ($(r1a)+(.1,-.5)$) .. ($(r1o)+(-.1,.7)$)
  .. controls ($(r1o)+(-.2,.5)$) and ($(r1o)+(0,.3)$) .. (r1o);%
  \draw[]%
  (r2a) .. controls ($(r2a)+(.05,-.3)$) .. ($(r2o)+(.2,.7)$)
  .. controls ($(r2o)+(.1,.5)$) and ($(r2o)+(0,.3)$) .. (r2o);%
}

\begin{figure}
  \centering
  \subfloat[Constructing a path in a term.]{
    \begin{tikzpicture}[%
      remember picture,%
      jump/.style={%
        dashed,%
        ->,%
        draw%
      },%
      node distance=1.7cm,%
      tree triangle/.style={%
        isosceles triangle,%
        inner sep=1pt,%
        shape border rotate=90%
      },%
      big term/.style={%
        minimum width=5cm,%
        isosceles triangle apex angle=45,%
        tree triangle,%
        term back%
      },%
      small term 1/.style={%
        minimum width=1cm,%
        isosceles triangle apex angle=30,%
        tree triangle,%
        term const%
      },%
      small term 2/.style={%
        minimum width=1cm,%
        isosceles triangle apex angle=45,%
        tree triangle,%
        term const%
      },%
      small term 3/.style={%
        minimum width=1.3cm,%
        isosceles triangle apex angle=45,%
        tree triangle,%
        term stable%
      },%
      small term 4/.style={%
        minimum width=1.3cm,%
        isosceles triangle apex angle=55,%
        tree triangle,%
        term stable%
      }]%
      \node[big term, anchor=apex] (t) at (0,0) {};%
      \node at ($(t)+(1.5,-1)$) {$t$};%
      \node[small term 1, anchor=apex] (s1) at (-.3,-1.5) {$l_1$};%
      \node[small term 2, anchor=apex] (s2) at (.1,-4) {$l_2$};%

      \node[small term 1, anchor=center] (l1) at (-7,-1.5) {$l_1$};%
      \node[small term 2, anchor=center] (l2) at (-7,-4) {$l_2$};%

      \node[small term 3, anchor=center,right=of l1] (r1) {};%
      \node[small term 4, anchor=center,right=of l2] (r2) {};%
      \node at ($(r1)+(.3,-.25)$) {$r_1$};%
      \node at ($(r2)+(.3,-.2)$) {$r_2$};%

      \draw[shorten <=.9cm, shorten >=.9cm]%
      ($(l1)+(0,.4)$) edge[single step] ($(r1)+(0,.4)$)%
      ($(l2)+(0,.3)$) edge[single step] ($(r2)+(0,.3)$);%

      \coordinate (b1) at ($(l2.left corner)+(0,-.7)$);
      \coordinate (b2) at ($(b1)+(1,0)$);

      \draw[decorate,decoration=brace]%
      ($(b1)!(r1.right corner)!(b2)$) -- ($(b1)!(l1.left corner)!(b2)$)
      node[midway,below=.3cm] {$\calR$};

      \coordinate (ta) at (t.apex);%
      \coordinate (s1a) at (s1.apex);%
      \coordinate (s2a) at (s2.apex);%
      \coordinate (r1a) at (r1.apex);%
      \coordinate (r2a) at (r2.apex);%
      \coordinate (s1o) at (s1.250);%
      \coordinate (s2o) at (s2.290);%
      \coordinate (l1o) at (l1.250);%
      \coordinate (l2o) at (l2.290);%
      \coordinate (to)  at ($(t.290)+(0,.3)$);%
      \coordinate (r1o) at (r1.270);%
      \coordinate (r2o) at (r2.230);%

      \begin{scope}
        \foreach \n/\v in {1/$x$,2/$y$} {%
          \draw[black!30] (l\n o) -- ($(l\n o)+(0,-.1)$);%
          \node at ($(l\n o)+(0,-.3)$) {\v};%

          \draw[black!30] (r\n o) -- ($(r\n o)+(-.1,-.1)$);%
          \node at ($(r\n o)+(-.2,-.25)$) {\v};%
        }%

      \end{scope}
      \thepath

      \draw [every edge/.style={jump,bend right=15}]%
      (s1a) edge (r1a)%
      (r1o) edge (s1o)%
      (s2a) edge (r2a)%
      (r2o) edge (s2o);%

    \end{tikzpicture}%
  \label{fig:path1}
  }%
\\[1cm]
\subfloat[The constructed path.]{
  \begin{tikzpicture}[%
    remember picture,%
    jump/.style={%
      dashed,%
      ->,%
      draw%
    } ]
    \path[use as bounding box] (0,.6) rectangle (9.5,-1.6);
    \begin{scope}[transform canvas={rotate=90}]

      \coordinate (tmp) at (0,0);%
      \coordinate (s1a) at ($(tmp)-(ta)+(s1a)$);%
      \coordinate (ta) at (tmp);%
      
      \coordinate (tmp) at ($(s1a)-(0,1)$);%
      \coordinate (r1o) at ($(tmp)-(r1a)+(r1o)$);%
      \coordinate (r1a) at (tmp);%
      
      \coordinate (tmp) at ($(r1o)-(0,1)$);%
      \coordinate (s2a) at ($(tmp)-(s1o)+(s2a)$);%
      \coordinate (s1o) at (tmp);%
      
      \coordinate (tmp) at ($(s2a)-(0,1)$);%
      \coordinate (r2o) at ($(tmp)-(r2a)+(r2o)$);%
      \coordinate (r2a) at (tmp);%

      \coordinate (tmp) at ($(r2o)-(0,1)$);%
      \coordinate (to) at ($(tmp)-(s2o)+(to)$);%
      \coordinate (s2o) at (tmp);%

      \coordinate (top1) at (.5,0);%
      \coordinate (top2) at (.5,1);%

      \coordinate (bot1) at (-1.5,0);%
      \coordinate (bot2) at (-1.5,1);%

      \coordinate (to) at ($(to)-(0,.3)$);%

      \foreach \d in {top, bot} {%
        \foreach \n in {ta,s1a,r1a,r1o,s1o,s2a,r2a,r2o,s2o,to} {%
          \coordinate (\d _\n) at ($(\d 1)!(\n)!(\d 2)$);%
        }%
      }%
      \foreach \na / \nb / \ci / \co / \n in {%
        ta/s1a/termback/termfringe/$t$,%
        r1a/r1o/termstable/termstablefringe/$r_1$,%
        s1o/s2a/termback/termfringe/$t$,%
        r2a/r2o/termstable/termstablefringe/$r_2$,%
        s2o/to/termback/termfringe/$t$%
      }%
      {%
        \path[fill=\ci,decoration={random steps,segment length=2mm}]%
        decorate{(top_\na) -- (top_\nb)} -- (bot_\nb) decorate{ --
          (bot_\na)} -- cycle; %
        \path[draw=\co]%
        (top_\na) -- (bot_\na)%
        (top_\nb) -- (bot_\nb);%
        \node[rotate=-90] at ($(bot_\nb)!.5!(bot_\na) +(.5,0) $) {\n};%
      };%
      
      \coordinate (to) at ($(to)+(0,.3)$);%
      \thepath

      \draw [every edge/.style={jump,bend left=15}]%
      (s1a) edge (r1a)%
      (r1o) edge (s1o)%
      (s2a) edge (r2a)%
      (r2o) edge (s2o);%
      
    \end{scope}
  \end{tikzpicture}
  \label{fig:path2}
}
  
  \caption{A path.}
  \label{fig:path}
\end{figure}

%%% Local Variables: 
%%% mode: latex
%%% TeX-master: "thesis"
%%% End: 

%
Figure~\ref{fig:path1} illustrates this idea. It shows a path in a
term $t$ that encounters two redex occurrences of the complete
development. As soon as such a redex occurrence is encountered, the
path jumps to the right-hand side of the corresponding rule as
indicated by the dashed arrows. Then the path runs through the
right-hand side. When a variable is encountered, the path jumps back
to the position of the term $t$ that matches the variable. This jump
is again indicated by a dashed arrow. The path that is obtained by
this construction is shown in Figure~\ref{fig:path2}. With the
collection of the thus obtained paths one can then construct the final
term of the complete development. This technique -- slightly modified
-- can also be applied in the present setting.

A path consists of nodes, which are connected by edges. We have two
kinds of nodes: a node $(\top,\pi)$ represents a location in the term
$t$ and a node $(r,\pi,u)$ represents a location in the right-hand
side $r$ of a rule. These nodes of the form $(\top,\pi)$ and
$(r,\pi,u)$ encode that the path is currently at position $\pi$ in the
term $t$ resp.\ $r$. The additional component $u$ provides the
information that the path jumped to the right-hand side $r$ from the
redex $\atPos{t}{u}$. Both nodes and the edges between them are
labelled. Each node is labelled with the symbol at the current
location of the path, unless it is a redex occurrence in $t$ or a
variable occurrence in a right-hand side. The labellings of the edges
provide information on how the path moves through the terms: a
labelling $i$ represents a move along the $i$-th edge in the term tree
from the current location whereas an empty labelling indicates a jump
from or to a right-hand side of a rule.
\begin{defi}[path]
  \label{def:redexPath}
  % path %
  Let $\calR$ be a left-linear TRS, $t$ a partial term in $\calR$, and
  $U$ a set of pairwise non-conflicting redex occurrence in $t$. A
  $U,\calR$-\emph{path} (or simply \emph{path}) in $t$ is a sequence
  of length at most $\omega$ containing so-called \emph{nodes} and
  \emph{edges} in an alternating manner like this:
  \[
  \seq{n_0, e_0, n_1, e_1, n_2, e_2, \dots}
  \]
  where the '$n_i$'s are nodes and the '$e_i$'s are edges. A node is
  either a pair of the form $(\top,\pi)$ with $\pi \in \pos{t}$ or a
  triple of the form $(r,\pi,u)$ with $r$ the right-hand side of a
  rule in $\calR$, $\pi \in \pos{r}$, and $u \in U$. Edges are denoted
  by arrows $\edge$. Both edges and nodes might be labelled by
  elements in $\Sigma_\bot \cup \calV$ and $\nat$, respectively. We
  write paths as the one sketched above as
  \[
  \node{n_0} \edge \node{n_1} \edge \node{n_2} \edge \cdots
  \]
  or, when explicitly indicating labels, as
  \[
  \node{n_0}[l_0] \edge[l_1] \node{n_1}[l_2] \edge[l_3] \node{n_2}[l_4] \edge[l_5] \cdots
  \]
  where empty labels are explicitly given by the symbol $\emptylab$.
  If a path has a segment of the form $n \edge n'$, then we say there
  is an edge from $n$ to $n'$ or that $n$ has an outgoing edge to
  $n'$.

  Every path starts with the node $(\top,\emptyseq)$ and is either
  infinitely long or ends with a node. For each node $n$ having an
  outgoing edge to a node $n'$, the following must hold:
  \begin{enumerate}%[(1)]
  \item If $n$ is of the form $(\top,\pi)$, then
    \label{item:redexPath1}
    \begin{enumerate}[label=(\alph*)]
    \item $n' =(\top,\pi \concat i)$ and the edge is labelled by $i$,
      with $\pi\concat i \in \pos{t}$ and $\pi \nin U$, or
      \label{item:redexPath1a}
    \item $n' = (r,\emptyseq,u)$ and the edge is unlabelled, with
      $\atPos{t}{u}$ a $\rho$-redex for $\rho\fcolon l \to r \in R$
      and $u \in U$.
      \label{item:redexPath1b}
    \end{enumerate}
  \item If $n$ is of the form $(r,\pi,u)$, then
    \label{item:redexPath2}
    \begin{enumerate}[label=(\alph*)]
    \item $n' = (r,\pi \concat i,u)$ and the edge is labelled by $i$,
      with $\pi \concat i \in \pos{r}$, or
      \label{item:redexPath2a}
    \item $n' = (\top,u \concat \pi')$ and the edge is unlabelled, with
      $\atPos{t}{u}$ a $\rho$-redex for $\rho\fcolon l \to r \in R$,
      $\atPos{r}{\pi}$ a variable, and $\pi'$ the unique occurrence of
      $\atPos{r}{\pi}$ in $l$.  .
      \label{item:redexPath2b}
    \end{enumerate}
  \end{enumerate}

  Additionally, the nodes of a path are supposed to be labelled in the
  following way:
  \begin{enumerate}%[(1)]
    \setcounter{enumi}{2}
  \item A node of the form $(\top,\pi)$ is unlabelled if $\pi \in U$
    and is labelled by $t(\pi)$ otherwise.
    \label{item:redexPath3}
  \item A node of the form $(r,\pi, u)$ is unlabelled if
    $\atPos{r}{\pi}$ is a variable and labelled by $r(\pi)$ otherwise.
    \label{item:redexPath4}
  \end{enumerate}
\end{defi}
\def\itemRedexPathI{(\ref{item:redexPath1})}
\def\itemRedexPathIa{(\ref{item:redexPath1a})}
\def\itemRedexPathIb{(\ref{item:redexPath1b})}
\def\itemRedexPathII{(\ref{item:redexPath2})}
\def\itemRedexPathIIa{(\ref{item:redexPath2a})}
\def\itemRedexPathIIb{(\ref{item:redexPath2b})}
\def\itemRedexPathIII{(\ref{item:redexPath3})}
\def\itemRedexPathIV{(\ref{item:redexPath4})}
\begin{rem}
  The above definition is actually a coinductive one. This is
  necessary to also define paths of infinite length. Also in
  \cite{kennaway03book} paths are considered to be possibly infinite,
  although they are defined inductively and are, therefore, finite.
\end{rem}

\begin{rem}
  \label{rem:paths}
  Our definition of paths deviates slightly from the usual definition
  found in the literature \cite{kennaway95ic,ketema10lmcs,ketema11ic}:
  In our setting, term nodes are of the form $(\top,\pi)$. The symbol
  $\top$ is used to indicate that we are in the host term $t$. In the
  definitions found in the literature, the term $t$ itself is used for
  that, i.e.\ term nodes are of the form $(t,\pi)$. Our definition of
  paths makes them less dependant on the term $t$ they are constructed
  in. This makes it easier to construct a path in a host term from
  other paths in different host terms. This will become necessary in
  the proof of Lemma~\ref{lem:presPath}.  However, we have to keep in
  mind that the node labels in a path are dependent on the host term
  under consideration. Thus, the labelling of a path might be
  different depending on which host term it is considered to be in.
\end{rem}

Returning to the schematic example illustrated in
Figure~\ref{fig:path}, we can observe how the construction of a path
is carried out: The path starts with a segment in the term $t$. This
segment is entirely regulated by the rule \itemRedexPathIa{}; all its
edges and nodes are labelled according to \itemRedexPathIa{} and
\itemRedexPathIII{}. The jump to the right-hand side $r_1$ following
that initial segment is justified by rule \itemRedexPathIb{}. This
jump consists of a node $(\top,u_1)$, unlabelled according to
\itemRedexPathIII{}, corresponding to the redex occurrence $u_1$, and
an unlabelled edge to the node $(r_1,\emptyseq,u_1)$, corresponding to
the root of the right-hand side $r_1$. The segment of the path that
runs through the right-hand side $r_1$ is subject to rule
\itemRedexPathIIa{}; again all its nodes and edges are labelled, now
according to \itemRedexPathIIa{} and \itemRedexPathIV{}. As soon as a
variable is reached in the right-hand side term (in the schematic
example it is the variable $x$) a jump to the main term $t$ is
performed as required by rule \itemRedexPathIIb{}. This jump consists
of a node $(r_1,\pi,u_1)$, unlabelled according to \itemRedexPathIV{},
where $\pi$ is the current position in $r_1$, i.e.\ the variable
occurrence, and an unlabelled edge to the node $(\top,u_1\concat
\pi')$. The position $\pi'$ is the occurrence of the variable $x$ in
the left-hand side. As we only consider left-linear systems, this
occurrence is unique. Afterwards, the same behaviour is repeated: A
segment in $t$ is followed by a jump to a segment in the right-hand
side $r_2$ which is in turn followed by a jump back to a final segment
in $t$.

Note that paths do not need to be maximal. As indicated in the
schematic example, the path ends somewhere within the main term, i.e.\
not necessarily at a constant symbol or a variable. What the example
does not show, but which is obvious from the definition, is that paths
can also terminate within a right-hand side. A jump back to the main
term is only required if a variable is encountered.

The purpose of the concept of paths is to simulate the contraction of
all redexes of the complete development in a locally restricted
manner, i.e.\ only along some branch of the term tree. This locality
will keep the proofs more concise and makes them easier to understand
once we have grasped the idea behind paths. The strategy to prove our
conjecture of uniquely determined final terms is to show that paths
can be used to define a term and that a contraction of a redex of the
complete development preserves a property of the collection of all
paths which ensures that the induced term remains invariant. Then we
only have to observe that the induced term of paths in a term with no
redexes (in $U$) is the term itself.

The following fact is obvious from the definition of a path.
\begin{fact}
  \label{fact:emptyEdge}
  Let $\calR$ be a left-linear TRS, $t$ a partial term in $\calR$, and
  $U$ a set of redex occurrences in $t$.
  \begin{enumerate}[label=(\roman*)]
  \item An edge in a $U,\calR$-path in $t$ is unlabelled iff the
    preceding node is unlabelled.
  \item Any prefix of a $U,\calR$-path in $t$ that ends in a node is
    also a $U,\calR$-path in $t$.
  \end{enumerate}
\end{fact}

\noindent As we have already mentioned, collapsing rules and in particular
so-called infinite collapsing towers play a significant role in
$\mrs$-convergent reductions as they obstruct complete
developments. Also in our setting of $\prs$-convergent reductions they
are important as they are responsible for volatile positions:
\begin{defi}[collapsing rules]
  Let $\calR$ be a TRS.
  \begin{enumerate}[label=(\roman*)]
  \item A rule $l \to r$ in $\calR$ is called \emph{collapsing} if $r$
    is a variable. The unique position of the variable $r$ in $l$ is
    called the \emph{collapsing position} of the rule.
  \item A $\rho$-redex is called \emph{collapsing} if $\rho$ is a
    collapsing rule.
  \item A \emph{collapsing tower} is a non-empty sequence $(u_i)_{i <
      \alpha}$ of collapsing redex occurrences in a term $t$ such that
    $u_{i+1} = u_i \concat \pi_i$ for each $i<\alpha$, where $\pi_i$
    is a collapsing position of the redex at $u_i$. It is called
    \emph{maximal} if it is not a proper prefix of another collapsing
    tower.
  \end{enumerate}
\end{defi}

\noindent One can easily see that, in orthogonal TRSs, maximal collapsing towers
in the same term are uniquely determined by their topmost redex
occurrence. That is, two maximal collapsing towers $(u_i)_{i<\alpha},
(v_i)_{i<\alpha}$ in the same term are equal iff $u_0 = v_0$.

As mentioned, we shall use the $U,\calR$-paths in a term $t$ in order
to define the final term of a complete development of $U$ in
$t$. However, in order to do that, we only need the information that
is available from the labellings. The inner structure of nodes is only
used for the bookkeeping that is necessary for defining paths. The
following notion of traces defines projections to the labels of paths:
\begin{defi}[trace]
  Let $\calR$ be a left-linear TRS, $t$ a partial term in $\calR$, and
  $U$ a set of pairwise non-conflicting redex occurrences in $t$.
  \begin{enumerate}[label=(\roman*)]
  \item Let $\Pi$ be a $U,\calR$-path in $t$. The \emph{trace} of
    $\Pi$, denoted $\trace{t}{\Pi}$, is the projection of $\Pi$ to the
    labelling of its nodes and edges ignoring empty labels and the
    node label $\bot$.
  \item $\paths{t}{U}{\calR}$ is used to denote the set of all
    $U,\calR$-paths in $t$ that end in a labelled node, or are
    infinite but have a finite trace. The set of traces of paths in
    $\paths{t}{U}{\calR}$ is denoted by $\traces{t}{U}{\calR}$.
  \end{enumerate}
\end{defi}

\noindent By Fact~\ref{fact:emptyEdge}, the trace of a path is a sequence
alternating between elements in $\Sigma \cup \calV$ and $\nat$, which,
if non-empty, starts with an element in $\Sigma \cup \calV$. Moreover,
by definition, $\traces{t}{U}{\calR}$ is a set of finite traces of
$U,\calR$-paths in $t$.

As we have mentioned in Remark~\ref{rem:paths}, the labelling of a
path depends on the host term under consideration. Hence, also the
trace of a path is depended on the host term. That is why we need to
index the trace mapping $\trace{t}{\cdot}$ with the corresponding host
term $t$.

\begin{exa}
  Consider the term $t = g(f(g(h(\bot))))$ and the TRS $\calR$
  consisting of the two rules 
  \[
  f(x) \to h(x), \qquad h(x) \to x.
  \]
  Furthermore, let $U$ be the set of all redex occurrences in $t$,
  viz.\ $U = \set{\seq{0},\seq{0}^3}$. The following path $\Pi$ is a
  $U,\calR$-path in $t$:
  \begin{align*}
    \node{(\top,\emptyseq)}[g] &\edge[0] \node{(\top,\seq{0})}[\emptylab]
    \edge[\emptylab] \node{(r_1,\emptyseq,\seq{0})}[h] \edge[0]
    \node{(r_1,\seq{0},\seq{0})}[\emptylab] \edge[\emptylab]
    \node{(\top,\seq{0}^2)}[g] \\ &\edge[0] \node{(\top,\seq{0}^3)}[\emptylab]
    \edge[\emptylab] \node{(r_2,\emptyseq,\seq{0}^3)}[\emptylab]
    \edge[\emptylab] \node{(\top,\seq{0}^4)}[\bot]
  \end{align*}
  As a matter of fact, $\Pi$ is the greatest path of $t$. Hence,
  according to Fact~\ref{fact:emptyEdge}, the set of all prefixes of
  $\Pi$ ending in a node is the set of all $U,\calR$-paths in
  $t$. Note that since $\Pi$ itself ends in a labelled node, it is
  contained in $\paths{t}{U}{\calR}$. The trace $\trace{t}{\Pi}$ of $\Pi$
  is the sequence
  \[
  \seq{g, 0, h, 0, g, 0}
  \]

  Now consider the term $t' = g(f(g(h^\omega)))$ and the set $U'$ of
  all its redexes, viz.\ $U' =
  \set{\seq{0}}\cup\set{\seq{0}^3,\seq{0}^4,\dots}$. Then the
  following path $\Pi'$ is a $U,\calR$-path in $t'$:
  \begin{align*}
    \node{(\top,\emptyseq)}[g] &\edge[0] \node{(\top,\seq{0})}[\emptylab]
    \edge[\emptylab] \node{(r_1,\emptyseq,\seq{0})}[h] \edge[0]
    \node{(r_1,\seq{0},\seq{0})}[\emptylab] \edge[\emptylab]
    \node{(\top,\seq{0}^2)}[g] \edge[0] \node{(\top,\seq{0}^3)}[\emptylab]
    \\ &\edge[\emptylab] \node{(r_2,\emptyseq,\seq{0}^3)}[\emptylab]
    \edge[\emptylab] \node{(\top,\seq{0}^4)}[\emptylab] \edge[\emptylab] \node{(r_2,\emptyseq,\seq{0}^4)}[\emptylab]
    \edge[\emptylab] \node{(\top,\seq{0}^5)}[\emptylab] \edge[\emptylab] \dots
  \end{align*}
  $\Pi'$ is the greatest path of $t'$. The trace $\trace{t'}{\Pi'}$ of
  $\Pi'$ is the sequence
  \[
  \seq{g, 0, h, 0, g, 0}
  \]
  Since $\Pi'$ is infinitely long but has a finite trace, it is
  contained in $\paths{t'}{U}{\calR}$.
\end{exa}

The lemma below shows that there is a one-to-one correspondence
between paths in $\paths{t}{U}{\calR}$ and their traces in
$\traces{t}{U}{\calR}$.
\begin{lem}[$\trace{t}{\cdot}$ is a bijection]
  \label{lem:traceBij}
  % $\trace{\cdot}$ is a bijection %
  Let $\calR$ be an orthogonal TRS, $t$ a partial term in $\calR$, and
  $U$ a set of redex occurrences in $t$. $\trace{t}{\cdot}$ is a
  bijection from $\paths{t}{U}{\calR}$ to $\traces{t}{U}{\calR}$.
\end{lem}
\begin{proof}
  By definition, $\trace{t}{\cdot}$ is surjective. Let $\Pi_1, \Pi_2$ be
  two paths having the same trace. We will show that then $\Pi_1 =
  \Pi_2$ by an induction on the length of the common trace.

  Let $\trace{t}{\Pi_1} = \emptyseq$. Following Fact~\ref{fact:emptyEdge},
  there are two different cases: The first case is that $\Pi_1 =
  \Pi\concat \node{(\top, \pi)}[\bot]$, where the prefix $\Pi$
  corresponds to a finite maximal collapsing tower $(u_i)_{i \le
    \alpha}$ starting at the root of $t$ or $\Pi$ is empty if such a
  collapsing tower does not exists. If the collapsing tower exists,
  then
  \[
  \Pi = \node{(\top,u_0)}[\emptylab] \edge[\emptylab]
  \node{(r_0,\emptyseq,u_0)}[\emptylab] \edge[\emptylab]
  \node{(\top, u_1)}[\emptylab] \edge[\emptylab]
  \node{(r_1,\emptyseq,u_1)}[\emptylab] \edge[\emptylab] 
  \dots \edge[\emptylab]
  \node{(\top, u_\alpha)}[\emptylab] \edge[\emptylab]
  \]
  But then also $\Pi_2$ starts with the prefix $\Pi \concat (\top,\pi)$
  due to the uniqueness of the collapsing tower and the involved
  rules. In both cases, $\Pi_1 = \Pi_2$ follows immediately.

  The second case is that $\Pi_1$ is infinite. Then there is an
  infinite collapsing tower $(u_i)_{i < \omega}$ starting at the root
  of $t$. Hence,
  \[
  \Pi_1 = \node{(\top,u_0)}[\emptylab] \edge[\emptylab]
  \node{(r_0,\emptyseq,u_0)}[\emptylab] \edge[\emptylab]
  \node{(\top, u_1)}[\emptylab] \edge[\emptylab]
  \node{(r_1,\emptyseq,u_1)}[\emptylab] \edge[\emptylab] \dots
  \]
  $\Pi_1 = \Pi_2$ follows from the uniqueness of the infinite
  collapsing tower.

  At first glance one might additionally find a third case where
  $\Pi_1 = \Pi \concat \node{(\top,\pi)}[\emptylab] \edge[\emptylab]
  \node{(r,\emptyseq,\pi)}[\bot]$ with $\Pi$ a prefix corresponding to
  a collapsing tower as in the first case. However, this is not
  possible as it would require the occurrence of $\bot$ in a rule.

  Let $\trace{t}{\Pi_1} = f$. Then there are two cases: Either $\Pi_1 =
  \Pi\concat \node{(\top, \pi)}[f]$ or $\Pi_1 = \Pi \concat
  \node{(\top,\pi)}[\emptylab] \edge[\emptylab]
  \node{(r,\emptyseq,\pi)}[f]$, where the prefix $\Pi$ corresponds to a
  finite maximal collapsing tower $(u_i)_{i \le \alpha}$ starting at
  the root of $t$ or $\Pi$ is empty if such a collapsing tower does
  not exists. The argument is analogous to the argument employed for
  the first case of the induction base above.

  Finally, we consider the induction step. Hence, there are the two
  cases: Either $\trace{t}{\Pi_1} = T \concat \seq{i}$ or
  $\trace{t}{\Pi_1} = T \concat \seq{i, f}$. For both cases,
  the induction hypothesis can be invoked by taking two prefixes
  $\Pi'_1$ and $\Pi'_2$ of $\Pi_1$ and $\Pi_2$, respectively, which
  both have the trace $T$ and, therefore, are equal according to the
  induction hypothesis. The argument that the remaining suffixes of
  $\Pi_1$ and $\Pi_2$ are equal is then analogous to the argument for
  two base cases.
\end{proof}

As mentioned above, the traces of paths contain all information
necessary to define a term which we will later identify to be the
final term of the corresponding complete development. The following
definition explains how such a term, called a \emph{matching term}, is
determined:
\begin{defi}[matching term]
  Let $\calR$ be a left-linear TRS, $t$ a partial term in $\calR$, and
  $U$ a set of pairwise non-conflicting redex occurrences in $t$.
  \begin{enumerate}[label=(\roman*)]
  \item The \emph{position} of a trace $T \in \traces{t}{U}{\calR}$,
    denoted $\postrace{T}$, is the subsequence of $T$ containing only
    the edge labels. The set of all positions of traces in
    $\traces{t}{U}{\calR}$ is denoted $\postraces{t}{U}{\calR}$.
  \item The \emph{symbol} of a trace $T \in \traces{t}{U}{\calR}$,
    denoted $\symtrace{t}{T}$, is $f$ if $T$ ends in a node label $f$,
    and is $\bot$ otherwise, i.e.\ whenever $T$ is empty or ends in an
    edge label.
  \item A term $t'$ is said to \emph{match} $\traces{t}{U}{\calR}$ if
    $\pos{t'} = \postraces{t}{U}{\calR}$ and $t'(\postrace{T}) =
    \symtrace{t}{T}$ for all $T \in \traces{t}{U}{\calR}$.
  \end{enumerate}
\end{defi}

\noindent Returning to the definition of paths, one can see that the label of a
node is the symbol of the ``current'' position in a term. Similarly,
the label of an edge says which edge in the term tree was taken at
that point in the construction of the path. Hence, by projecting to
the edge labels, we obtain the ``history'' of the path, i.e.\ the
position. In the same way we obtain the symbol of that node by taking
the label of the last node of the path, provided the corresponding
path ends in a non-$\bot$-labelled node. In the other case that the
trace does not end in a node label, the corresponding path either ends
in a node labelled $\bot$ or is infinite. As we will see, infinite
paths with finite traces correspond to infinite collapsing towers,
which in turn yield volatile positions within the complete
development. Eventually, these volatile positions will also give rise
to $\bot$ subterms.

The following lemma shows that there is also a one-to-one
correspondence between the traces in $\traces{t}{U}{\calR}$ and their
positions in $\postraces{t}{U}{\calR}$:
\begin{lem}[$\postrace{\cdot}$ is a bijection]
  \label{lem:postraceBij}
  % $\postrace{\cdot}$ is a bijection %
  Let $\calR$ be an orthogonal TRS, $t$ a partial term in $\calR$ and
  $U$ a set of redex occurrences in $t$. $\postrace{\cdot}$ is a
  bijection from $\traces{t}{U}{\calR}$ to $\postraces{t}{U}{\calR}$.
\end{lem}
\begin{proof}
  An argument similar to the one for Lemma~\ref{lem:traceBij} can be
  given in order to show that the composition
  $\postrace{\cdot}\circ\trace{t}{\cdot}$ is a bijection. Together with
  the bijectivity of $\trace{s}{\cdot}$, according to
  Lemma~\ref{lem:traceBij}, this yields the bijectivity of
  $\postrace{\cdot}$.
\end{proof}

Having this lemma, the following proposition is an easy consequence of
the definition of matching terms. It shows that matching terms do
always exists and are uniquely determined:
\begin{prop}[unique matching term]
  Let $\calR$ be an orthogonal TRS, $t$ a partial term in $\calR$, and
  $U$ a set of redex occurrences in $t$. Then there is a unique term,
  denoted $\devTerm{t}{U}{\calR}$, that
  matches $\traces{t}{U}{\calR}$.
\end{prop}
\begin{proof}
  Define the mapping $\phi\fcolon \postraces{t}{U}{\calR} \funto
  \Sigma_\bot \cup \calV$ by setting $\phi(\postrace{T}) =
  \symtrace{t}{T}$ for each trace $T \in \traces{t}{U}{\calR}$. By
  Lemma~\ref{lem:postraceBij}, $\phi$ is well-defined. Moreover, it is
  easy to see from the definition of paths, that
  $\postraces{t}{U}{\calR}$ is closed under prefixes and that $\phi$
  respects the arity of the symbols, i.e.\ $\pi\concat i \in
  \postraces{t}{U}{\calR}$ iff $0 \le i < \srank{\phi(\pi)}$. Hence,
  $\phi$ uniquely determines a term $s$ with $s(\pi) = \phi(\pi)$ for
  all $\pi \in \postraces{t}{U}{\calR}$. By construction, $s$ matches
  $\traces{t}{U}{\calR}$. Moreover, any other term $s'$ matching
  $\traces{t}{U}{\calR}$ must satisfy $s'(\pi) = \phi(\pi)$ for all
  $\pi \in \postraces{t}{U}{\calR}$ and is therefore equal to $s$.
\end{proof}

It is also obvious that the matching term of a term $t$ w.r.t.\ an
empty set of redex occurrences is the term $t$ itself.
\begin{lem}[matching term w.r.t.\ empty redex set]
  \label{lem:matchEmpty}
  % matching term w.r.t.\ empty redex set %
  For any TRS $\calR$ and any partial term $t$ in $\calR$, it holds
  that $\devTerm{t}{\emptyset}{\calR} = t$.
\end{lem}
\begin{proof}
  Straightforward.
\end{proof}

\begin{rem}
  \label{rem:invDevTerm}
  Now it only remains to be shown that the matching term stays
  invariant during a development, i.e.\ that, for each development
  $S\fcolon t \pato t'$ of $U$, the matching terms
  $\devTerm{t}{U}{\calR}$ and $\devTerm{t'}{\dEsc{U}{S}}{\calR}$
  coincide. Since the matching term $\devTerm{t}{U}{\calR}$ only
  depends on the set $\traces{t}{U}{\calR}$ of traces, it is
  sufficient to show that $\traces{t}{U}{\calR}$ and
  $\traces{t'}{\dEsc{U}{S}}{\calR}$ coincide. The key observation is
  that in each step $s \to s'$ in a development the paths in $s'$
  differ from the paths in $s$ only in that they might omit some
  jumps. This can be seen in Figure~\ref{fig:path1}: In a step $s \to
  s'$ of a development, (some residual of) some redex occurrence in
  $U$ is contracted. In the picture this corresponds to removing the
  pattern, say $l_1$, of the redex and replacing it by the
  corresponding right-hand side $r_1$ of the rule. One can see that,
  except for the jump to and from the right-hand side $r_1$ the path
  remains the same.
\end{rem}

In order to establish the above observation formally, we need a means
to simulate reduction steps in a development directly as an operation
on paths. The following definition provides a tool for this.
\begin{defi}[position and prefix of a path]
  Let $\calR$ be a left-linear TRS, $t$ a partial term in $\calR$, $U$
  a set of pairwise non-conflicting redex occurrences in $t$, and $\Pi
  \in \paths{t}{U}{\calR}$.
  \begin{enumerate}[label=(\roman*)]
  \item $\Pi$ is said to \emph{contain} a position $\pi \in \pos{t}$
    if it contains the node $(\top,\pi)$.
  \item For each $u \in U$, the \emph{prefix} of $\Pi$ by $u$, denoted
    $\Pi^{(u)}$, is defined as $\Pi$ whenever $\Pi$ does not contain
    $u$ and otherwise as the unique prefix of $\Pi$ that ends in
    $(\top,u)$.
  \end{enumerate}
\end{defi}

\begin{rem}
  It is obvious from the definition that each prefix $\Pi^{(u)}$ of a
  path $\Pi\in \paths{t}{U}{\calR}$ by an occurrence $u$ is the
  maximal prefix of $\Pi$, that does not contain positions that are
  proper extensions of $u$. Hence, if $\Pi$ contains $u$, then
  $\Pi^{(u)}$ is the maximal prefix of $\Pi$ that only contains
  prefixes of $u$ (including $u$ itself).
\end{rem}

The following lemma is the key step towards proving the invariance of
matching terms in developments. It formalises the observation
described in Remark~\ref{rem:invDevTerm}.
\begin{lem}[preservation of traces]
  \label{lem:presPath}
  % preservation of paths %
  Let $\calR$ be an orthogonal TRS, $t$ a partial term in $\calR$, $U$
  a set of redex occurrences in $t$, and $S\fcolon t \pato t'$ a
  development of $U$ in $t$. There is a surjective mapping
  $\theta_S\fcolon \paths{t}{U}{\calR} \funto
  \paths{t'}{\dEsc{U}{S}}{\calR}$ such that $\trace{t}{\Pi} =
  \trace{t'}{\theta_S(\Pi)}$ for all $\Pi \in \paths{t}{U}{\calR}$.
\end{lem}
\begin{proof}
  Let $S = (t_\iota \to[\pi_\iota,c_\iota] t_{\iota + 1})_{\iota
    < \alpha}$. We prove the statement by an induction on $\alpha$.

  If $\alpha = 0$, then the statement is trivially true.

  Suppose that $\alpha$ is a successor ordinal $\beta + 1$. Let
  $T\fcolon t_0 \pto{\beta} t_\beta$ be the prefix of $S$ of length
  $\beta$ and $\phi_\beta\fcolon t_\beta \to[\pi_\beta] t_\alpha$ the
  last step of $S$, i.e.\ $S = T \concat \seq{\phi_\beta}$. By the
  induction hypothesis, there is a surjective mapping $\theta_T\fcolon
  \paths{t}{U}{\calR} \funto \paths{t_\beta}{U'}{\calR}$, with $U' =
  \dEsc{U}{T}$ and $\trace{t}{\Pi} = \trace{t_\beta}{\theta_T(\Pi)}$
  for all $\Pi \in \paths{t}{U}{\calR}$. By a careful case analysis
  (as done in \cite{ketema11ic}), one can show that there is a
  surjective mapping $\theta\fcolon \paths{t_\beta}{U'}{\calR} \funto
  \paths{t_\alpha}{U''}{\calR}$, with $U'' =
  \dEsc{U'}{\seq{\phi_\beta}} = \dEsc{U}{S}$ and $\trace{t_\beta}{\Pi}
  = \trace{t_\alpha}{\theta(\Pi)}$ for all $\Pi \in
  \paths{t_\beta}{U'}{\calR}$. Hence, the composition $\theta_S =
  \theta \circ \theta_T$ is a surjective mapping from
  $\paths{t}{U}{\calR}$ to $\paths{t'}{\dEsc{U}{S}}{\calR}$ and
  satisfies $\trace{t}{\Pi} = \trace{t'}{\theta_S(\Pi)}$ for all $\Pi
  \in \paths{t}{U}{\calR}$.

  Let $\alpha$ be a limit ordinal. By induction hypothesis, there is a
  surjective mapping $\theta_{\prefix{S}{\iota}}$ for each proper
  prefix $\prefix{S}{\iota}$ of $S$ satisfying $\trace{t_0}{\Pi} =
  \trace{t_\iota}{\theta_{\prefix{s}{\iota}}(\Pi)}$ for all $\Pi \in
  \paths{t}{U}{\calR}$. Let $\Pi \in \paths{t}{U}{\calR}$ and
  $\Pi_\iota = \theta_{\prefix{S}{\iota}}(\Pi)$ for each $\iota <
  \alpha$. We define $\theta_S(\Pi)$ as follows:
  \[
  \theta_S(\Pi) = \liminf_{\iota \limto \alpha} \Pi_\iota^{(\pi_\iota)}
  \]

  At first we have to show that $\theta_S$ is well-defined, i.e.\ that
  $\liminf_{\iota \limto \alpha} \Pi_\iota^{(\pi_\iota)}$ is indeed a
  path in $\paths{t'}{\dEsc{U}{S}}{\calR}$, and that it preserves
  traces. There are two cases to be considered: If there is an
  outermost-volatile position $\pi$ in $S$ that is contained in
  $\Pi_\iota$ whenever $\pi_\iota = \pi$, then there is some $\beta <
  \alpha$ with $\pi_\iota \not< \pi$ for all $\beta \le \iota <
  \alpha$. Hence, $\theta_S(\Pi) = \Pi_\beta^{(\pi)}$. By
  Lemma~\ref{lem:nonBotLimRed} and Lemma~\ref{lem:botLimRed}, we have
  that $\Pi_\beta^{(\pi)} \in \paths{t'}{\dEsc{U}{S}}{\calR}$, in
  particular because $t'(\pi) = \bot$. Since the suffix $\Pi'$ with
  $\Pi_\beta = \Pi_\beta^{(\pi)} \concat \Pi'$ follows an infinite
  collapsing tower and is therefore entirely unlabelled, it cannot
  contribute to the trace of $\Pi_\beta$. Consequently,
  \[
  \trace{t}{\Pi} \stackrel{IH}{=} \trace{t_\beta}{\Pi_\beta} =
  \trace{t'}{\Pi_\beta^{(\pi)}} = \trace{t'}{\theta_{S}(\Pi)}.
  \]
  If, on the other hand, there is no such outermost-volatile position,
  then either the sequence $(\Pi_\iota^{(\pi_\iota)})_{\iota<\alpha}$
  becomes stable at some point or the sequence $(\Glb_{\iota<\gamma}
  \Pi_\iota^{(\pi_\iota)})_{\gamma<\alpha}$ grows monotonically
  towards the infinite path $\theta_S(\Pi)$. In both cases
  well-definedness and preservation of traces follows easily from the
  induction hypothesis.

  Lastly, we show the surjectivity of $\theta_S$. To this end, assume
  some $\Pi \in \paths{t'}{\dEsc{U}{S}}{\calR}$. We show the existence
  of a path $\ol \Pi \in \paths{t}{U}{\calR}$ with $\theta_S(\ol \Pi)
  = \Pi$ by distinguishing three cases:

  \begin{enumerate}[label=(\alph*)]
  \item%
  \label{item:presPathA}%
  $\Pi$ ends in a redex node $(r,\pi,u)$. Hence, $u \in
    \dEsc{U}{S}$. According to Lemma~\ref{lem:descLimRed}, this means
    that there is some $\beta < \alpha$ such that
    \begin{gather}
      \label{eq:presPathI}
      \text{$\pi_\iota \not\le u$ for all $\beta \le \iota < \alpha$.}
      \tag{1}
    \end{gather}
    Consequently, all terms in $\setcom{t_\iota}{\beta \le \iota <
      \alpha}$ coincide in all prefixes of $u$, and each $v \in
    \dEsc{U}{S}$ with $v \le u$ is in $\dEsc{U}{\prefix{S}{\iota}}$
    for all $\beta \le \iota < \alpha$.  Hence, for all $\beta \le
    \gamma < \alpha$ we have $\Pi \in
    \paths{t_\gamma}{\dEsc{U}{\prefix{S}{\gamma}}}{\calR}$ with
    $\trace{t'}{\Pi}=\trace{t_\gamma}{\Pi}$. By induction hypothesis
    there is for each $\beta \le \gamma < \alpha$ some $\Pi_\gamma \in
    \paths{t}{U}{\calR}$ that is mapped to $\Pi \in
    \paths{t_\gamma}{\dEsc{U}{\prefix{S}{\gamma}}}{\calR}$ by
    $\theta_{\prefix{S}{\gamma}}$ with $\trace{t}{\Pi_\gamma} =
    \trace{t_\gamma}{\Pi}$. Hence, $\trace{t}{\Pi_\gamma} =
    \trace{t'}{\Pi}$ which means that all paths $\Pi_\gamma$, with
    $\beta \le \gamma < \alpha$, have the same trace in $t$ and are
    therefore equal according to Lemma~\ref{lem:traceBij}. Let us call
    this path $\ol \Pi$. That is, $\theta_{\prefix{S}{\gamma}}(\ol
    \Pi) = \Pi$ for all $\beta \le \gamma < \alpha$. Since $\pi_\gamma
    \not\le u$, we also have $(\theta_{\prefix{S}{\gamma}}
    (\ol\Pi))^{(\pi_\gamma)} = \Pi$. Consequently, $\theta_S(\ol\Pi) =
    \Pi$.

  \item $\Pi$ ends in a term node $(\top, \pi)$. Let $f = t'(\pi)$. If
    $f \neq \bot$, then we can apply Lemma~\ref{lem:nonBotLimRed} to
    obtain some $\beta < \alpha$ such that $\pi_\iota \not\le \pi$ for
    all $\beta \le \iota < \alpha$. Then we can reason as in case
    (\ref{item:presPathA}) starting from \eqref{eq:presPathI}. If $f =
    \bot$, then we have to distinguish two cases according to
    Lemma~\ref{lem:botLimRed}: If there is some $\beta < \alpha$ with
    $t_\beta(\pi) = \bot$ and $\pi_\iota \not\le \pi$ for all $\beta
    \le \iota < \alpha$, then we can again employ the same argument as
    for case (\ref{item:presPathA}) starting from
    \eqref{eq:presPathI}. Otherwise, i.e.\ if $\pi$ is an
    outermost-volatile position in $S$, then we have some $\beta <
    \alpha$ such that $\pi_\iota \not< \pi$ for all $\beta \le \iota <
    \alpha$ and such that
    \begin{gather}
      \label{eq:presPathII}
      \text{for each $\beta \le \gamma < \alpha$ there is some $\gamma
        \le \gamma' < \alpha$ with $\pi_\gamma' = \pi$.}  \tag{2}
    \end{gather}
    Hence, we have for each $\beta \le \gamma < \alpha$ some
    $\Pi_\gamma \in
    \paths{t_\gamma}{\dEsc{U}{\prefix{S}{\gamma}}}{\calR}$ and an
    infinite collapsing tower $(u_i)_{i < \omega}$ in
    $\dEsc{U}{\prefix{S}{\gamma}}$ with $u_0 = \pi$ such that
    $\Pi_\gamma$ is of the form
    \[
    \Pi \concat \edge[\emptylab] \node{(r_0,\emptyseq,u_0)}[\emptylab]
    \edge[\emptylab] \node{(\top, u_1)}[\emptylab] \edge[\emptylab]
    \node{(r_1,\emptyseq,u_1)}[\emptylab] \edge[\emptylab] \dots
    \]
    Therefore, $\trace{t_\gamma}{\Pi_\gamma} = \trace{t'}{\Pi}$. By
    induction hypothesis there is some $\ol \Pi_\gamma \in
    \paths{t}{T}{\calR}$ with
    $\theta_{\prefix{S}{\gamma}}(\ol\Pi_\gamma)=\Pi_\gamma$ and
    $\trace{t}{\ol\Pi_\gamma} = \trace{t_\gamma}{\Pi_\gamma}$. Hence,
    $\trace{t}{\ol\Pi_\gamma} = \trace{t'}{\Pi}$, i.e.\ all
    $\ol\Pi_\gamma$ have the same trace in $t$ and are therefore equal
    according to Lemma~\ref{lem:traceBij}. Let us call this path $\ol
    \Pi$. Since
    $(\theta_{\prefix{S}{\gamma}}(\ol\Pi))^{(\pi)}=\Pi_\gamma^{(\pi)}=
    \Pi$ we can use \eqref{eq:presPathII} to obtain that
    $\theta_S(\ol\Pi) = \Pi$.

  \item $\Pi$ is infinite. Hence, $\Pi$ is of the form
    \[
    \Pi' \concat \node{(\top, u_0)}[\emptylab] \edge[\emptylab]
    \node{(r_0,\emptyseq,u_0)}[\emptylab] \edge[\emptylab]
    \node{(\top, u_1)}[\emptylab] \edge[\emptylab]
    \node{(r_1,\emptyseq,u_1)}[\emptylab] \edge[\emptylab] \dots
    \]
    with $(u_i)_{i<\omega}$ an infinite collapsing tower in
    $\dEsc{U}{S}$. Consequently, by Lemma~\ref{lem:descLimRed}, for
    each $u_i \in \dEsc{U}{S}$ there is some $\beta_i < \alpha$ such
    that
    \begin{gather}
      \label{eq:presPathIII}
      \text{$u_i \in \dEsc{U}{\prefix{S}{\gamma}}$ and $\pi_\gamma
        \not\le u_\gamma$ for all $\beta_i \le \gamma < \alpha$.}
      \tag{3}
    \end{gather}
    Since $(u_i)_{i<\omega}$ is a chain (w.r.t.\ the prefix order), we
    can assume w.l.o.g.\ that $(\beta_i)_{i<\omega}$ is a chain as
    well. Following Remark~\ref{rem:prsAncestor}, we obtain for each
    $u_i \in \dEsc{U}{S}$ its ancestor $v_i \in U$ with $\dEsc{v_i}{S}
    = u_i$. Let $\ol \Pi$ be the unique path in $\paths{t}{U}{\calR}$
    that contains each $v_i$ and for each $j < \omega$ let $\Pi_j$ be
    the unique path in
    $\paths{t_{\beta_j}}{\dEsc{U}{\prefix{S}{\beta_j}}}{\calR}$
    containing each $\dEsc{v_i}{\prefix{S}{\beta_j}}$. Clearly,
    $\theta_{\prefix{S}{\beta_j}}(\ol\Pi) = \Pi_j$. Note that we have
    for each $j < \omega$ that all paths
    $\theta_{\prefix{S}{\iota}}(\ol\Pi)$ with $\beta_j \le \iota <
    \alpha$ coincide in their prefix by $u_j$, which is a prefix of
    $\Pi$. Since additionally $(u_i)_{i<\omega}$ is a strict chain and
    because of \eqref{eq:presPathIII}, we can conclude that
    $\theta_S(\ol\Pi) = \Pi$.\qedhere
  \end{enumerate}

\end{proof}

\noindent The above lemma effectively establishes the invariance of matching
terms during a development. Together with Lemma~\ref{lem:matchEmpty}
this implies the uniqueness of final terms of complete developments of
the same redex occurrences. As a corollary from this, we obtain that
descendants are also unique among all complete developments:
\begin{prop}[final term and descendants of complete
  developments]
  \label{prop:finalCompDev}
  Let $\calR$ be an orthogonal TRS, $t$ a partial term in $\calR$, and
  $U$ a set of redex occurrences in $t$. Then the following holds:
  \begin{enumerate}[label=(\roman*)]
  \item Each complete development of $U$ in $t$ strongly
    $\prs$-converges to $\devTerm{t}{U}{\calR}$.
    \label{item:finalCompDev1}
  \item For each set $V \subseteq \posNonBot{t}$ and two complete
    developments $S$ and $T$ of $U$ in $t$, respectively, it holds
    that $\dEsc{V}{S} = \dEsc{V}{T}$.
    \label{item:finalCompDev2}
  \end{enumerate}
\end{prop}
\begin{proof}
  \def\itema{(\ref{item:finalCompDev1})}
  \def\itemb{(\ref{item:finalCompDev2})}
  \itema{} Let $S\fcolon t \pato[U] t'$ be a complete development of
  $U$ in $t$ strongly $\prs$-converging to $t'$. By
  Lemma~\ref{lem:presPath}, there is a surjective mapping
  $\theta\fcolon \paths{t}{U}{\calR} \funto \paths{t'}{U'}{\calR}$
  with $\trace{t}{\Pi} = \trace{t'}{\theta(\Pi)}$ for all $\Pi \in
  \paths{t}{U}{\calR}$, where $U' = \dEsc{U}{S}$. Hence, it holds that
  $\traces{t}{U}{\calR} = \traces{t'}{U'}{\calR}$ and, consequently,
  $\devTerm{t}{U}{\calR} = \devTerm{t'}{U'}{\calR}$. Since $S$ is a
  complete development of $U$ in $t$, we have that $U' = \emptyset$
  which implies, according to Lemma~\ref{lem:matchEmpty}, that
  $\devTerm{t'}{U'}{\calR} = t'$. Therefore, $\devTerm{t}{U}{\calR} =
  t'$.

  \itemb{} Let $t' = t^{(V)}$. By Proposition~\ref{prop:chaDesc}, both
  reductions $S$ and $T$ can be uniquely lifted to reductions $S'$ and
  $T'$ in $\calR^\lab$, respectively, such that $\dEsc{V}{S}$ and
  $\dEsc{V}{T}$ are determined by the final term of $S'$ and $T'$,
  respectively. It is easy to see that also $\calR^\lab$ is an orthogonal
  TRS and that $S'$ and $T'$ are complete developments of $U$ in
  $t'$. Hence, we can invoke clause \itema{} of this proposition to
  conclude that the final terms of $S'$ and $T'$ coincide and that,
  therefore, also $\dEsc{V}{S}$ and $\dEsc{V}{T}$ coincide.
\end{proof}

By the above proposition, the descendants of a complete development of
a particular set of redex occurrences are unique. Therefore, we adopt
the notation $\dEsc{U}{V}$ for the descendants $\dEsc{U}{S}$ of $U$ by
some complete development $S$ of $V$. According to
Proposition~\ref{prop:exComplDev} and
Proposition~\ref{prop:finalCompDev}, $\dEsc{U}{V}$ is well-defined for
any orthogonal TRS.

Furthermore, Proposition~\ref{prop:finalCompDev} yields the following
corollary establishing the diamond property of complete developments:
\begin{cor}[diamond property of complete developments]
  \label{cor:prsCRCompDev}
  % generalisation of [kennaway95ic, Cor. 4.7]
  Let $\calR$ be an orthogonal TRS and $t \pato[U] t_1$ and $t
  \pato[V] t_2$ be two complete developments of $U$ respectively $V$
  in $t$. Then $t_1$ and $t_2$ are joinable by complete developments
  $t_1 \pato[\dEsc{V}{U}] t'$ and $t_2 \pato[\dEsc{U}{V}] t'$.
\end{cor}
\begin{proof}
  By Proposition~\ref{prop:descPoint}, it holds that 
  \[
  \dEsc{(U\cup
    V)}{U} = \dEsc{U}{U} \cup \dEsc{V}{U} = \dEsc{V}{U}.
  \]
  Let $S\fcolon t \pato[U] t_1$, $T\fcolon t \pato[V] t_2$, $S'\fcolon
  t_1 \pato[\dEsc{V}{U}] t'$ and $T'\fcolon t_2 \pato[\dEsc{U}{V}]
  t''$. By the equation above and Proposition~\ref{prop:descSeqRed},
  we have that $S\concat S'\fcolon t \pato[U] t_1 \pato[\dEsc{V}{U}]
  t'$ is a complete development of $U \cup V$. Analogously, we obtain
  that $T\concat T'\fcolon t \pato[V] t_2 \pato[\dEsc{U}{V}] t''$ is a
  complete development of $U\cup V$, too. According to
  Proposition~\ref{prop:finalCompDev}, this implies that both
  $S\concat S'$ and $T\concat T'$ strongly $\prs$-converge in the same
  term, i.e.\ $t' = t''$.
\end{proof}

In the next section we shall make use of complete developments in
order to obtain the Infinitary Strip Lemma for $\prs$-converging
reductions and a limited form of infinitary confluence for orthogonal
systems.

\subsection{The Infinitary Strip Lemma}
\label{sec:results}
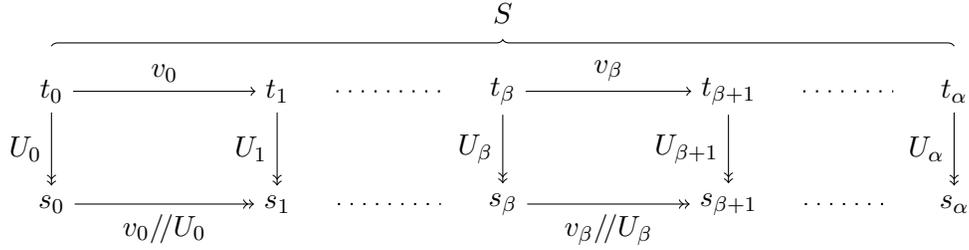
\begin{figure}
  \centering
  \begin{tikzpicture}[on grid,dots/.style={shorten
      >=.5cm,shorten <=.5cm,loosely dotted,thick}]
    \begin{scope}[node distance=3cm]
      \node (t0) {$t_0$}%
      node[right=of t0] (t1) {$t_1$}%
      node[right=of t1] (tb) {$t_\beta$}%
      node[right=of tb] (tb1) {$t_{\beta+1}$}%
      node[right=of tb1] (ta) {$t_\alpha$};
    \end{scope}

    \begin{scope}[node distance=1.5cm]
      \node[below=of t0] (s0) {$s_0$}%
      node[below=of t1] (s1) {$s_1$}%
      node[below=of tb] (sb) {$s_\beta$}%
      node[below=of tb1] (sb1) {$s_{\beta+1}$}%
      node[below=of ta] (sa) {$s_\alpha$};
    \end{scope}

    \draw%
    (t0)%
    edge[single step] node[midway,above] {$v_0$} (t1)%
    edge[strongly] node[midway,left] {$U_0$} (s0)%
    (t1)%
    edge[dots] (tb)%
    edge[strongly] node[midway,left] {$U_1$} (s1)%
    (tb)%
    edge[single step] node[midway,above] {$v_\beta$} (tb1)%
    edge[strongly] node[midway,left] {$U_\beta$} (sb)%
    (tb1)%
    edge[dots] (ta)%
    edge[strongly] node[midway,left] {$U_{\beta + 1}$} (sb1)%
    (ta) edge[strongly] node[midway,left] {$U_\alpha$} (sa)%

    (s0) edge[strongly] node[midway,below] {$\dEsc{v_0}{U_0}$} (s1)%
    (s1) edge[dots] (sb)%
    (sb) edge[strongly] node[midway,below] {$\dEsc{v_\beta}{U_\beta}$}
    (sb1)%
    (sb1) edge[dots] (sa)%
    ;

    \draw ($(t0)+(0,.6)$) edge[decorate,decoration=brace]
    node[midway,above=5pt] {$S$}
      ($(ta)+(0,.6)$);%
  \end{tikzpicture}
\caption{The Infinitary Strip Lemma.}
\label{fig:stripLem}
\end{figure}

In this section we use the results we have obtained for complete
developments in the previous two sections in order to establish that a
complete development of a set of pairwise disjoint redex occurrences
commutes with any strongly $\prs$-convergent reduction:
\begin{prop}[Infinitary Strip Lemma]
  \label{prop:prsStripLem}
  % generalisation of [kennaway95ic, Lem. 4.8]
  Let $\calR$ be an orthogonal TRS, $S\fcolon t_0 \pto{\alpha}
  t_\alpha$ a strongly $\prs$-convergent reduction, and $t_0 \pato[U]
  s_0$ a complete development of a set $U$ of pairwise disjoint redex
  occurrences in $t_0$. Then $t_\alpha$ and $s_0$ are joinable by a
  reduction $\proj{S}{T}\fcolon s_0 \pato s_\alpha$ and a complete
  development $\proj{T}{S}\fcolon t_\alpha \pato[\dEsc{U}{S}]
  s_\alpha$.
\end{prop}
\begin{proof}
  We prove this statement by constructing the diagram shown in
  Figure~\ref{fig:stripLem}. The '$U_\iota$'s in the diagram are
  sets of redex occurrences: $U_\iota = \dEsc{U}{\prefix{S}{\iota}}$
  for all $0 \le \iota \le \alpha$. In particular, $U_0 = U$. All
  arrows in the diagram represent complete developments of the
  indicated sets of redex occurrences. Particularly, in each
  $\iota$-th step of $S$ the redex at $v_\iota$ is contracted. We will
  construct the diagram by an induction on $\alpha$.

  If $\alpha=0$, then the diagram is trivial. If $\alpha$ is a
  successor ordinal $\beta + 1$, then we can take the diagram for the
  prefix $\prefix{S}{\beta}$, which exists by induction hypothesis,
  and extend it to a diagram for $S$. The existence of the additional
  square that completes the diagram for $S$ is affirmed by
  Corollary~\ref{cor:prsCRCompDev} since $U_{\beta + 1} =
  \dEsc{U_\beta}{v_\beta}$.

  Let $\alpha$ be a limit ordinal. Moreover, let $s_\alpha'$ be the
  uniquely determined final term of a complete development of
  $U_\alpha$ in $t_\alpha$. By induction hypothesis, the diagram
  exists for each proper prefix of $S$. Let $T_\iota\fcolon s_0 \pato
  s_\iota$ denote the reduction at the bottom of the diagram for the
  reduction $\prefix{S}{\iota}$ for each $\iota < \alpha$. The set of
  all $T_\iota$ is directed. Hence, $T = \Lub_{\iota < \alpha}
  T_\iota$ exists. Since $T_\iota < T$ for each $\iota < \alpha$, the
  diagram for $S$ with $T\fcolon s_0 \pato s_\alpha$ at the bottom
  satisfies almost all required properties. Only the equality of
  $s_\alpha$ and $s_\alpha'$ remains to be shown.

  Note that, by Proposition~\ref{prop:disjDesc}, the redex occurrences
  in $U_\alpha$ are pairwise disjoint. Let $\pi \in U_\alpha$. By
  Lemma~\ref{lem:descLimRed} and the definition of descendants, there
  is some $\beta < \alpha$ such that $\pi \in U_\iota$ and $v_\iota
  \not\le \pi$ for all $\beta \le \iota < \alpha$. Hence, for all
  $\pi' \in \dEsc{v_\iota}{U_\iota}$ with $\beta \le \iota <\alpha$,
  we also have $\pi' \not\le \pi$. That is, in the remaining
  reductions $t_\beta \pato t_\alpha$ and $t_\beta \pato[U_\beta]
  s_\beta \pato s_\alpha$, no reduction takes place at a proper prefix
  of $\pi$. Hence, by Lemma~\ref{lem:nonBotLimRed}, $t_\beta$
  coincides with $t_\alpha$ and $s_\alpha$ in all proper prefixes of
  $\pi$. Since in the reduction $t_\alpha \pato[U_\alpha] s_\alpha'$
  also no reduction takes place at a proper prefix of $\pi$, we obtain
  that $t_\alpha$ and $s_\alpha'$ and, thus, also $s_\alpha$ and
  $s_\alpha'$ coincide in all proper prefixes of $\pi$.
  
  Let $\rho\fcolon l \to r$ be the rule for the redex
  $\atPos{t_\beta}{\pi}$ and $\Cxt{a}{,\dots,}, \Cxt[D]{a}{,\dots,}$
  ground contexts such that $l = \Cxt{a}{x_1,\dots,x_k}$ and $r =
  \Cxt[D]{a}{x_{p(1)},\dots,x_{p(m)}}$ for some pairwise distinct
  variables $x_1,\dots,x_k$ and an appropriate mapping $p\fcolon
  \set{1,\dots,m} \funto \set{1,\dots, k}$. Moreover, let
  $t^\iota_1,\dots,t^\iota_k$ be terms such that $t_\iota =
  \substAtPos{t_\iota}{\pi}{\Cxt[C]{a}{t^\iota_1,\dots,t^\iota_k}}$
  and $s_\iota =
  \substAtPos{s_\iota}{\pi}{\Cxt[D]{a}{t^\iota_{p(1)},\dots,t^\iota_{p(m)}}}$
  for all $\beta \le \iota \le \alpha$. The argument in the previous
  paragraph justifies the assumption of these elements. From $\beta$
  onward, all horizontal reduction steps in the diagram take place
  within the contexts $\substAtPos{t_\iota}{\pi}{\cdot}$ and
  $\substAtPos{s_\iota}{\pi}{\cdot}$, respectively, or inside the
  terms $t^\iota_i$, and all vertical reductions take place within the
  contexts $\substAtPos{t_\iota}{\pi}{\Cxt[C]{a}{,\dots,}}$ and
  $\substAtPos{s_\iota}{\pi}{\Cxt[D]{a}{,\dots,}}$, respectively. In
  particular, we have $t_\alpha =
  \substAtPos{t_\alpha}{\pi}{\Cxt[C]{a}{t^\alpha_1,\dots,t^\alpha_k}}$
  and $s_\alpha =
  \substAtPos{s_\alpha}{\pi}{\Cxt[D]{a}{t^\alpha_{p(1)},\dots,t^\alpha_{p(m)}}}$. Let
  $t_\alpha \to[\pi] t'_\alpha$. This reduction contracts the redex
  $\Cxt[C]{a}{t^\alpha_1,\dots,t^\alpha_k}$ to the term
  $\Cxt[D]{a}{t^\alpha_{p(1)},\dots,t^\alpha_{p(m)}}$ using rule
  $\rho$. Note that a complete development $t_\alpha \pato[U_\alpha]
  s_\alpha'$ contracts, besides $\pi$, only redex occurrences disjoint
  with $\pi$. Hence, $t'_\alpha$ and $s'_\alpha$ coincide in all
  extensions of $\pi$. Since $t'_\alpha =
  \substAtPos{t_\alpha}{\pi}{\Cxt[D]{a}{t^\alpha_{p(1)},\dots,t^\alpha_{p(k)}}}$
  (and $s_\alpha =
  \substAtPos{s_\alpha}{\pi}{\Cxt[D]{a}{t^\alpha_{p(1)},\dots,t^\alpha_{p(m)}}}$),
  we can conclude that $s_\alpha$ and $s_\alpha'$ coincide in all
  extensions of $\pi$.

  Since the residual $\pi \in U_\alpha$ was chosen arbitrarily, the
  above holds for all elements in $U_\alpha$. That is, $s_\alpha$ and
  $s'_\alpha$ coincide in all prefixes and all extensions of elements
  in $U_\alpha$. It remains to be shown, that they also coincide in
  positions that are disjoint to all positions in $U_\alpha$. To this
  end, we only need to show that $t_\alpha$ and $s_\alpha$ coincide in
  these positions since the complete development $t_\alpha
  \pato[U_\alpha] s_\alpha'$ keeps positions disjoint with all
  positions in $U_\alpha$ unchanged. Let $\pi$ be such a position.

  Suppose $t_\alpha(\pi) = f \neq \bot$. By
  Lemma~\ref{lem:nonBotLimRed}, there is some $\beta < \alpha$ such
  that $t_\beta(\pi) = f$ and $v_\iota \not\le \pi$ for all $\beta \le
  \iota < \alpha$. Note that no prefix $\pi'$ of $\pi$ is in $U_\beta$
  since otherwise $\pi' \in U_\alpha$, by Lemma~\ref{lem:descLimRed},
  which contradicts the assumption that $\pi$ is disjoint to all
  positions in $U_\alpha$. Hence, $s_\beta(\pi) = f$ and $\pi' \not\le
  \pi$ for all $\pi' \in \dEsc{v_\iota}{U_\iota}$ and $\beta \le \iota
  < \alpha$, which means that no reduction step in $s_\beta \pato
  s_\alpha$ takes place at some prefix of $\pi$. Thus, we can
  conclude, according to Lemma~\ref{lem:nonBotLimRed}, that
  $s_\alpha(\pi) = f$. Similarly, one can show that $s_\alpha(\pi) = f
  \neq \bot$ implies $t_\alpha(\pi) = f$.

  Suppose $t_\alpha(\pi) = \bot$. Hence, according to
  Lemma~\ref{lem:botLimRed}, $\pi$ is outermost-volatile in $S$ or
  there is some $\beta < \alpha$ such that $t_\beta(\pi) = \bot$ and
  $v_\iota \not\le \pi$ for all $\beta \le \iota < \alpha$. For the
  latter case, we can argue as in the case for $t_\alpha(\pi) \neq
  \bot$ above. In the former case, $\pi$ is outermost-volatile in $T$
  as well. Thus, by applying Lemma~\ref{lem:botLimRed}, we obtain that
  $s_\alpha(\pi) = \bot$. A similar argument can be employed for the
  reverse direction. 
\end{proof}
The reduction $\proj{S}{T}$ constructed in the proof above is called
the \emph{projection} of $S$ by $T$. Likewise, the reduction
$\proj{T}{S}$ is called the \emph{projection} of $T$ by $S$. As a
corollary we obtain the following semi-infinitary confluence result:
\begin{cor}[semi-infinitary confluence]
  \label{cor:prsSemiConf}
  % semi-infinitary confluence %
  In every orthogonal TRS, two reductions $t \pato t_2$ and $t
  \fto* t_1$ can be joined by two reductions $t_2 \pato t_3$ and
  $t_1 \pato t_3$.
\end{cor}
\begin{proof}
  This can be shown by an induction on the length of the reduction $t
  \fto* t_1$. If it is empty, the statement trivially holds. The
  induction step follows from Proposition~\ref{prop:prsStripLem}.
\end{proof}

In the next section we shall, based on the Infinitary Strip Lemma,
show that strong $\prs$-reachability coincides with Böhm-reachability,
which then yields, amongst other things, full infinitary confluence of
orthogonal systems.

\section{Comparing Strong \texorpdfstring{$\prs$}{p}-Convergence and Böhm-Convergence}
\label{sec:relation-bohm-trees}

In this section we shall show the core result of this paper: For
orthogonal, left-finite TRSs, strong $\prs$-reachability and
Böhm-reachability w.r.t.\ the set $\rAct$ of root-active terms
coincide. As corollaries of that, leveraging the properties of
Böhm-convergence, we obtain both infinitary normalisation and
infinitary confluence of orthogonal systems in the partial order
model. Moreover, we will show that strong $\prs$-convergence also
satisfies the compression property.

The central step of the proof of the equivalence of both models of
infinitary rewriting is an alternative characterisation of root-active
terms which is captured by the following definition:
\begin{defi}[destructiveness, fragility]
  Let $\calR$ be a TRS.
  \begin{enumerate}[label=(\roman*)]
  \item A reduction $S\fcolon t \pato s$ is called \emph{destructive}
    if $\emptyseq$ is a volatile position in $S$.
  \item A partial term $t$ in $\calR$ is called \emph{fragile} if a
    destructive reduction starts in $t$.
  \end{enumerate}
\end{defi}

\noindent Looking at the definition, fragility seems to be a more general
concept than root-activeness: A term is fragile iff it admits a
reduction in which infinitely often a redex at the root is
contracted. For orthogonal TRSs, root-active terms are characterised
in almost the same way. The difference is that only total terms are
considered and that the stipulated reduction contracting infinitely
many root redexes has to be of length $\omega$. However, we shall show
the set of total fragile terms to be equal to the set of root-active
terms by establishing a compression lemma for destructive reductions.

Using Lemma~\ref{lem:botLimRed} we can immediately derive the
following alternative characterisations:
\begin{fact}[destructiveness, fragility]
  Let $\calR$ be a TRS.
  \begin{enumerate}[label=(\roman*)]
  \item A reduction $S\fcolon s \pato t$ is destructive iff
    $S$ is open and $t = \bot$
  \item A partial term $t$ in $\calR$ is fragile iff there is an open
    strongly $\prs$-convergent reduction $t \pato \bot$.
  \end{enumerate}
\end{fact}

\noindent One has to keep in mind, however, that a closed reduction to $\bot$ is
not destructive. Such a notion of destructiveness would include the
empty reduction from $\bot$ to $\bot$, and reductions that end with
the contraction of a collapsing redex as, for example, in the single
step reduction $f(\bot) \to \bot$ induced by the rule $f(x) \to
x$. Such reductions do not ``produce'' the term $\bot$. They are
merely capable of ``moving'' an already existent subterm $\bot$ by a
collapsing rule. In this sense, fragile terms are, according to
Lemma~\ref{lem:totalRed}, the only terms which can produce the term
$\bot$. This is the key observation for studying the relation between
strong $\prs$-convergence and Böhm-convergence.

In order to show that strong $\prs$-reachability and Böhm-reachability
w.r.t.\ $\rAct$ coincide we will proceed as follows: At first we will
show that strong $\prs$-reachability implies Böhm-reachability w.r.t.\
the set of total fragile terms, i.e. the fragile terms in
$\iterms$. From this we will derive a compression lemma for
destructive reductions. We will then use this to show that the set
$\rAct$ of root-active terms coincides with the set of total fragile
terms. From this we conclude that strong $\prs$-reachability implies
Böhm-reachability w.r.t.\ $\rAct$. Finally, we then show the other
direction of the equality.

\subsection{From Strong \texorpdfstring{$\prs$}{p}-Convergence to Böhm-Convergence}
\label{sec:from-strong-prs}

For the first step we have to transform a strongly $\prs$-converging
reduction in to a Böhm-converging reduction w.r.t.\ the set of total
fragile terms, i.e.\ a strongly $\mrs$-converging reduction w.r.t.\
the corresponding Böhm extension $\calB$. Recall that, by
Theorem~\ref{thr:strongExt}, the only difference between strongly
$\prs$-converging reductions and strongly $\mrs$-converging reductions
is the ability of the former to produce $\bot$ subterms. This happens,
according to Lemma~\ref{lem:botLimRed}, precisely at volatile
positions.

We can, therefore, proceed as follows: Given a strongly
$\prs$-converging reduction we construct a Böhm-converging reduction
by removing reduction steps which cause the volatility of a position
in some open prefix of the reduction and then replacing them by a
\emph{single} $\to[\bot]$-step.

The intuition of this construction is illustrated in
Figure~\ref{fig:bohmDestrFig}. It shows a strongly $\prs$-converging
reduction of length $\omega\mult4$ from $s$ to $t$. In order to
maintain readability, we restrict the attention to a particular branch
of the term (tree) as indicated in Figure~\ref{fig:bohmDestrFig1}. The
picture shows five positions which are volatile in some open prefix of
the reduction. We assume that they are the only volatile positions at
least in the considered branch. Note that the positions do not need to
occur in all of the terms in the reduction. They might disappear and
reappear repeatedly. Each of them, however, appears in infinitely many
terms in the reduction, as, by definition of volatility, infinitely
many steps take place at each of these positions. In
Figure~\ref{fig:bohmDestrFig2}, the prefixes of the reduction that
contain a volatile position are indicated by a waved rewrite arrow
pointing to a $\bot$. The level of an arrow indicates the position
which is volatile. A prefix might have multiple volatile
positions. For example, both $\pi_2$ and $\pi_4$ are volatile in the
prefix of length $\omega$. But a position might also be volatile for
several prefixes.  For instance, $\pi_3$ is volatile in the prefix of
length $\omega\mult2$ and the prefix of length $\omega\mult4$.

\begin{figure}
  \centering%
  \subfloat[Nested volatile positions.] {%
    \begin{tikzpicture}[%
      tree triangle/.style={%
        isosceles triangle,%
        isosceles triangle apex angle=45,%
        inner sep=1pt,%
        shape border rotate=90,%
        minimum size=4cm%
      }%
      ]
      \coordinate (a) at (0,0);%
      \coordinate (b) at (-2,-5);%
      \coordinate (c) at (2,-5);%
      \path[fill=termback,decoration={random
        steps,segment length=2mm}] (a) -- (b) decorate{-- (c)} --
      cycle;% 
      \path[draw=termfringe] (b) -- (a) -- (c);

      \coordinate (z) at ($(b)!.3!(c)$);
      
      \draw (a) .. controls ($(a)!.3!(z)+(.3,0)$) and
      ($(z)!.3!(a)+(-.3,0)$) .. (z)%
      \foreach \p/\r in%
      {1/.1,2/.3,3/.5,4/.7,5/.9} {%
        coordinate[pos=\r] (p\p)%
      };%
      \foreach \p/\l in%
      {1/$\pi_1$,2/$\pi_2$,3/$\pi_3$,4/$\pi_4$,5/$\pi_5$} {%
        \draw (p\p) -- ($(p\p)+(3pt,0)$);%
        \node at ($(p\p)+(10pt,0)$) {\l};%
      }%
      \draw[dotted] (z) -- ($(z)-(0,.3)$);%
    \end{tikzpicture}
    \label{fig:bohmDestrFig1}
  }\quad \subfloat[Replacing nested destructive reductions with
    $\oldTo_\bot$ steps.] {%
    \begin{tikzpicture}[%
      omega/.style={%
        shading = axis,%
        right color=black!30,%
        left color=white%,%
        % middle color=black!10%
      },%
      nested red/.style={%
        strongly,%
        decorate,%
        decoration={%
          snake,%
          segment length=20pt,%
          amplitude=2pt%
        }%
      }%
      ]
      \node (s) at (0,0) {$s$};%
      \node (t) at (8,0) {$t$};%
      \path[semithick, strongly] (s) edge (t);%

      \coordinate (a0) at ($(s)-(0,1)$);%
      \coordinate (z0) at ($(a0)-(0,5)$);%
      \coordinate (ae) at ($(a0)+(8,0)$);%
      \coordinate (ze) at ($(z0)+(8,0)$);%
      \foreach \d/\r in {b/.1,c/.3,d/.5,e/.7,f/.9} {%
        \foreach \n in {0,e} {%
          \coordinate (\d \n) at ($(a\n)!\r!(z\n)$);%
        }%
      }%
      \foreach \n / \r in {1/.25,2/.5,3/.75} {%
        \foreach \d in {a,b,c,d,e,f,z} {%
          \coordinate (\d\n) at ($(\d 0)!\r!(\d e)$);%
        }%
      }%
      \foreach \na / \nb in {0/1,1/2,2/3,3/e} {%
        \path[omega,decoration={random steps,segment length=2mm}]
        (a\na) -- (a\nb) -- (z\nb) decorate{-- (z\na)} -- cycle;%
      }%
      \draw[-latex] (a0) -- ($(z0)-(0,.4)$);%
      \draw[-latex] (a0) -- ($(ae)+(.2,0)$);%

      \foreach \d/\l in {%
        a/$\emptyseq$,b/$\pi_1$,c/$\pi_2$,%
        d/$\pi_3$,e/$\pi_4$,f/$\pi_5$} {%
        \draw (\d 0) -- ($(\d 0)-(2pt,0)$);%
        \node at ($(\d 0)-(10pt,0)$) {\l};%
      }%

      \foreach \n/\l in%
      {0/$0$,1/$\omega$,2/$\omega\mult 2$,3/$\omega\mult 3$,e/$\omega
        \mult 4$} {%
        \draw (a\n) -- ($(a\n)+(0,2pt)$);%
        \node at ($(a\n)+(0,10pt)$) {\l};%
      }%

      \foreach \n / \m / \r in %
      {b3/b1/.8,c1/c0/.3,d2/d1/.4,de/d3/.5} {%
        \coordinate (l) at ($(\n)!\r!(\m)$);
        \coordinate (k) at ($(\n)-(.05,0)$);
        \draw[semithick,dashed] ($(l)+(0,.2)$) rectangle
        ($(k)-(0,.2)$);%
        \draw[semithick,dotted]%
        (l) -- ($(z0)!(l)!(ze)$)%
        (k) -- ($(z0)!(k)!(ze)$);%
      }%

      \foreach \n in {b3, c1, d2, de, e1, f2} {%
        \node (\n p) at ($(\n)+(.1,0)$) {$\bot$};%
      }%

      \foreach \n in {b,c,d,e,f} {%
        \coordinate (\n 0p) at ($(\n 0)+(.1,0)$);%
      }%

      \draw[every edge/.style={nested red,draw}]%
      (b0p) edge (b3p)%
      (c0p) edge (c1p)%
      (d0p) edge (d2p)%
      (d2p) edge (dep)%
      (e0p) edge (e1p)%
      (f0p) edge (f2p);%
    
    \end{tikzpicture}
    \label{fig:bohmDestrFig2}
  }
  \caption{Turning a $\prs$-converging reduction into a Böhm-converging reduction.}
  \label{fig:bohmDestrFig}
\end{figure}
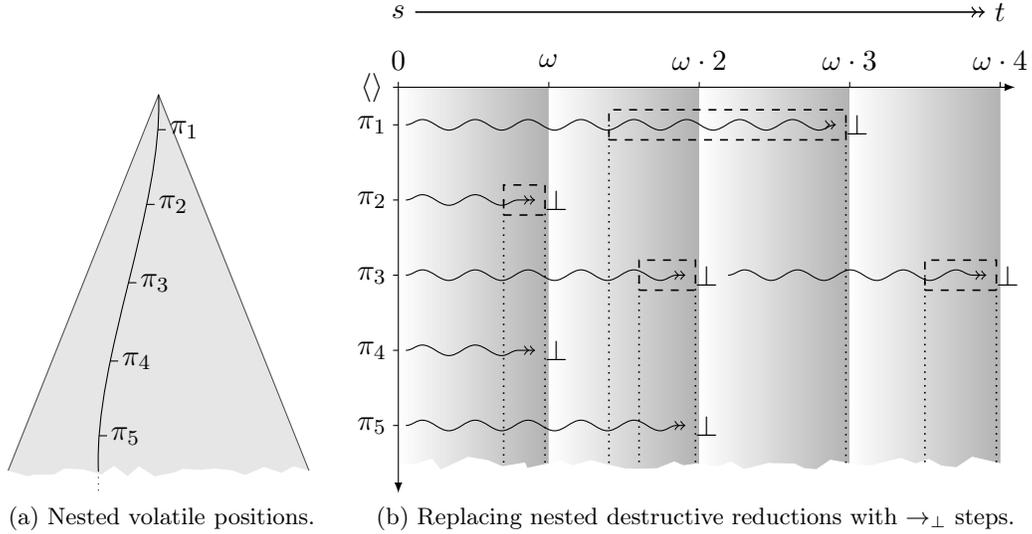

%%% Local Variables: 
%%% mode: latex
%%% TeX-master: "paper"
%%% End: 

By Lemma~\ref{lem:botLimRed}, outermost-volatile positions are
responsible for the generation of $\bot$ subterms. By their nature, at
some point there are no reductions taking place above
outermost-volatile positions. The suffix where this is the case is a
\emph{nested} destructive reduction. The subterm where this suffix
starts is, therefore, a fragile term and we can replace this suffix
with a \emph{single} $\to[\bot]$-step. The segments which are replaced
in this way are highlighted by dashed boxes in
Figure~\ref{fig:bohmDestrFig2}. As indicated by the dotted lines, this
then also includes reduction steps which occur below the
outermost-volatile positions. Therefore, also volatile positions which
are not outermost are removed as well. Eventually, we obtain a
reduction without volatile positions, which is, by
Lemma~\ref{lem:totalRed}, a strongly $\mrs$-converging reduction in
the Böhm extension, i.e.\ a Böhm-converging reduction in the original
system:
\begin{prop}[strong $\prs$-reachability implies Böhm-reachability]
  \label{prop:prsBohm}
  % PRS reductions are Böhm reductions %
  Let $\calR$ be a TRS, $\calU$ the set of fragile terms in $\iterms$,
  and $\calB$ the Böhm extension of $\calR$ w.r.t.\ $\calU$. Then, for
  each strongly $\prs$-convergent reduction $s \pato[\calR] t$, there is a
  Böhm-convergent reduction $s \mato[\calB] t$.
\end{prop}
\begin{proof}
  Assume that there is a reduction $S =(t_\iota \to[\pi_\iota]
  t_{\iota + 1})_{\iota < \alpha}$ in $\calR$ that strongly
  $\prs$-converges to $t_\alpha$. We will construct a strongly
  $\mrs$-convergent reduction $T\fcolon t_0 \mato[\calB] t_\alpha$ in
  $\calB$ by removing reduction steps in $S$ that take place at or
  below outermost-volatile positions of some prefix of $S$ and replace
  them by $\to[\bot]$-steps.

  Let $\pi$ be an outermost-volatile position of some prefix
  $\prefix{S}{\lambda}$. Then there is some ordinal $\beta < \lambda$
  such that no reduction step between $\beta$ and $\lambda$ in $S$
  takes place strictly above $\pi$, i.e.\ $\pi_\iota \not < \pi$ for
  all $\beta \le \iota < \lambda$. Such an ordinal $\beta$ must exist
  since otherwise $\pi$ would not be an outermost-volatile position in
  $\prefix{S}{\lambda}$. Hence, we can construct a destructive
  reduction $S'\fcolon \atPos{t_\beta}{\pi} \pato \bot$ by taking the
  subsequence of the segment $\segm{S}{\beta}{\lambda}$ that contains
  the reduction steps at $\pi$ or below. Note that
  $\atPos{t_\beta}{\pi}$ might still contain the symbol $\bot$. Since
  $\bot$ is not relevant for the applicability of rules in $\calR$,
  each of the $\bot$ symbols in $\atPos{t_\beta}{\pi}$ can be safely
  replaced by arbitrary total terms, in particular by terms in
  $\calU$. Let $r$ be a term that is obtained in this way. Then there
  is a destructive reduction $S''\fcolon r \pato \bot$ that applies
  the same rules at the same positions as in $S'$. Hence, $r \in
  \calU$. By construction, $r$ is a $\bot,\calU$-instance of
  $\atPos{t_\beta}{\pi}$ which means that $\atPos{t_\beta}{\pi} \in
  \calU_\bot$. Additionally, $\atPos{t_\beta}{\pi} \neq \bot$ since
  there is a non-empty reduction $S'\fcolon \atPos{t_\beta}{\pi} \pato
  \bot$ starting in $\atPos{t_\beta}{\pi}$. Consequently, there is a
  rule $\atPos{t_\beta}{\pi} \to \bot$ in $\calB$. Let $T'$ be the
  reduction that is obtained from $\prefix{S}{\lambda}$ by replacing
  the $\beta$-th step, which we can assume w.l.o.g.\ to take place at
  $\pi$, by a step with the rule $\atPos{t_\beta}{\pi} \to \bot$ at
  the same position $\pi$ and removing all reduction steps
  $\phi_\iota$ taking place at $\pi$ or below for all $\beta < \iota <
  \lambda$. Let $t'$ be the term that the reduction $T'$
  strongly $\prs$-converges to. $t_\lambda$ and $t'$ can only differ at
  position $\pi$ or below. However, by construction, we have $t'(\pi)
  = \bot$ and, by Lemma~\ref{lem:botLimRed}, $t_\lambda(\pi) =
  \bot$. Consequently, $t' = t_\lambda$.

  This construction can be performed for all prefixes of $S$ and their
  respective outermost-volatile positions. Thereby, we obtain a
  strongly $\prs$-converging reduction $T\fcolon t_0 \pato[\calB] t_\alpha$ for
  which no prefix has a volatile position. By
  Lemma~\ref{lem:totalRed}, $T$ is a total reduction. Note that
  $\calB$ is a TRS over the extended signature $\Sigma' = \Sigma
  \uplus \set{\bot}$, i.e.\ terms containing $\bot$ are considered
  total. Hence, by Theorem~\ref{thr:strongExt}, $T\fcolon t_0
  \mato[\calB] t_\alpha$.
\end{proof}

\subsection{From Böhm-convergence to Strong \texorpdfstring{$\prs$}{p}-Convergence}
\label{sec:from-bohm-conv}

Next, we establish a compression lemma for destructive reductions,
i.e.\ that each destructive reduction can be compressed to length
$\omega$. Before we continue with this, we need to mention the
following lemma from Kennaway et al.\ \cite{kennaway99jflp}:
\begin{lem}[postponement of $\protect{\to[\bot]}$-steps]
  \label{lem:procBot}
  % postponement of $\protect{\to[\bot]}$-steps %
  Let $\calR$ be a left-linear, left-finite TRS and $\calB$ some Böhm
  extension of $\calR$. Then $s \mato[\calB] t$
  implies $s \mato[\calR] s' \mato[\bot] t$ for some term
  $s'$.\footnote{Strictly speaking, if $s$ is not a total term,
    i.e.\ it contains $\bot$, then we have to consider the system that
    is obtained from $\calR$ by extending its signature to
    $\Sigma_\bot$.}
  % TODO: is the restriction to left-finite TRS necessary?
\end{lem}

In the next proposition we show that, excluding $\bot$ subterms, the
final term of a strongly $\prs$-converging reduction can be
approximated arbitrarily well by a finite reduction. This corresponds
to Corollary~\ref{cor:mrsFinApprox} which establishes finite
approximations for strongly $\mrs$-convergent reductions.
\begin{prop}[finite approximation]
  \label{prop:prsFinApprox}
  % finite approximation of non-$\bot$ occurrences %
  Let $\calR$ be a left-linear, left-finite TRS and $s \pato[\calR]
  t$. Then, for each finite set $P \subseteq \posNonBot{t}$, there is
  a reduction $s \fto{*}[\calR] t'$ such that $t$ and $t'$ coincide in
  $P$.
\end{prop}
\begin{proof}
  Assume that $s \pato[\calR] t$. Then, by
  Proposition~\ref{prop:prsBohm}, there is a reduction $s \mato[\calB]
  t$, where $\calB$ is the Böhm extension of $\calR$ w.r.t.\ the set
  of total, fragile terms of $\calR$. By Lemma~\ref{lem:procBot}, there
  is a reduction $s \mato[\calR] s' \mato[\bot] t$. Clearly, $s'$ and
  $t$ coincide in $\posNonBot{t}$. Let $d = \max\setcom{\len{\pi}}{\pi
    \in P}$. Since $P$ is finite, $d$ is well-defined. By
  Corollary~\ref{cor:mrsFinApprox}, there is a reduction $s
  \fto{*}[\calR] t'$ such that $t'$ and $s'$ coincide up to depth $d$
  and, thus, in particular they coincide in $P$. Consequently, since $s'$ and
  $t$ coincide in $\posNonBot{t} \supseteq P$, $t$ and
  $t'$ coincide in $P$, too.
\end{proof}
% left-finiteness is necessary: Consider a -> f(a), f^\omega -> b and
% reduction a -> f(a) -> ... -> f^\omega -> b

In order to establish a compression lemma for destructive reductions
we need that fragile terms are preserved by finite reductions. We can
obtain this from the following more general lemma showing that
destructive reductions are preserved by forming projections as
constructed in the Infinitary Strip Lemma:
\begin{lem}[preservation of destructive reductions by projections]
  \label{lem:presPerpRed}
  % preservation of destructive reductions %
  Let $\calR$ be an orthogonal TRS, $S\fcolon t_0 \pato t_\alpha$ a
  destructive reduction, and $T\fcolon t_0 \pato[U] s_0$ a complete
  development of a set $U$ of pairwise disjoint redex
  occurrences. Then the projection $\proj{S}{T}\fcolon s_0 \pato
  s_\alpha$ is also destructive.
\end{lem}
\begin{proof}
  We consider the situation depicted in
  Figure~\ref{fig:stripLem}. Since $S\fcolon t_0 \pato t_\alpha$ is
  destructive, we have, for each $\beta < \alpha$, some $\beta\le
  \gamma < \alpha$ such that $v_\gamma = \emptyseq$. If $v_\gamma =
  \emptyseq$, then also $\emptyseq \in \dEsc{v_\gamma}{U_\gamma}$ unless
  $\emptyseq \in U_\gamma$. As by Proposition~\ref{prop:disjDesc},
  $U_\gamma$ is a set of pairwise disjoint positions, $\emptyseq \in
  U_\gamma$ implies $U_\gamma = \set{\emptyseq}$. This means that if
  $v_\gamma = \emptyseq$ and $\emptyseq \in U_\gamma$, then $U_\iota =
  \emptyset$ for all $\gamma < \iota < \alpha$. Thus, there is only at
  most one $\gamma < \alpha$ with $\emptyseq \in U_\gamma$. Therefore,
  we have, for each $\beta < \alpha$, some $\beta\le \gamma < \alpha$
  such that $\emptyseq \in \dEsc{v_\gamma}{U_\gamma}$. Hence, $T$ is
  destructive.
\end{proof}

As a consequence of this preservation of destructiveness by forming
projections, we obtain that the set of fragile terms is closed under
finite reductions:
\begin{lem}[closure of fragile terms under finite reductions]
  \label{lem:perpFinRed}
  % closure of fragile terms under finite reductions %
  In each orthogonal TRS, the set of fragile terms is closed under
  finite reductions.
\end{lem}
\begin{proof}
  Let $t$ be a fragile term and $T\fcolon t \fto{*} t'$ a finite
  reduction. Hence, there is a destructive reduction starting in
  $t$. A straightforward induction proof on the length of $T$, using
  Lemma~\ref{lem:presPerpRed}, shows that there is a destructive
  reduction starting in $t'$. Thus, $t'$ is fragile.
\end{proof}

Now we can show that destructiveness does not need more that $\omega$
steps in orthogonal, left-finite TRSs. This property will be useful
for proving the equivalence of root-activeness and fragility of total
terms as well as the Compression Lemma for strongly $\prs$-convergent
reductions.
\begin{prop}[Compression Lemma for destructive reductions]
  \label{prop:comprPerpRed}
  % compression of destructive reductions %
  Let $\calR$ be an orthogonal, left-finite TRS and $t$ a partial term in
  $\calR$. If there is a destructive reduction starting in $t$, then
  there is a destructive reduction of length $\omega$ starting in $t$.
\end{prop}
\begin{proof}
  Let $S\fcolon t_0 \pto{\lambda} \bot$ be a destructive reduction
  starting in $t_0$. Hence, there is some $\alpha < \lambda$ such that
  $\prefix{S}{\alpha}\fcolon t_0 \pato s_1$, where $s_1$ is a
  $\rho$-redex for some $\rho\fcolon l \to r \in R$. Let $P$ be the
  set of pattern positions of the $\rho$-redex $s_1$, i.e.\ $P =
  \posFun{l}$. Due to the left-finiteness of $\calR$, $P$ is
  finite. Hence, by Proposition~\ref{prop:prsFinApprox}, there is a
  finite reduction $t_0 \fto{*} s'_1$ such that $s_1$ and $s'_1$
  coincide in $P$. Hence, because $\calR$ is left-linear, also $s'_1$
  is a $\rho$-redex. Now consider the reduction $T_0\fcolon t_0
  \fto{*} s'_1 \to[\rho,\emptyseq] t_1$ ending with a contraction at
  the root. $T_0$ is of finite length and, according to
  Lemma~\ref{lem:perpFinRed}, $t_1$ is fragile.

  Since $t_1$ is again fragile, the above argument can be iterated
  arbitrarily often which yields for each $i < \omega$ a finite
  non-empty reduction $T_i\fcolon t_i \fto{*} t_{i + 1}$ whose last
  step is a contraction at the root. Then the concatenation $T =
  \Concat_{i < \omega} T_i$ of these reductions is a destructive
  reduction of length $\omega$ starting in $t_0$.
\end{proof}
% left-finiteness is necessary: Consider a -> f(a), f^\omega -> b, b
% -> b and reduction a -> f(a) -> ... -> f^\omega -> b -> b ...

The above proposition bridges the gap between fragility and
root-activeness. Whereas the former concept is defined in terms of
transfinite reductions, the latter is defined in terms of finite
reductions. By Proposition~\ref{prop:comprPerpRed}, however, a fragile
term is always finitely reducible to a redex. This is the key to the
observation that fragility is not only quite similar to
root-activeness but is, in fact, essentially the same concept.
\begin{prop}[root-activeness = fragility]
  \label{prop:rootAct}
  % root-active terms %
  Let $\calR$ be an orthogonal, left-finite TRS and $t$ a total term
  in $\calR$. Then $t$ is root-active iff $t$ is fragile.
\end{prop}
\begin{proof}
  The ``only if'' direction is easy: If $t$ is root-active, then there
  is a reduction $S$ of length $\omega$ starting in $t$ with
  infinitely many steps taking place at the root. Hence, $S\fcolon t
  \pto{\omega} \bot$ is a destructive reduction, which
  makes $t$ a fragile term.

  For the converse direction we assume that $t$ is fragile and show
  that, for each reduction $t \fto{*} s$, there is a reduction $s
  \fto{*} t'$ to a redex $t'$. By Lemma~\ref{lem:perpFinRed}, also $s$
  is fragile. Hence, there is a destructive reduction $S\fcolon s
  \pato \bot$ starting in $s$. According to
  Proposition~\ref{prop:comprPerpRed}, we can assume that $S$ has
  length $\omega$. Therefore, there is some $n < \omega$ such that
  $\prefix{S}{n}\fcolon s \fto{*} t'$ for a redex $t'$.
\end{proof}
% left-finiteness is necessary: Consider a -> f(a), f^\omega -> b, b
% -> b and reduction a -> f(a) -> ... -> f^\omega -> b -> b ...; a is
% fragile but not root-active

To prove the other direction of the equality of strong
$\prs$-reachability and Böhm-reachability we need the property that
strongly $\mrs$-convergent reductions consisting only of
$\to[\bot]$-steps, i.e.\ contractions of $\rAct_\bot$-terms to $\bot$,
can be compressed to length at most $\omega$ as well. In order to show
this, we will make use of the following lemma from Kennaway et al.\
\cite{kennaway99jflp}:
\begin{lem}[$\bot,\rAct$-instances]
  \label{lem:botInst}
  % $\bot,\calU$-instances %
  Let $\rAct$ be the root-active terms of an orthogonal, left-finite
  TRS and $t\in \ipterms$. If some $\bot,\rAct$-instance of $t$ is in
  $\rAct$, then every $\bot,\rAct$-instance of $t$ is.
  % TODO: is the restriction to left-finite TRS necessary?
\end{lem}

\begin{lem}[compression of $\protect{\to[\bot]}$-steps]
  \label{lem:comprBot}
  % compression of $\protect{\to[\bot]}$ steps %
  Consider the Böhm extension of an orthogonal TRS w.r.t.\ its
  root-active terms and $S\fcolon s \mato[\bot] t$ with $s \in
  \iterms$, $t \in \ipterms$. Then there is a strongly
  $\mrs$-converging reduction $T\fcolon s \mato[\bot] t$ of length at
  most $\omega$ that is a complete development of a set of disjoint
  occurrences of root-active terms in $s$.
\end{lem}
\begin{proof}
  The proof is essentially the same as that of Lemma 7.2.4 from Ketema
  \cite{ketema06phd}.

    Let $S = (t_\iota \to[\pi_\iota] t_{\iota +1})_{\iota < \alpha}$ be
  the mentioned reduction strongly $\mrs$-converging to $t_\alpha$, and let
  $\pi$ be a position at which some reduction step in $S$ takes
  place. That is, there is some $\beta$ such that $\pi_\beta =
  \pi$. We will prove by induction on $\beta$ that $\atPos{t_0}{\pi}
  \in \rAct$.

  Consider the term $\atPos{t_\beta}{\pi}$. Since a $\to[\bot]$-rule
  is applied here, we have, according to Remark~\ref{rem:cloSub}, that
  $\atPos{t_\beta}{\pi} \in \rAct_\bot$. Let $V =
  \posBot{\atPos{t_\beta}{\pi}}$. Hence, for each $v \in V$, there is
  some $\gamma < \beta$ such that $\pi_\gamma = \pi \concat
  v$. Therefore, we can apply the induction hypothesis and get that
  $\atPos{t_0}{\pi\concat v} \in \rAct$ for all $v \in V$. It is clear
  that we can obtain $\atPos{t_0}{\pi}$ from $\atPos{t_\beta}{\pi}$ by
  replacing each $\bot$-occurrence at $v \in V$ with the corresponding
  term $\atPos{t_0}{\pi\concat v}$. That is, $\atPos{t_0}{\pi}$ is a
  $\bot,\rAct$-instance of $\atPos{t_\beta}{\pi}$. Because
  $\atPos{t_\beta}{\pi} \in \rAct_\bot$, there is some
  $\bot,\rAct$-instance of $\atPos{t_\beta}{\pi}$ in $\rAct$. Thus, by
  Lemma~\ref{lem:botInst}, also $\atPos{t_0}{\pi}$ is in $\rAct$. This
  closes the proof of the claim.

  Now let $V = \posBot{t_\alpha}$. Clearly, all positions in $V$ are
  pairwise disjoint. Moreover, for each $v \in V$, there is a step in
  $S$ that takes place at $v$. Hence, by the claim shown above, $V$ is
  a set of occurrences in $t_0$ of terms in $\rAct$. A complete
  development of $V$ in $t_0$ leads to $t_\alpha$ and can be performed
  in at most $\omega$ steps by an outermost reduction strategy.
\end{proof}

The important part of the above lemma is the statement that only terms
in $\rAct$ are contracted instead of the general case where a
$\to[\bot]$ -step contracts a term in $\rAct_\bot\supset \rAct$.

Finally, we have gathered all tools necessary in order to prove the
converse direction of the equivalence of strong $\prs$-reachability and
Böhm-reachability w.r.t.\ root-active terms.
\begin{thm}[strong $\prs$-reachability = Böhm-reachability
  w.r.t.\ $\rAct$]
  \label{thr:prsEqBohm}
  % PRS reductions = Böhm reductions %
  Let $\calR$ be an orthogonal, left-finite TRS and $\calB$ the Böhm extension of
  $\calR$ w.r.t.\ its root-active terms. Then $s \pato[\calR] t$ iff
  $s \mato[\calB] t$.
\end{thm}
\begin{proof}
  The ``only if'' direction follows immediately from
  Proposition~\ref{prop:rootAct} and Proposition~\ref{prop:prsBohm}.

  Now consider the converse direction: Let $s \mato[\calB] t$ be a strongly
  $\mrs$-convergent reduction in $\calB$. W.l.o.g.\ we assume $s$ to
  be total. Due to Lemma~\ref{lem:procBot}, there is a term $s' \in
  \iterms$ such that there are strongly $\mrs$-convergent reductions $S\fcolon
  s \mato[\calR] s'$ and $T\fcolon s' \mato[\bot] t$. By
  Lemma~\ref{lem:comprBot}, we can assume that in $s' \mato[\bot] t$
  only pairwise disjoint occurrences of root-active terms are
  contracted. By Proposition~\ref{prop:rootAct}, each root-active term
  $r\in \rAct$ is fragile, i.e.\ we have a destructive reduction $r
  \pato[\calR] \bot$ starting in $r$. Thus, following
  Remark~\ref{rem:cloSub}, we can construct a strongly $\prs$-converging
  reduction $T'\fcolon s' \pato[\calR] t$ by replacing each step
  $\Cxt{b}{r}\to[\bot] \Cxt{b}{\bot}$ in $T$ with the corresponding
  reduction $\Cxt{b}{r} \pato[\calR] \Cxt{b}{\bot}$. By combining $T'$
  with the strongly $\mrs$-converging reduction $S$, which, according to
  Theorem~\ref{thr:strongExt}, is also strongly $\prs$-converging, we obtain
  the strongly $\prs$-converging reduction $S\concat T'\fcolon s \pato[\calR]
  t$.
\end{proof}
% (TODO:) is the restriction to left-finiteness necessary?

\subsection{Corollaries}
\label{sec:corollaries}

With the equivalence of strong $\prs$-reachability and
Böhm-reachability established in the previous section, strongly
$\prs$-convergent reductions inherit a number of important properties
that are enjoyed by Böhm-convergent reductions:
\begin{thm}[infinitary confluence]
  \label{thr:prsCR}
  % infinitary confluence %
  Every orthogonal, left-finite TRS is infinitarily confluent. That
  is, for each orthogonal, left-finite TRS, $s_1 \pafrom t \pato s_2$
  implies $s_1 \pato t' \pafrom s_2$.
\end{thm}
\begin{proof}
  Leveraging Theorem~\ref{thr:prsEqBohm}, this theorem follows from
  Theorem~\ref{thr:bohmCR}.
\end{proof}

Returning to Example~\ref{ex:mconfl} again, we can see that, in the
setting of strongly $\prs$-converging reduction, the terms $g^\omega$
and $f^\omega$ can now be joined by repeatedly contracting the redex
at the root which yields two destructive reductions $g^\omega \pato
\bot$ and $f^\omega \pato \bot$, respectively.

\begin{thm}[infinitary normalisation]
  \label{thr:prsWN}
  % infinitary normalisation %
  Every orthogonal, left-finite TRS is infinitarily normalising. That
  is, for each orthogonal, left-finite TRS $\calR$ and a partial term
  $t$ in $\calR$, there is an $\calR$-normal form strongly
  $\prs$-reachable from $t$.
\end{thm}
\begin{proof}
  This follows immediately from Theorem~\ref{thr:prsEqBohm} and
  Theorem~\ref{thr:bohmWn}.
\end{proof}

Combining Theorem~\ref{thr:prsCR} and Theorem~\ref{thr:prsWN}, we
obtain that each term in an orthogonal TRS has a unique normal form
w.r.t.\ strong $\prs$-convergence. Due to Theorem~\ref{thr:prsEqBohm},
this unique normal form is the Böhm tree w.r.t.\ root-active terms.

Since strongly $\prs$-converging reductions in orthogonal TRS can
always be transformed such that they consist of a prefix which is a
strongly $\mrs$-convergent reduction and a suffix consisting of nested
destructive reductions, we can employ the Compression Lemma for strongly
$\mrs$-convergent reductions (Theorem~\ref{thr:mrsCompr}) and the
Compression Lemma for destructive reductions
(Proposition~\ref{prop:comprPerpRed}) to obtain the Compression Lemma
for strongly $\prs$-convergent reductions:
\begin{thm}[Compression Lemma for strongly $\prs$-convergent reductions]
  \label{thr:prsCompr}
  For each orthogonal, left-finite TRS, $s \pato t$ implies $s
  \pto{\le \omega} t$.
\end{thm}
\begin{proof}
  Let $s \pato[\calR] t$. According to Theorem~\ref{thr:prsEqBohm}, we
  have $s \mato[\calB] t$ for the Böhm extension $\calB$ of $\calR$
  w.r.t.\ $\rAct$ and, therefore, by Lemma~\ref{lem:procBot}, we have
  reductions $S\fcolon s \mato[\calR] s'$ and $T\fcolon s'\mato[\bot]
  t$. Due to Theorem~\ref{thr:mrsCompr}, we can assume $S$ to be of
  length at most $\omega$ and, due to Theorem~\ref{thr:strongExt}, to
  be strongly $\prs$-convergent, i.e $S\fcolon s \pto{\le
    \omega}[\calR] s'$. If $T$ is the empty reduction, then we are
  done. If not, then $T$ is a complete development of pairwise
  disjoint occurrences of root-active terms according to
  Lemma~\ref{lem:comprBot}. Hence, each step is of the form
  $\Cxt{b}{r}\to[\bot] \Cxt{b}{\bot}$ for some root-active term
  $r$. By Proposition~\ref{prop:rootAct}, for each such term $r$,
  there is a destructive reduction $r \pato[\calR] \bot$ which we can
  assume, in accordance with Proposition~\ref{prop:comprPerpRed}, to
  be of length $\omega$. Hence, each step $\Cxt{b}{r}\to[\bot]
  \Cxt{b}{\bot}$ can be replaced by the reduction $\Cxt{b}{r}
  \pto{\omega}[\calR] \Cxt{b}{\bot}$. Concatenating these reductions
  results in a reduction $T'\fcolon s'\pato[\calR] t$ of length at
  most $\omega\cdot\omega$. If $S\fcolon s \pto{\le\omega}[\calR] s'$
  is of finite length, we can interleave the reduction steps in $T'$
  such that we obtain a reduction $T''\fcolon s'\pto{\omega}[\calR] t$
  of length $\omega$. Then we have $S\concat T''\fcolon s
  \pto{\omega}[\calR] t$. If $S\fcolon s \pto{\le\omega}[\calR] s'$
  has length $\omega$, we construct a reduction $s \pato[\calR] t$ as
  follows: As illustrated above, $T'$ consists of destructive
  reductions taking place at some pairwise disjoint positions. These
  steps can be interleaved into the reduction $S$ resulting into a
  reduction $s \pato[\calR] t$ of length $\omega$. The argument for
  that is similar to that employed in the successor case of the
  induction proof of the Compression Lemma of Kennaway et al.\
  \cite{kennaway95ic}.

  % \todo{How is it exactly similar to the original Compression Lemma
  % proof?}
\end{proof}

We do not know whether full orthogonality is essential for the
Compression Lemma. However, as for strongly $\mrs$-convergent
reductions, the left-linearity part of it is:
\begin{exa}[\cite{kennaway95ic}]
  Consider the TRS consisting of the rules $f(x,x) \to c, a \to g(a),
  b \to g(b)$. Then there is a strongly $\prs$-converging reduction
  \[
  f(a,b) \to f(g(a),b) \to f(g(a),g(b)) \to f(g(g(a)),g(b)) \to \dots
  \; f(g^\omega,g^\omega) \to c
  \]
  of length $\omega+1$.  However, there is no strongly $\prs$-converging
  reduction $f(a,b) \pto{\le\omega} c$ (since there is no such
  strongly $\mrs$-converging reduction).
\end{exa}

We can use the Compression Lemma for strongly $\prs$-convergent reductions to
obtain a stronger variant of Theorem~\ref{thr:strongExt} for
orthogonal TRSs:
\begin{cor}[strong $\mrs$-reachability = strong
  $\prs$-reachability of total terms]
  \label{cor:prsMrsEq}
  Let $\calR$ be an orthogonal, left-finite TRS and $s,t \in \iterms$. Then
  $s\mato t$ iff $s\pato t$.
\end{cor}
\begin{proof}
  The ``only if'' direction follows immediately from
  Theorem~\ref{thr:strongExt}. For the ``if'' direction assume a
  reduction $S\fcolon s \pato t$. According to
  Theorem~\ref{thr:prsCompr}, there is a reduction $T\fcolon s
  \pto{\le\omega} t$. Hence, since $s$ is total and totality is
  preserved by single reduction steps, $T\fcolon s \pto{\le\omega} t$
  is total. Applying Theorem~\ref{thr:strongExt}, yields that
  $T\fcolon s \mto{\le\omega} t$.
\end{proof}

Notice the similarity of the above corollary with the Compression
Lemma. The Compression Lemma states that the reachability relation
$\mato$ (as well as $\pato$) is the same whether we consider ordinals
beyond $\omega$ or not. Analogously, Corollary~\ref{cor:prsMrsEq}
states that the reachability relation $\pato$ on total terms is the
same whether we allow partial convergence in between or not. More apt,
however, is the comparison to the following corollary of the
Compression Lemma (cf.\ Corollary~\ref{cor:mrsFinApprox}): $s \mato t$
implies $s \fto* t$ whenever $t$ is a finite term. In other words, if
the final term is finite, we only need finite reductions. Analogously,
Corollary~\ref{cor:prsMrsEq} states that if the initial term and the
final term are total, we only need metric convergence.

\section{Conclusions}
\label{sec:conclusions}

Infinitary term rewriting in the partial order model provides a more
fine-grained notion of convergence. Formally, every meaningful, i.e.\
$\prs$-continuous, reduction is also $\prs$-converging. However,
$\prs$-converging reductions can end in a term containing '$\bot$'s
indicating positions of local divergence. Theorem~\ref{thr:weakExt},
Theorem~\ref{thr:strongExt} and Corollary~\ref{cor:prsMrsEq} show that
the partial order model coincides with the metric model but
additionally allows a more detailed inspection of
non-$\mrs$-converging reductions. Instead of the coarse discrimination
between convergence and divergence provided by the metric model, the
partial order model allows different levels between full convergence
(a total term as result) and full divergence ($\bot$ as result).

The equivalence of strong $\prs$-reachability and Böhm-reachability
shows that the differences between the metric and the partial order
model can be compensated by simply adding rules that allow to
replicate destructive reductions by $\to[\bot]$-steps. By this
equivalence, we additionally obtain infinitary normalisation and
infinitary confluence for orthogonal systems -- a considerable
improvement over strong $\mrs$-convergence. Both strong
$\prs$-convergence and Böhm-convergence are defined quite differently
and have independently justified intentions, yet they still induce the
same notion of transfinite reachability. This suggests that this
notion of transfinite reachability can be considered a ``natural''
choice -- also because of its properties that admit unique normal
forms.  Nevertheless, while achieving the same goals as
Böhm-extensions, the partial order approach provides a more intuitive
and more elegant model for transfinite reductions as it does not need
the cumbersomely defined ``shortcuts'' provided by $\to[\bot]$-steps,
which depend on allowing infinite left-hand sides in rewrite
rules. Vice versa destructive reductions in the partial order model
provide a justification for admitting these shortcuts.

\subsubsection*{Related Work}
\label{sec:related-work}

This study of partial order convergence is inspired by Blom
\cite{blom04rta} who investigated strong partial order convergence in
lambda calculus and compared it to strong metric
convergence. Similarly to our findings for orthogonal term rewriting
systems, Blom has shown for lambda calculus that reachability in the
metric model coincides with reachability in the partial order model
modulo equating so-called $0$-undefined terms.

Also Corradini \cite{corradini93tapsoft} studied a partial order
model. However, he uses it to develop a theory of parallel reductions
which allows simultaneous contraction of a set of mutually independent
redexes of left-linear rules. To this end, Corradini defines the
semantics of redex contraction in a non-standard way by allowing a
partial matching of left-hand sides. Our definition of complete
developments also provides, at least for orthogonal systems, a notion
of parallel reductions but does so using the standard semantics of
redex contraction.

\subsubsection*{Future Work}
\label{sec:future-work}

While we have studied both weak and strong $\prs$-convergence and have
compared it to the respective metric counterparts, we have put the
focus on strong $\prs$-convergence. It would be interesting to find
out whether the shift to the partial order model has similar benefits
for weak convergence, which is known to be rather unruly in the metric
model \cite{simonsen04ipl}. A starting point in this direction would
be to find correspondences between weak and strong
$\prs$-convergence. For example, in the metric setting we have that $s
\mwato[\calR] t$ implies that there is some $t'$ with $s \mato[\calB]
t'$ and $t \mato[\calB] t'$ \cite[Theorem~12.9.14]{kennaway03book}. If
we had the analogous correspondence for $\prs$-convergence, we would
immediately obtain infinitary normalisation and confluence for weak
$\prs$-convergence.

Moreover, we have focused on orthogonal systems in this paper. It
should be easy to generalise our results to almost orthogonal
systems. The only difficulty is to deal with the ambiguity of paths
when rules are allowed to overlay. This could be resolved by
considering equivalence classes of paths instead. The move to weakly
orthogonal systems is much more complicated: For strong
$\mrs$-convergence Endrullis et al.\ \cite{endrullis10rta} have shown
that weakly orthogonal systems do not even satisfy the infinitary
unique normal form property (\iUN{}), a property that orthogonal
systems do enjoy \cite{kennaway95ic}. Due to
Theorem~\ref{thr:strongExt}, this means that also in the setting of
strong $\prs$-convergence, weakly orthogonal systems do not satisfy
\iUN{} and are therefore not infinitarily confluent either! Endrullis
et al.\ \cite{endrullis10rta} have shown that this can be resolved in
the metric setting by prohibiting collapsing rules. However, it is not
clear whether this result can be transferred to the partial order
setting.

Another interesting direction to follow is the ability to finitely
simulate transfinite reductions by term graph rewriting. For strong
$\mrs$-convergence this is possible, at least to some extent
\cite{kennaway94toplas}. We think that a different approach to term
graph rewriting, viz.\ the \emph{double-pushout approach}
\cite{ehrig73swat} or the \emph{equational approach}
\cite{ariola96fi}, is more appropriate for the present setting of
$\prs$-convergence \cite{corradini97rep,bahr09master}.

\section*{Acknowledgements}
\label{sec:acknowledgements}
I am indebted to Bernhard Gramlich for his constant support during the
work on my master's thesis which made this work possible.

\bibliographystyle{plain}
\bibliography{po-inf-rew}

\end{document}